\documentclass[12pt]{article}
\usepackage{amssymb,amsmath,amsthm,bbm}
\usepackage{setspace,appendix,pdflscape,afterpage}
\usepackage{xcolor}
\usepackage[nohead,margin=1.25in]{geometry}
\usepackage[flushleft]{threeparttable}
\usepackage{array,booktabs,enumerate}
\usepackage[section]{placeins}
\newtheorem{theorem}{Theorem}
\newtheorem{lemma}[theorem]{Lemma}
\newtheorem{algorithm}{Algorithm}
\theoremstyle{definition}
\newtheorem*{example*}{Example}
\newtheorem{remark}{Remark}
\usepackage{hyperref}
\usepackage[longnamesfirst,authoryear]{natbib}
\usepackage{ifthen}
\hypersetup{colorlinks=true,anchorcolor=blue,citecolor=black,linkcolor=[rgb]{0,0.0,0.75}}
\newtheoremstyle{namedAssumption}{}{}{\itshape}{}{\bfseries}{.}{.5em}{#1 \ifthenelse{\equal{#3}{}}{{#2}}{{#3}}}
\theoremstyle{namedAssumption}
\newtheorem{assumption}{Assumption}
\numberwithin{assumption}{section}

\makeatletter
\newcommand{\namedlabel}[2]{\begingroup\def\@currentlabel{#2} \label{#1}\endgroup}
\makeatother
\newcommand{\pr}{\mathbb{P}}
\newcommand{\ind}{\mathbbm{1}}
\newcommand{\e}{\mathbb{E}}
\newcommand{\prob}{\overset{p}{\rightarrow}}
\newcommand{\dist}{\overset{d}{\rightarrow}}
\newcommand{\sumin}{\sum_{i=1}^n}
\renewcommand{\qed}{\hfill Q.E.D.}
\DeclareMathOperator*{\argmax}{\mathrm{arg}\!\max\limits}
\DeclareMathOperator*{\argmin}{\mathrm{arg}\!\min\limits}
\newcommand*\norm[1]{\left\Vert #1\right\Vert}
\newcommand{\abs}[1]{\left\vert #1\right\vert}
\newcommand{\ex}[1]{\e\left[#1\right]}
\newcommand{\inter}[1]{\text{int}\left(#1\right)}
\newcommand{\supp}[1]{\text{supp}\left(#1\right)}
\newcommand*{\QEDA}{\hfill\ensuremath{\blacksquare}}

\newcommand{\paperGreek}[1]{%
  {\setbox0=\hbox{$\mathbf{x\global\edef\boldMathFamily{\the\mathgroup}}$}}%
  \ifnum\mathgroup=\boldMathFamily
    \mathchoice
      {\mbox{\boldmath$\displaystyle\mathchar"#1$}}%
      {\mbox{\boldmath$\textstyle\mathchar"#1$}}%
      {\mbox{\boldmath$\scriptstyle\mathchar"#1$}}%
      {\mbox{\boldmath$\scriptscriptstyle\mathchar"#1$}}%
  \else
    \mathchar"#1\relax
  \fi}
\renewcommand{\alpha}{{\paperGreek{010B}}}
\renewcommand{\beta}{{\paperGreek{010C}}}
\renewcommand{\gamma}{{\paperGreek{010D}}}
\renewcommand{\delta}{{\paperGreek{010E}}}
\renewcommand{\epsilon}{{\paperGreek{010F}}}
\renewcommand{\zeta}{{\paperGreek{0110}}}
\renewcommand{\eta}{{\paperGreek{0111}}}
\renewcommand{\theta}{{\paperGreek{0112}}}
\renewcommand{\kappa}{{\paperGreek{0114}}}
\renewcommand{\lambda}{{\paperGreek{0115}}}
\renewcommand{\mu}{{\paperGreek{0116}}}
\renewcommand{\nu}{{\paperGreek{0117}}}
\renewcommand{\xi}{{\paperGreek{0118}}}
\renewcommand{\pi}{{\paperGreek{0119}}}
\renewcommand{\rho}{{\paperGreek{011A}}}
\renewcommand{\sigma}{{\paperGreek{011B}}}
\renewcommand{\tau}{{\paperGreek{011C}}}
\renewcommand{\psi}{{\paperGreek{0120}}}
\renewcommand{\varepsilon}{{\paperGreek{0122}}}
\renewcommand{\varrho}{{\paperGreek{0125}}}
\renewcommand{\varphi}{{\paperGreek{0127}}}
\renewcommand{\nabla}{{\paperGreek{0272}}}

\begin{document}

\title{\vspace*{-1in}\textsc{Simple Estimation of Semiparametric Models with
Measurement Errors}}
\author{Kirill S. \textsc{Evdokimov}\thanks{Universitat Pompeu Fabra and Barcelona School of Economics:
\textsf{kirill.evdokimov@upf.edu}.} \and Andrei \textsc{Zeleneev}\thanks{University College London:
\textsf{a.zeleneev@ucl.ac.uk}.} \thanks{First version: November 7, 2016. A part of the material of this paper was
previously circulated as a part of Evdokimov and Zeleneev (2018).} \thanks{We thank the participants of the
numerous seminars and conferences for
helpful comments and suggestions. We are also grateful to the Gregory C.
Chow Econometrics Research Program at Princeton University and the
Department of Economics at the Massachusetts Institute of Technology for
their hospitality and support. Evdokimov also gratefully acknowledges the
support from the National Science Foundation via grant SES-1459993, and from
the Spanish MCIN/AEI via grants RYC2020-030623-I, PID2019-107352GB-I00, PID2022-140825NB-I00, and Severo Ochoa
Programme CEX2024-001476-S, funded by MICIU/AEI/10.13039/501100011033.}}
\date{This version: November 6, 2026}
\maketitle

\begin{abstract}
We develop a practical way of addressing the Errors-In-Variables (EIV)
problem in the Generalized Method of Moments (GMM) framework. We focus on
the settings in which the variability of the EIV is a fraction of that of
the mismeasured variables, which is typical for empirical applications. For
any initial set of moment conditions our approach provides a ``corrected''
set of moment conditions that are robust to the EIV. We show that the GMM
estimator based on these moments is $\sqrt{n}$-consistent, with the standard
tests and confidence intervals providing valid inference. This is true even
when the EIV are so large that naive estimators (that ignore the EIV
problem) are heavily biased with their confidence intervals having 0\%
coverage. Our approach involves no nonparametric estimation, which is
especially important for applications with many covariates and settings
with multivariate EIV. In particular, the approach makes it
easy to use instrumental variables to address EIV in nonlinear models.

\end{abstract}

\newpage

\section{Introduction}

Measurement errors are a common problem for empirical studies. Addressing
the Errors-In-Variables (EIV) bias in nonlinear models requires elaborate
strategies.\footnote{See \cite
{HINP1991JoE,HausmanNeweyPowell1995JoE,Newey2001REStat,Schennach2007Ecta,Li2002JoE,Schennach2004Ecta,ChenHongTamer2005ReStud,HuSchennach2008Ecta,Schennach2014Ecta-ELVIS,Wilhelm2019WP-TestingForME}, among others.} Despite the fundamental theoretical progress in
identification and estimation of nonlinear models with EIV, the problem of
EIV is still rarely addressed in empirical work outside of linear
specifications.

\setlength{\belowdisplayskip}{8.0pt plus 2.0pt minus 7.0pt}
\setlength{\abovedisplayskip}{6.0pt plus 2.0pt minus 5.0pt}

The goal of this paper is to develop a simple and practical approach to
estimation of nonlinear semiparametric models that can be expressed in the
form of general moment conditions
\begin{equation}
\mathbb{E}[g(X_{i}^{\ast },S_{i},\theta )]=0\text{ iff }\theta =\theta _{0}\text{,}  \label{eq:g descr moment}
\end{equation}
where $g\left( \cdot \right) $ is a vector of moment functions and $\theta
_{0}$ is the parameter vector of interest. The researcher has a random
sample of $\left\{ X_{i},S_{i}\right\} _{i=1}^{n}$, where scalar or vector $X_{i}$ is a mismeasured version of
unobserved $X_{i}^{\ast }$ with
measurement error $\varepsilon _{i}$:
\begin{equation*}
X_{i}=X_{i}^{\ast }+\varepsilon _{i}.
\end{equation*}
We will refer to $g\left( \cdot \right) $ as the \emph{original} moment function, since it
would have been valid had the researcher observed $X_{i}^{\ast }$. A naive
GMM estimator (that ignores the EIV and uses $X_{i}$ in place of $X_{i}^{\ast }$) based on $g\left( \cdot \right)
$ is biased because $\mathbb{E}[g(X_{i},S_{i},\theta _{0})]\neq 0$, in contrast to equation~(\ref{eq:g
descr moment}).

\begin{example*}[Nonlinear Regression, NLR]
\label{eg:NLR}Let $Y_{i}$ denote a scalar outcome, and let $X_{i}^{\ast }$
and $W_{i}$ be the covariates. Suppose
\begin{equation}
E\left[ Y_{i}|X_{i}^{\ast },W_{i}\right] =\rho \left( X_{i}^{\ast
},W_{i},\theta _{0}\right)  \label{eq:Intro:eg NLR}
\end{equation}
for some function $\rho $ known up to the parameter $\theta $. For example,
in the Logit model, $Y_{i}$ is binary, $\rho \left( x,w,\theta \right)
\equiv 1\left/ \left( 1+\exp \left( -\left( \theta _{x}^{\prime }x+\theta
_{w}^{\prime }w\right) \right) \right) \right. $, and $\theta \equiv \left(
\theta _{x}^{\prime },\theta _{w}^{\prime }\right) ^{\prime }$.

Suppose the researcher has an instrumental variable $Z_{i}$. Then, they can
use
\begin{equation*}
g\left( y,x,w,z;\theta \right) \equiv \left( y-\rho \left( x,w,\theta
\right) \right) \varphi \left( x,w,z\right)
\end{equation*}
as the original moment function, where $\varphi \left( x,w,z\right) $ is a
vector that, for example, can include $x$, $z$, $w$, their powers and/or
interactions.\footnote{Note that the moment condition~(\ref{eq:g descr moment}) is stated in terms
of the true (correctly measured) $X_{i}^{\ast }$. Determining what functions
$g\left( \cdot \right) $ (or $\varphi \left( \cdot \right) $ in the NLR model)
satisfy this moment condition does not involve any consideration of the
measurement errors and hence is straightforward.} Here $S_{i}=\left(
Y_{i},W_{i},Z_{i}\right) $. \QEDA
\end{example*}

Even in this well-studied example of nonlinear regression, estimation in the
presence of the EIV is a difficult problem. Importantly, the nonlinear
instrumental variable regression estimator cannot be used, since it is
inconsistent in the presence of EIV \citep{Amemiya1985JoE}. The existing
approaches typically require nonparametric estimation that can be
impractical in many empirical applications. In contrast, in this paper, we
develop an alternative class of estimators that are essentially GMM
estimators that modify the original moment functions $g\left( \cdot \right) $
in a way that makes the moment conditions robust to the EIV. In particular,
our approach makes it easy to use instrumental variables to address EIV\ in
nonlinear models.

To provide a practical estimation approach\ for the general class of models~(\ref{eq:g descr moment}), we focus on
empirical settings in which the
researcher believes the variability of the measurement error to be at most a
fraction of the variability of the mismeasured variable, i.e., the
noise-to-signal ratio $\tau \equiv \sigma _{\varepsilon }/\sigma _{X^{\ast
}} $ to be moderate. {}
The absolute magnitude of the measurement error $\sigma _{\varepsilon }$
does not need to be small. Existing validation studies provide insights into
the magnitude of $\tau $ for some key economic variables and datasets. \citet{BoundKrueger1991JoLaborEcon}
consider log-earnings in the Current
Population Survey (CPS) data matched to the Social Security payroll records.
Their estimates of the variance of the measurement errors correspond to $\tau $ of approximately $0.47$ and $0.30$
for the subsamples of men and
women, respectively. \citet{BoundBrownDuncanRodgers1994JoLaborEcon} consider
a validation study of the Panel Study of Income Dynamics (PSID). Their estimates
imply $\tau =0.39-0.66$ for log-earnings and $\tau =0.63-0.77$ for hours
worked. The estimates in \citet{Pischke1995JBES} correspond to $\tau =0.39-0.50$
for the log-earnings in PSID. \citet{AshenfelterKrueger1994AER} assess the
mismeasurement in the years of education; their estimates correspond to $\tau =0.29-0.37$.

Focusing on these settings allows us to isolate the most important aspects
of the problem and, as a result, to develop a simple estimator, which does not
require any nonparametric estimation or simulation. Such simple estimation
becomes possible because in these settings we can obtain a simple
approximation of the EIV\ bias of the moment conditions as a function of $\theta $.

We propose to bias correct the original moments $g\left( \cdot \right) $,
which in turn removes the bias of the corresponding estimator of $\theta
_{0} $. This bias correction depends on some moments of the distribution of
the measurement errors that are unknown. Another difficulty is that the
estimators of some components of the bias correction themselves may need to
be bias corrected. To address these issues, we develop the \emph{corrected
moment conditions}, which depend on $\theta $ and additional parameters $\gamma $ that govern the bias correction.
The true
parameter value $\gamma _{0}$ is associated with (possibly conditional)
low-order moments of $\varepsilon _{i}$. Despite some theoretical subtleties
with the construction of the corrected moment conditions, their practical
implementation is straightforward and they can be automatically computed for
any original moment function $g\left( \cdot \right) $.

We introduce the Measurement Error Robust Moments (MERM) estimator, which is
a GMM estimator that uses the corrected moment conditions to jointly
estimate parameters $\theta _{0}$ and $\gamma _{0}$. The estimator can be
computed using any standard software for GMM estimation. Joint estimation of
parameters $\theta _{0}$ and $\gamma _{0}$ using the corrected moment
conditions effectively robustifies moment conditions $g\left( \cdot \right) $
against the impact of the measurement errors.

To make these ideas precise and to study the properties of the proposed
estimators, we develop an asymptotic theory using a nonstandard asymptotic
approximation that models $\tau $ as slowly shrinking with the sample size.
Standard asymptotics considers $\tau $ to be constant, which implies that as
$n\rightarrow \infty $ the bias of a naive estimator dwarfs its sampling
variability: the bias is constant while the standard errors shrink
proportionally to $1/\sqrt{n}$. As a result, under the standard asymptotics,
the problem of removing the EIV bias becomes central in the analysis, with
relatively little attention paid to the sampling variability of estimators.
However, this focus does not seem to be appropriate in many empirical
applications, in which the researcher does not expect the potential EIV bias
to be several orders of magnitude larger than the standard errors.\footnote{Such empirical settings appear to be
widespread. Although the concerns about
measurement errors are often raised, the majority of applied work does not
explicitly correct the EIV bias in nonlinear models, and instead implicitly
or explicitly argues or conjectures that the EIV bias is likely not to be
too large.} By considering $\tau $ as drifting towards zero with the sample size, our
approach provides better guidance on construction of EIV robust estimators
with good finite sample properties when $\tau $ is small or moderate.\footnote{Nonstandard asymptotic
approximations with drifting parameters are often
used to obtain better approximations of the finite sample behavior of
estimators and tests. For example, in the instrumental variable regression
settings, to consider the settings with relatively small first stage
coefficients, \cite{StaigerStock1997} model them as shrinking with $n$. It
is important to keep in mind that such nonstandard asymptotic approximations
are merely mathematical tools. One should not take them literally and think
of parameters somehow changing if more data is collected. }

Using this approximation, we show that the proposed estimation approach
indeed addresses the EIV problem. The MERM estimator is shown to be $\sqrt{n}
$-consistent and asymptotically normal and unbiased. The standard confidence
intervals and tests for GMM estimators are also valid for the MERM\
estimator. Additionally, the standard GMM arsenal of assessment tools can be
applied to the MERM estimator, allowing one to test model identification,
conduct valid inference, and perform model specification diagnostics.

The usefulness of a large sample theory is measured by its ability to
approximate the finite sample properties of the estimators and inference
procedures. Thus, we study the MERM estimators in a variety of simulation
experiments. The results confirm that the nonstandard asymptotic theory
indeed provides a good approximation of the finite sample properties of the
estimators even in the settings with relatively large EIV. In some of the
simulation experiments, the EIV are so large that the naive estimators'
standard $95\%$ confidence intervals have actual coverages of $0\%$ in
finite samples, due to the magnitude of the EIV\ bias. At the same time,
even in these settings the MERM estimators perform well, removing the EIV
bias and providing confidence intervals with the correct coverage. In
particular, the simulation results show that despite the simplicity of
implementation, the MERM estimators can compete with and outperform
semi-nonparametric estimators.

The MERM estimator is structurally different from the existing approaches
that require nonparametric estimation of some nuisance parameters, for
example, of the density $f_{X^{\ast }|Z,W}$. Avoiding nonparametric
estimation has at least two advantages. First, since the majority of
empirical applications include additional covariates $W_{i}$, nonparametric estimation is often infeasible due to
the curse of
dimensionality. Because the MERM estimator does not involve any
nonparametric estimation, it can be used in applications with a relatively
large number of additional covariates $W_{i}$, and remains feasible even in
the more complicated settings, including multi-equation and structural
models, and applications with multiple mismeasured variables $X_{i}$.
Second, estimation of infinite-dimensional nuisance parameters is typically
more demanding in terms of the sources of identification available in the data,
for example, requiring an instrumental variable with a large support
(continuously distributed). In contrast, having a discrete instrument is
sufficient for the MERM approach because the nuisance parameter $\gamma _{0}$
is finite-dimensional.

For example, in Section~\ref{sec:empirical} we consider estimation of the
model of multinomial choice among three modes of transportation. A leading
alternative approach to the errors-in-variables problem in this model is the
semi-nonparametric sieve-MLE estimator advocated by \citet{HuSchennach2008Ecta,CarrollChenHu2010JoNS}, among
others. This
approach requires, among other things, estimating the conditional density $f_{X^{\ast }|Z,W}$ of $X_{i}^{\ast }$
given the instrument $Z_{i}$ and
covariates $W_{i}$. In this empirical example, $W_{i}$ includes four
continuously distributed covariates (two continuously distributed
characteristics per choice) and a discrete one, while scalar $X_{i}^{\ast }$
and $Z_{i}$ are also continuous. Thus, $f_{X^{\ast }|Z,W}$ is a function of
six continuous and one discrete variable. Hence, for typical sample sizes,
estimating $f_{X^{\ast }|Z,W}$ in this example is infeasible due to the
curse of dimensionality. In contrast, as the results of Section~\ref
{sec:empirical} demonstrate, the MERM approach is practical and effective in
this application, in part because
it avoids estimation of the high-dimensional nuisance functions like $f_{X^{\ast }|Z,W}$ altogether.

The simplicity and practicality of the MERM\ approach do come at a cost:
there is a limit on the magnitude of the measurement errors it can handle.
For example, one generally should not expect the MERM\ approach to work well
when $\tau >1$, i.e., when the noise dominates the signal; in this case the
researcher should seek an alternative estimation method.

We view the MERM\ approach as providing a bridge between the settings in
which the measurement errors are guaranteed to be absent or negligible, and
the settings where the measurement errors are so large that one has to use
the relatively more complicated estimators from the earlier literature (if
they exist at all for the model of interest).

\smallskip

{\noindent \textbf{Related Literature}} \cite{ChenHongNekipelov2011JEL},
\cite{Schennach2016AnnRev}, and \cite{Schennach2020HB-ME} provide excellent
overviews of the measurement error literature.

The existing semiparametric approaches to estimation and inference in models
with EIV\ involve nonparametric estimation of infinite-dimensional nuisance
parameters (e.g., \citealp{Chesher2000WP,Li2002JoE,Schennach2004Ecta,Schennach2007Ecta,HuSchennach2008Ecta,SchennachHu2013JASA,Song2015JoE}), simulation (e.g., \citealp{Schennach2014Ecta-ELVIS}), or both
(e.g., \citealp{Newey2001REStat,WangHsiao2011JoE}$)$. The exceptions
include models with linear and polynomial regression functions (see
\citealp{HINP1991JoE,HausmanNeweyPowell1995JoE}$)$, and Gaussian
control variable models such as Probit and Tobit with endogeneity (see
\citealp{SmithBlundell1986Ecta,RiversVuong1988JoE}).

To the best of our knowledge, this paper is the first to provide an approach
for $\sqrt{n}$-consistent and asymptotically normal and unbiased estimation
of general GMM models with EIV that does not require any nonparametric
estimation (or simulation).

We are able to provide such an estimator because we focus on the models with
moderate measurement errors. Modeling the variance of the measurement error
as shrinking to zero with the sample size is a popular approach in
Statistics. The method has been proposed by \cite{WolterFuller1982AS}, who
used it to construct an approximate MLE\ estimator of a nonlinear regression
model with Gaussian errors. Following their approach, the Statistics
literature has mainly focused on the settings where the moments of the EIV
needed to bias correct the estimators are either known or can be directly
estimated from the available data such as repeated measurements (e.g.,
\citealp{CarrollStefanski1990JASA,CarrollEtAl2006Book-ME}). In
Economics, such data are relatively rare. The use of approximations with
shrinking variance of measurement errors in the Econometrics literature has been
pioneered by \cite{Kadane1971Ecta}, \cite{Amemiya1985JoE}, and \cite
{Chesher1991Biomet}. Such approximations have been used to check the
sensitivity of naive estimators to the EIV by considering how the estimates
change as the unknown moments of the measurement errors vary within some set
of plausible values, e.g., see \cite{ChesherSchluter2002ReStud}, \cite
{ChesherDumanganeSmith2002JoE}, \cite{BattistinChesher2014JoE}, \cite
{Chesher2017JoE}, and \cite{HongTamer2003JoE}. \citet{BoundBrownMathiowetz2001HBoE} review a broad list of validation
studies matching standard economic datasets to administrative records. The
estimates they report suggest that the measurement errors of moderate
magnitude are typical for empirical applications. This suggests that the
approach developed in this paper could prove valuable for a wide range of
applied work.

This paper differs from the earlier literature in several ways. First, it
presents a way to estimate the unknown nuisance parameters (moments of the
measurement errors) jointly with the parameters of interest. As a result,
the approach can, for example, use instrumental variables as a source of
identification. {} Second, the method applies to
a very general class of semiparametric models specified by moment
conditions. Third, the MERM approach allows the measurement errors to have
larger magnitudes than most of the papers in the earlier literature; this is
achieved by the MERM approach recursively bias correcting the bias
correction terms.

The most widespread approach to identification of the EIV\ models in
economic applications is to use instrumental variables, e.g., see \cite
{HINP1991JoE,Newey2001REStat,Schennach2007Ecta,WangHsiao2011JoE}. In a
recent paper, \cite{HahnHausmanKim2021EL} reconsider the regression model in
\cite{Amemiya1990JoE} using a bias correction similar to ours. When proper
excluded variables are not available, researchers have considered using
higher moments of $X_{i}$ as instruments, e.g., see \cite
{Reiersol1950Ecta,Lewbel1997Ecta,EricksonWhited2002ET,SchennachHu2013JASA,BenMosheDHaultfeuilleLewbel2017JoE}.
When available, repeated measurements can also be used to identify the
model, e.g., see \cite
{HINP1991JoE,LiVuong1998JoE,Li2002JoE,Schennach2004Ecta}. The MERM estimator
accommodates these identification approaches within a unified estimation
framework.

The power of the general MERM\ approach can be illustrated in the NLR model.
For example, when a candidate instrumental variable is available, the
conditions it needs to satisfy are much weaker than what is required by many
existing approaches. Availability of a discrete instrument is sufficient for
identification; and the instrument is allowed to have heterogeneous impact
on covariates $X_{i}^{\ast }$.\ One can also take a nonclassical, nonlinear (e.g., discretized or
censored), or biased measurement of $X_{i}^{\ast }$ as an instrument in the
MERM approach. We discuss identification in Section~\ref{ssec:MERM-ID}. In
addition, in a related paper \cite{EvdokimovZeleneev2022WP-NPID} study
nonparametric regression with EIV using the $\tau \rightarrow 0$
approximation, and demonstrate that the MERM approach can also be motivated
from a nonparametric perspective.

\cite
{KitamuraOtsuEvdokimov2013Ecta,AndrewsGentzkowShapiro2017QJE,ArmstrongKolesar2021QE,BonhommeWeidner2021QE}, among
others, develop tools for estimation and inference in GMM, which are
robust to general perturbation or misspecification of the true data
generating process. They focus on the settings in which these perturbations
are sufficiently small, so that naive estimators remain $\sqrt{n}$-consistent, and their biases are of the same
order of magnitude as their
standard errors. In contrast, we focus on more specific forms of data
contamination due to the EIV. This allows the MERM approach to remain valid
even in the settings with larger measurement errors, in which naive
estimators may have slower than $\sqrt{n}$ rates of convergence.

The MERM approach also provides a useful foundation for dealing with EIV in
more complicated settings. \cite{EvdokimovZeleneev2018WP-Inference} utilize
the MERM framework to address an issue of nonstandard inference, which turns
out to arise generally when EIV models are identified using instrumental
variables. \cite{EvdokimovZeleneev2019WP-Panel} extend the analysis of this
paper to long panel and network settings.

\smallskip

{\noindent \textbf{Organization of the paper}} Section~\ref{sec:framework}
introduces the Moderate Measurement Error framework and the proposed MERM\
estimator. Section \ref{sec:MC} presents several Monte Carlo experiments
that illustrate finite sample properties of the MERM\ estimators. Section \ref{sec:Extensions} considers several
extensions of the framework. The appendices contain all proofs and
additional results for the numerical and empirical illustrations.

\section{Moderate Measurement Errors Framework\label{sec:framework}}

To present the main ideas we first consider the case of univariate $X_{i}^{\ast }$. We will consider multivariate
$X_{i}^{\ast }$ later. We assume that the measurement error is classical, i.e., that $\varepsilon_i$ is
independent of $X_i^*$ and $S_i$; later we will discuss how this assumption can be relaxed. Following the rest of
the literature, we assume
that $\mathbb{E}\left[ \varepsilon _{i}\right] =0$.\footnote{A location normalization such as $\mathbb{E}\left[
\varepsilon _{i}\right]
=0 $ is usually necessary because it is not possible to separately identify
the means $\mathbb{E}\left[ X_{i}^{\ast }\right] $ and $\mathbb{E}\left[
\varepsilon _{i}\right] $.}

To develop a practical estimation approach for general moment condition
models we focus on the settings in which $\tau \equiv \sigma _{\varepsilon
}/\sigma _{X^{\ast }}$ is small or moderate. We consider an asymptotic
approximation with $\tau _{n}\equiv \tau \rightarrow 0$ as $n\rightarrow
\infty $.

Note that economically meaningful parameters are usually invariant
to rescaling of~$X_{i}^{\ast }$. Likewise, the extent of the EIV problem
does not change with such rescaling. For simplicity of exposition, it is convenient to
assume that $X_i^*$ is scaled so that $\sigma_{X^*}$ is of
order one and, correspondingly, the moments $\mathbb{E}[\left\vert \varepsilon
_{i}\right\vert ^{k}]\propto \tau _{n}^{k}$ decrease with $k$ when $\tau_{n}<1$.
For example, this could be ensured by normalizing observed $X_{i}$
to have $\sigma _{X}=1$. Let us stress that this normalization is used only
to simplify the exposition; as we show in Appendix~\ref{sec:MME expansion
with large ME}, the proposed MERM\ estimator does not require any
normalizations in practice.

\subsection{Special Case: Quadratic Expansion}

For clarity, we first consider a simple special case of the general
approach. Let us denote $g_{x}^{(k)}\left( x,s,\theta \right) \equiv
\partial ^{k}g\left( x,s,\theta \right) /\partial x^{k}$. Since $\mathbb{E}[\left\vert \varepsilon _{i}\right\vert
^{k}]\propto \tau
_{n}^{k}\rightarrow 0$ as $n\rightarrow \infty $, under some regularity
conditions, we can write the quadratic Taylor expansion of function $g(X_{i},S_{i},\theta )=g(X_{i}^{\ast
}+\varepsilon _{i},S_{i},\theta )$
around $\varepsilon _{i}=0$ as
\begin{align}
\mathbb{E}[g(X_{i},S_{i},\theta )]& =\mathbb{E}\left[ g(X_{i}^{\ast
},S_{i},\theta )+g_{x}^{(1)}(X_{i}^{\ast },S_{i},\theta )\varepsilon _{i}+\dfrac{1}{2}g_{x}^{(2)}(X_{i}^{\ast
},S_{i},\theta )\varepsilon _{i}^{2}\right] +O(\mathbb{E}[\left\vert \varepsilon _{i}\right\vert ^{3}])  \notag
\\
& =\mathbb{E}[g(X_{i}^{\ast },S_{i},\theta )]+\dfrac{\mathbb{E}[\varepsilon
_{i}^{2}]}{2}\mathbb{E}[g_{x}^{(2)}(X_{i}^{\ast },S_{i},\theta )]+O(\tau
_{n}^{3}),  \label{eq:moment expectation for K=2}
\end{align}
where the second equality holds because $\varepsilon _{i}$ and $\left( X_{i}^{\ast },S_{i}\right) $ are
independent, and ${\mathbb{E}\left[
\varepsilon _{i}\right] =0}$.

Evaluating the expansion above at $\theta = \theta_0$ gives $\mathbb{E}[g(X_{i},S_{i},\theta _{0})]=O\left(
\sigma _{\varepsilon }^{2}\right) =O(\tau _{n}^{2})$, because $\mathbb{E}[g(X_{i}^*,S_{i},\theta _{0})] = 0$. As a
result, a naive
estimator that ignores the EIV and uses $X_{i}$ in place of $X_{i}^{\ast }$
has EIV\ bias of order $\tau _{n}^{2}$.\footnote{For example, consider a linear regression with a scalar mismeasured
regressor. The bias of the naive OLS\ estimator of the slope parameter $\theta _{01}$ is $-\theta _{01}\frac{\tau
_{n}^{2}}{1+\tau _{n}^{2}}=-\theta
_{01}\tau _{n}^{2}+O\left( \tau _{n}^{4}\right) $.} The bias of the naive
estimator should be compared with its standard error, which is of order $n^{-1/2}$. Thus, the bias of the naive
estimator is not negligible, unless the
measurement error is rather small (theoretically, unless $\tau
_{n}^{2}=o\left( n^{-1/2}\right) $). In particular, tests and confidence
intervals based on the naive estimator are invalid and can provide highly
misleading results. Moreover, if $\tau _{n}^{2}$ shrinks at a rate slower
than $O\left( n^{-1/2}\right) $, {} the rate of convergence of the naive estimator is slower
than $\sqrt{n}$.

Suppose $\tau _{n}=o\left( n^{-1/6}\right) $. Then, $O(\tau
_{n}^{3})=o\left( n^{-1/2}\right) $ and we can rearrange equation~(\ref
{eq:moment expectation for K=2}) as
\begin{equation}
\mathbb{E}[g(X_{i}^{\ast },S_{i},\theta )]=\mathbb{E}[g(X_{i},S_{i},\theta
)]-\frac{\mathbb{E}[\varepsilon _{i}^{2}]}{2}\mathbb{E}\left[
g_{x}^{(2)}(X_{i}^{\ast },S_{i},\theta )\right] +o(n^{-1/2}).
\label{eq:g MME mu for K=2}
\end{equation}
The left-hand side of this equation is exactly the moment condition~(\ref
{eq:g descr moment}) that we would like to use for estimation of $\theta
_{0} $. The first term on the right-hand side involves only observed
variables, and can be estimated by the sample average $\overline{g}(\theta
)\equiv n^{-1}\sum_{i=1}^{n}g(X_{i},S_{i},\theta ).$ The second term on the
right-hand side can be thought of as a bias correction that removes the
EIV-bias from the expected moment function $\mathbb{E}[g(X_{i},S_{i},\theta
)]$.

The idea of the MERM\ estimator we propose is to make use of expansions such
as~(\ref{eq:g MME mu for K=2}) to bias correct the moment condition $\mathbb{E}[g(X_{i},S_{i},\theta )]$, which in
turn removes the bias of the estimator
of the parameters of interest $\theta _{0}$. To perform the bias correction
we need to estimate two quantities: $\mathbb{E}[\varepsilon _{i}^{2}]$ and $\mathbb{E}[g_{x}^{(2)}(X_{i}^{\ast
},S_{i},\theta )]$. {}

First, we show that in equation~(\ref{eq:g MME mu for K=2}) we can
substitute $\mathbb{E}[g_{x}^{(2)}(X_{i}^{\ast },S_{i},\theta )]$ with $\mathbb{E}[g_{x}^{(2)}(X_{i},S_{i},\theta
)]$, which in turn can be
estimated by $\overline{g}_{x}^{(2)}(\theta )\equiv
n^{-1}\sum_{i=1}^{n}g_{x}^{(2)}(X_{i},S_{i},\theta )$. By the Taylor
expansion around $\varepsilon _{i}=0$ similar to equation~(\ref{eq:moment
expectation for K=2}), we can show that $\mathbb{E}[g_{x}^{(2)}(X_{i}^{\ast
},S_{i},\theta )]=\mathbb{E}[g_{x}^{(2)}(X_{i},S_{i},\theta )]+O(\tau
_{n}^{2})$ and hence
\begin{equation}
\frac{1}{2}\mathbb{E}[\varepsilon _{i}^{2}]\left( \mathbb{E}\left[
g_{x}^{(2)}(X_{i}^{\ast },S_{i},\theta )\right] -\mathbb{E}\left[
g_{x}^{(2)}(X_{i},S_{i},\theta )\right] \right) =\mathbb{E}[\varepsilon
_{i}^{2}]O\left( \tau _{n}^{2}\right) =O\left( \tau _{n}^{4}\right) .
\label{eq:bias replacing Xs with X for K=2}
\end{equation}
Here $O\left( \tau _{n}^{4}\right) =o\left( n^{-1/2}\right) $ because we
assume that $\tau _{n}=o\left( n^{-1/6}\right) $. The idea behind this
substitution is that the bias of order $O(\tau _{n}^{2})$ in $\mathbb{E}[g_{x}^{(2)}(X_{i},S_{i},\theta )]$ can be
ignored because it is multiplied
by $E\left[ \varepsilon _{i}^{2}\right] =O\left( \tau _{n}^{2}\right) $.\footnote{Such substitutions of $X^{\ast
}$ with $X$ have been used in other contexts,
e.g., \cite{ChesherSchluter2002ReStud}.} With the substitution, we can
rearrange equation~(\ref{eq:g MME mu for K=2}) and write it as
\begin{equation}
\mathbb{E}[g(X_{i}^{\ast },S_{i},\theta )]=\mathbb{E}\left[
g(X_{i},S_{i},\theta )-\frac{\mathbb{E}[\varepsilon _{i}^{2}]}{2}g_{x}^{(2)}(X_{i},S_{i},\theta )\right] +o(n^{-1/2}).
\label{eq:expansion - almost psi for K=2}
\end{equation}

Second, we propose estimating the unknown $\mathbb{E}[\varepsilon _{i}^{2}]$
together with the parameter of interest $\theta $. Specifically, let $\gamma
_{02}\equiv \mathbb{E}[\varepsilon _{i}^{2}]/2$ denote the true value of
parameter $\gamma _{2}$, and consider the following \emph{corrected moment
function}:
\begin{equation}
\psi (X_{i},S_{i},\theta ,\gamma )\equiv g(X_{i},S_{i},\theta )-\gamma
_{2}g_{x}^{(2)}(X_{i},S_{i},\theta ).
\label{eq: psi moms definition for K=2}
\end{equation}
Function $\psi $ is a moment function parameterized by $\theta $ and $\gamma
$, and
\begin{equation}
\mathbb{E}[\psi (X_{i},S_{i},\theta _{0},\gamma _{02})]=\mathbb{E}[g(X_{i}^{\ast },S_{i},\theta _{0})]+o\left(
n^{-1/2}\right) =o\left(
n^{-1/2}\right) ,  \label{eq:E psi = o(n-12) for K=2}
\end{equation}
where the first equality follows from equation~(\ref{eq:expansion - almost
psi for K=2}) and the definition of $\gamma _{02}$, and the second equality
follows from equation~(\ref{eq:g descr moment}). Hence, the corrected moment
conditions $\psi $ can be used to jointly estimate the true parameters $\theta _{0}$ and $\gamma _{02}$ by a
GMM\ estimator.\footnote{In the moment condition settings, having $o\left( n^{-1/2}\right) $ is
equivalent to having $0$ on the right-hand side\ of equation~(\ref{eq:E psi
= o(n-12) for K=2}).}

\begin{remark}
\label{remark: symmetric ME} If $\mathbb{E}[\varepsilon_{i}^{3}]=0$ (e.g., if the distribution of
$\varepsilon_{i}$ is symmetric), the remainder in equation~(\ref{eq:moment expectation for K=2}) is of a smaller
order $O(\tau_{n}^{4})$.
Hence, the corrected moments~(\ref{eq:E psi = o(n-12) for K=2})
remain valid for larger values of $\tau _{n}$, requiring only the weaker
condition $\tau _{n}=o(n^{-1/8})$. The bias of the naive estimators in this
case can be as large as $o(n^{-1/4})$.
\end{remark}

\subsection{General Case: Expansion of order $K$}

The quadratic expansion of equation~(\ref{eq:moment expectation for K=2})
can be extended to general order $K\geq 2$. Considering larger $K$
theoretically allows for $\tau _{n}$ converging to zero at a slower rate. In
finite samples this corresponds to the asymptotics providing good
approximations for larger values of $\tau _{n}$, i.e., large measurement
errors. {} Expanding $g(X_{i}^{\ast
}+\varepsilon _{i},S_{i},\theta )$ around $\varepsilon _{i}=0$, we have
\begin{equation}
\mathbb{E}[g(X_{i},S_{i},\theta )]=\mathbb{E}\left[ g(X_{i}^{\ast
},S_{i},\theta )+\sum_{k=1}^{K}\frac{\varepsilon _{i}^{k}}{k!}g_{x}^{(k)}(X_{i}^{\ast },S_{i},\theta )\right]
+O\left( \mathbb{E}\left[
\left\vert \varepsilon _{i}\right\vert ^{K+1}\right] \right) .
\label{eq: moment exp approx - step 1}
\end{equation}
The above special case of quadratic expansion corresponds to $K=2$.

The approximation we consider is formalized by the following assumption.

\begin{assumption}[MME]
\emph{(Moderate Measurement Errors)} \namedlabel{ass:MME}{MME} (i) $\tau _{n}=o(n^{-1/\left( 2K+2\right) })$ for
some integer $K\geq 2$;
and (ii) $\mathbb{E}[\left\vert \varepsilon _{i}\right\vert ^{L}] \leq C
\sigma_\varepsilon^L $ for some $L\geq K+1$ and $C > 0$.
\end{assumption}

Assumption~\ref{ass:MME}(i) limits the magnitude of the measurement errors
and implies that $\tau _{n}^{K+1}=o\left( n^{-1/2}\right) $. Assumption~\ref
{ass:MME}(ii) implies that $\mathbb{E}[\left\vert \varepsilon
_{i}\right\vert ^{k}]=O\left( \sigma _{\varepsilon }^{k}\right) $, and
requires the tails of $\varepsilon _{i}/\sigma _{\varepsilon }$ to be
sufficiently thin. Together, parts (i) and (ii) imply that $\mathbb{E}[\left\vert \varepsilon _{i}\right\vert
^{K+1}]=O\left( \tau
_{n}^{K+1}\right) =o\left( n^{-1/2}\right) $, and hence ensure that the
remainder in equation~(\ref{eq: moment exp approx - step 1})$\ $is
negligible. Using the independence of $\varepsilon_i$ and $(X_i^*,S_i)$ and
$\mathbb{E}[\varepsilon_i]=0$ to further simplify this expansion
and rearranging the terms we obtain
\begin{equation}
\mathbb{E}[g(X_{i}^{\ast },S_{i},\theta )]=\mathbb{E}[g(X_{i},S_{i},\theta
)]-\sum_{k=2}^{K}\frac{\mathbb{E}[\varepsilon _{i}^{k}]}{k!}\mathbb{E}\left[
g_{x}^{(k)}(X_{i}^{\ast },S_{i},\theta )\right] +o(n^{-1/2}).
\label{eq:g MME mu}
\end{equation}
This equation is the general expansion analog of equation~(\ref{eq:g MME mu
for K=2}). The summation on the right hand side is the bias correction term,
which we use to construct the MERM\ estimator.\footnote{It is useful to get a sense of the
magnitudes of the coefficients $\mathbb{E}\left[ \varepsilon _{i}^{k}\right]
/k!$ in equation~(\ref{eq:g MME mu}). Suppose $\varepsilon _{i}\sim N\left(
0,\sigma _{\varepsilon }^{2}\right) $, $\sigma _{\varepsilon }=0.5$, and $\sigma _{X^{\ast }}=1$, so$\ \tau
=\sigma _{\varepsilon }=0.5$. Then the
coefficients in front of $g_{x}^{\left( 2\right) }$, $g_{x}^{\left( 4\right)
}$, and $g_{x}^{\left( 6\right) }$ are $\mathbb{E}\left[ \varepsilon _{i}^{2}\right] /2!=0.125$, $\mathbb{E}\left[
\varepsilon _{i}^{4}\right] /4!\approx
0.008$, and $\mathbb{E}\left[ \varepsilon _{i}^{6}\right] /6!\approx 0.0003$. } {}

It turns out that for $K\geq 4$, estimation of $\mathbb{E}[g_{x}^{(k)}(X_{i}^{\ast },S_{i},\theta )]$ is more
intricate than in the
case of $K=2$, and the substitution we made in equation~(\ref{eq:expansion -
almost psi for K=2}) no longer works. Larger values of $K$ allow for larger
values of $\tau _{n}$ and hence larger EIV\ biases of naive estimators
$n^{-1}\sum_{i=1}^{n}g_{x}^{(k)}(X_{i},S_{i},\theta )$. The expansion of
order $K$ includes terms up to the order $\tau _{n}^{K}$, with the
asymptotically negligible remainder of order $O\left( \tau _{n}^{K+1}\right)
$. For $K\geq 4$, terms of order $\tau _{n}^{4}$ are not negligible. This
implies that we cannot ignore the EIV bias that would arise from
substituting $\mathbb{E}[g_{x}^{(2)}(X_{i}^{\ast },S_{i},\theta )]$ with
$\mathbb{E}[g_{x}^{(2)}(X_{i},S_{i},\theta )]$ in equation~(\ref{eq:g MME mu}), because this bias is of order
$O\left( \tau _{n}^{4}\right) $ according
to equation~(\ref{eq:bias replacing Xs with X for K=2}). To address this
problem, we instead replace $\mathbb{E}[g_{x}^{(2)}(X_{i}^{\ast
},S_{i},\theta )]$ with the bias corrected expression $\mathbb{E}[g_{x}^{(2)}(X_{i},S_{i},\theta )]-\left(
\mathbb{E}[\varepsilon
_{i}^{2}]/2\right) \mathbb{E}[g_{x}^{(4)}(X_{i},S_{i},\theta )]$. Thus, for $K\geq 4$, one needs to bias correct
the estimator of the bias correction
term. Moreover, for larger $K$ one needs to bias correct the bias correction
of the bias correction term and so on.

Fortunately, we show that these bias corrections can be constructed as
linear combinations of the expectations of the higher order derivatives of $g_{x}^{\left( k\right)
}(X_{i},S_{i},\theta )$. Let us define the following
\emph{corrected moment function}:
\begin{equation}
\psi (X_{i},S_{i},\theta ,\gamma )\equiv g(X_{i},S_{i},\theta
)-\sum_{k=2}^{K}\gamma _{k}g_{x}^{(k)}(X_{i},S_{i},\theta ),
\label{eq: psi moms definition}
\end{equation}
where $\gamma =(\gamma _{2},\dots ,\gamma _{K})^{\prime }$ is a $K-1$
dimensional vector of parameters. Let $\gamma _{0}\equiv (\gamma _{02},\dots
,\gamma _{0K})^{\prime }$ denote the vector of true parameters $\gamma _{0k}$, defined as
\begin{equation}
\gamma _{02}\equiv \frac{\mathbb{E}\left[ \varepsilon _{i}^{2}\right] }{2}\text{,} \;\; \gamma _{03}\equiv
\frac{\mathbb{E}\left[ \varepsilon _{i}^{3}\right] }{6}, \;\; \text{and} \;\; \gamma _{0k}\equiv
\frac{\mathbb{E}\left[ \varepsilon _{i}^{k}\right] }{k!}-\sum_{\ell =2}^{k-2}\frac{\mathbb{E}\left[ \varepsilon
_{i}^{k-\ell }\right] }{(k-\ell )!}\gamma
_{0\ell }\;\; \text{for} \;\; k\geq 4.  \label{eq: gammas}
\end{equation}

We formalize this discussion below.

\begin{assumption}[CME]
\namedlabel{ass: CME}{CME}
\emph{(Classical Measurement Error)} $\varepsilon _{i}$ is independent from $(X_{i}^{\ast },S_{i})$ and
$\mathbb{E}[\varepsilon _{i}]=0$.
\end{assumption}

The following lemma establishes validity of the corrected moment conditions
under Assumptions \ref{ass:MME}, \ref{ass: CME}, and some mild regularity
conditions provided in Appendix~\ref{sec:LargeSample}.

\begin{lemma}
\label{lem: corrected moment} Under Assumptions \ref{ass:MME}, \ref{ass: CME}
and \ref{ass: Moment function} in Appendix~\ref{sec:LargeSample},
\begin{equation*}
\mathbb{E}[\psi (X_{i},S_{i},\theta _{0},\gamma _{0})]=\mathbb{E}[g(X_{i}^{\ast },S_{i},\theta _{0})]+o\left(
n^{-1/2}\right) =o\left(
n^{-1/2}\right).
\end{equation*}
\end{lemma}

Lemma~\ref{lem: corrected moment} implies that the corrected moment
conditions $\psi $ are valid and can potentially be used to jointly estimate
parameters $\theta _{0}$ and $\gamma _{0}$. The total number of parameters
to be estimated is now $\dim \left( \theta \right) +K-1$. Thus, joint
estimation of $\theta _{0}$ and $\gamma _{0}$ requires that $\dim \left(
\psi \right) =\dim \left( g\right) \geq \dim \left( \theta \right) +K-1$,
i.e., that the original moment conditions $g$ include sufficiently many
overidentifying restrictions. For example, the overidentifying restrictions can be
constructed by using an instrumental variable; we discuss this in more
detail below.

\begin{remark}
        Construction of the corrected moment conditions $\psi$ requires the original moment function
        $g(x,s,\theta)$ to have a sufficient number of derivatives with respect to~$x$. Thus, the proposed
        correction method does not apply to settings with non-differentiable moment functions, for example, those
        arising in the instrumental variable quantile regression (IVQR).
\end{remark}

\subsection{Measurement Error Robust Moments (MERM) Estimator}

The MERM estimator jointly estimates the parameters $\theta _{0}$ and $\gamma _{0}$ using moment conditions $\psi
$. It is convenient to define the
joint vector of parameters
\begin{equation*}
\beta \equiv (\theta ^{\prime },\gamma ^{\prime })^{\prime },\ \beta
_{0}\equiv (\theta _{0}^{\prime },\gamma _{0}^{\prime })^{\prime },\text{ }\hat{\beta}\equiv (\hat{\theta}^{\prime
},\hat{\gamma}^{\prime })^{\prime },
\end{equation*}
and the parameter space $\mathcal{B}\equiv \Theta \times \Gamma $, where $\Theta $ and $\Gamma $ are the parameter
spaces for $\theta $ and $\gamma $.
Then, the MERM estimator is the GMM\ estimator (\citealp{Hansen1982Ecta}):
\begin{equation}
\hat{\beta}\equiv \argmin_{\beta \in \mathcal{B}}\hat{Q}(\beta ),\qquad \hat{Q}(\beta )\equiv \overline{\psi
}(\beta )^{\prime }\hat{\Xi}\overline{\psi }(\beta ),  \label{eq:MME definition}
\end{equation}
where $\overline{\psi }(\beta )\equiv n^{-1}\sum_{i=1}^{n}\psi _{i}(\beta )$, $\psi _{i}(\beta )\equiv \psi \left(
X_{i},S_{i},\beta \right) $, $\hat{\Xi
}$ is a weighting matrix, and $\hat{Q}(\beta )$ is the standard GMM
objective function.

While Lemma~\ref{lem: corrected moment} establishes validity of the
corrected moment restrictions $\psi $, the MERM estimator also relies on $\beta _{0}$ being identified from $\psi
$. This requirement is formalized by
the following assumption.

\begin{assumption}[ID]
\emph{(Identification)} \namedlabel{ass: ID}{ID}

\begin{enumerate}[(i)]

\item
\label{item: local ID} the Jacobian $\Psi ^{\ast }$ has full column rank,
where
\begin{align*}
\Psi ^{\ast }& \equiv \mathbb{E}\left[ \nabla _{\theta }\psi (X_{i}^{\ast
},S_{i},\theta _{0},0),\nabla _{\gamma }\psi (X_{i}^{\ast },S_{i},\theta
_{0},0)\right] \\
& =\mathbb{E}\left[ \nabla _{\theta }g(X_{i}^{\ast },S_{i},\theta
_{0}),-g_{x}^{(2)}(X_{i}^{\ast },S_{i},\theta _{0}),\dots
,-g_{x}^{(K)}(X_{i}^{\ast },S_{i},\theta _{0})\right] ;
\end{align*}

\item
\label{item: global ID} $\mathbb{E}\left[ \psi (X_{i}^{\ast },S_{i},\theta
,\gamma )\right] =0$ iff $\theta =\theta _{0}$ and $\gamma =0$.

\end{enumerate}
\end{assumption}

Assumptions \ref{ass: ID}(\ref{item: local ID}) and (\ref{item: global ID})
are the standard GMM local and global identification conditions applied to
the moment function $\psi (X_{i}^{\ast },S_{i},\theta ,\gamma )$. These are
high-level conditions, which we will return to later in the paper. The moment conditions formulation is
sufficiently general to encompass a wide variety of sources of identification. In Section~\ref{ssec:MERM-ID}, we
discuss identification in detail and illustrate the construction of the moment function using an instrumental
variable or a second measurement.

\bigskip

Under some additional regularity conditions, estimator $\hat{\beta}$ behaves
as a standard GMM-type estimator: it is $\sqrt{n}$-consistent and
asymptotically normal and unbiased. This result is formalized by the
following theorem.

\begin{theorem}[Asymptotic Normality]
\label{the: asy normality} Suppose that $\{(X_{i}^{\ast },S_{i}^{\prime
},\varepsilon _{i})\}_{i=1}^{n}$ are i.i.d. Then, under Assumptions \ref
{ass:MME}, \ref{ass: CME}, \ref{ass: ID}, and \ref{ass: Moment function}-\ref
{ass: basic conditions} in Appendix~\ref{sec:LargeSample},
\begin{equation}
n^{1/2}\Sigma ^{-1/2}(\hat{\beta}-\beta _{0})\overset{d}{\rightarrow }N(0,I_{\dim (\beta )}),\text{ where}
\label{eq:betah asy-N}
\end{equation}
\vspace{-4ex}
\begin{equation*}
\Sigma \equiv (\Psi ^{\prime }\Xi \Psi )^{-1}\Psi ^{\prime }\Xi \Omega
_{\psi \psi }\Xi \Psi (\Psi ^{\prime }\Xi \Psi )^{-1}.
\end{equation*}
\end{theorem}

Theorem \ref{the: asy normality} shows that the MERM\ approach addresses the
EIV bias problem, and in particular provides a $\sqrt{n}$-consistent
asymptotically normal and unbiased estimator $\hat{\theta}$, which can be
used to conduct inference about the true parameters $\theta _{0}$. The
asymptotic variance $\Sigma $ takes the standard sandwich\ form, with $\Psi
\equiv \mathbb{E}\left[ \nabla _{\beta }\psi _{i}(\beta _{0})\right] $, $\Omega _{\psi \psi }\equiv
\mathbb{E}\left[ \psi _{i}\left( \beta
_{0}\right) \psi _{i}^{\prime }\left( \beta _{0}\right) \right] $, and $\hat{\Xi}\rightarrow _{p}\Xi $.

\begin{remark}
Notice that the bias of naive estimators (such as a GMM estimator based on
the original moment conditions) is $O(\tau _{n}^{2})$, so their rate of
convergence is $O_{p}(\tau _{n}^{2}+n^{-1/2})$. In general, the bias dominates sampling
variability and naive estimators are not $\sqrt{n}$-consistent unless $\tau
_{n}=O(n^{-1/4})$, i.e., unless the magnitude of the measurement error is
rather small. At the same time, the MERM estimator remains $\sqrt{n}$-consistent for much larger values of $\tau
_{n}$, up to $\tau
_{n}=O(n^{-1/(2K+2)})$, whereas the rate of convergence of naive estimators
is only $O_{p}(n^{-1/(K+1)})$ in this case.
\end{remark}

Once the corrected moment condition $\psi $ is constructed, estimation of
and inference about parameters $\beta _{0}$ can be performed using any
standard software package for GMM\ estimation. In other words, the proposed
estimator can be simply treated as a standard GMM estimator based on the
corrected moment conditions $\psi $, and the conventional standard errors,
tests, and confidence intervals are valid.

In addition to the estimation of the parameters $\theta_{0}$, researchers are often interested in average
effects of the form $\lambda_{0}\equiv \mathbb{E}\left[ \lambda\left(
X_{i}^{\ast },S_{i},\theta _{0}\right) \right] $. For instance, in the NLR model, one may be
interested in the average partial effect of $x$ (i.e., $\lambda _{0}\equiv  \mathbb{E}\left[ \nabla _{x}\rho
\left( X_{i}^{\ast },S_{i},\theta_{0}\right) \right] $) or another covariate.
The naive average partial effect estimator $\hat{\lambda}_{\text{Naive}}\equiv \frac{1}{n}
\sum_{i=1}^{n}\lambda(X_{i},S_{i},\hat{\theta})$ suffers from the EIV bias,
unless function $\lambda$ is linear in $X_{i}^{\ast }$. Instead, one should use
estimates $\hat{\gamma}$ to construct the bias-corrected estimator $\hat{\lambda}_{\text{MERM}}\equiv
\frac{1}{n}\sum_{i=1}^{n}\left\{
\lambda(X_{i},S_{i},\hat{\theta})-\sum_{k=2}^{K}\hat{\gamma}_{k}\lambda_{x}^{(k)}(X_{i},S_{i},\hat{\theta})\right\} $.

\begin{remark}
The standard $J$-test of overidentifying restrictions remains valid in the
MERM settings, and can be used to check the model specification. The $J$-test jointly tests the following
hypotheses: \emph{(i)} $K$ is sufficiently
large to correct the EIV\ bias; \emph{(ii)} assumptions on the EIV are
valid; and \emph{(iii)} the original moment conditions $g$ are correctly
specified so equation~(\ref{eq:g descr moment}) holds, i.e., that the
original economic model is correctly specified aside from the presence of
the EIV in $X_{i}$. Thus, if the $J$-test rejects the validity of the corrected moment conditions $\psi$, the
researcher might want to \emph{(i)}~consider taking a larger $K$; \emph{(ii)} employ a different correction
method; or \emph{(iii)}~consider an alternative specification of the original moments $g$.
\end{remark}

\begin{remark}
Considering larger $K$ allows for $\tau _{n}\ $converging to zero at a slower
rate, which in finite samples corresponds to the asymptotics providing
better approximations for larger magnitudes of measurement errors. On the
other hand, taking a larger $K$ increases the dimension of the nuisance
parameter $\gamma _{0}$ and thus typically increases the variance of $\hat{\theta}$. We consider this issue in
more detail and provide a data-driven method for choosing $K$ in Section~\ref{ssec:adaptive K}.

\end{remark}

\begin{remark}
The MERM framework can be extended to the case of non-classical measurement errors; see
\citet{EvdokimovZeleneev2022WP-NPID} for details and a fully nonparametric analysis. In this paper, we focus on
the classical measurement errors, developing a practical bias correction approach, which can be easily implemented
in a wide range of economic applications. Even when Assumption~\ref{ass: CME} is violated, the deviations from it
are often limited in magnitude, so the corrections based on the \ref{ass: CME} assumption remove most of the EIV
bias. Thus, in practice, using the estimator designed for classical measurement errors is typically preferable to
ignoring mismeasurement altogether.
\end{remark}

\begin{remark}
It is important to note that $\gamma _{0k}\neq \mathbb{E}\left[ \varepsilon
_{i}^{k}\right] \left/ k!\right. $ for $k\geq 4$,\ contrary to what
equation~(\ref{eq:g MME mu}) might suggest. For example, $\gamma
_{04}=\left( \mathbb{E}\left[ \varepsilon _{i}^{4}\right] -6\sigma
_{\varepsilon }^{4}\right) \left/ 24\right. $ is negative for many
distributions, including normal. The reason that
generally $\gamma _{0k}\neq \mathbb{E}\left[ \varepsilon _{i}^{k}\right]
\left/ k!\right. $ is that the estimators of the correction terms themselves
need a correction, which is accounted for by the form of $\gamma _{0k}$.
Since there is a one-to-one relationship between $\gamma _{0}$ and the
moments $\mathbb{E}\left[ \varepsilon _{i}^{\ell }\right] $, parameter space
$\Gamma $ for $\gamma _{0}$ can incorporate restrictions that the moments
must satisfy (e.g., $\sigma _{\varepsilon }^{2}\geq 0$ and $\mathbb{E}\left[
\varepsilon _{i}^{4}\right] \geq \sigma _{\varepsilon }^{4}$). Such
restrictions can increase the efficiency of the estimator and the power of
tests.
\end{remark}

\begin{remark}
No parametric assumptions are imposed on the distribution of $\varepsilon
_{i}$, i.e. the distribution of $\varepsilon _{i}$ is treated
nonparametrically. The regularity conditions restrict only the magnitude of
the moments of $\varepsilon _{i}$. The approach imposes no restrictions on
the smoothness of the distributions of $X_{i}^{\ast }$ and $\varepsilon _{i}$, which are not even required to be
continuous. Examples in which this can
be useful include individual wages (whose distributions may have point
masses at round numbers), and allowing the measurement error $\varepsilon
_{i}$ to have a point mass at zero (a fraction of the population may have a
zero measurement or recall error).
\end{remark}

\begin{remark}
The formulas of the derivatives $g_{x}^{(k)}(\cdot )$ are typically easy to
compute analytically or using symbolic algebra software. Alternatively,
these derivatives can be computed using numerical differentiation. Thus, the corrected set of moments can be
automatically produced for a
generic moment function $g(\cdot )$ provided by the user.
\end{remark}

\subsection{\label{ssec:MERM-ID}Model Identification:\ Jacobian $\Psi $}

Theorem~\ref{the: asy normality} requires $\beta _{0}$ to be identified and
the Jacobian matrix $\Psi $ to be full rank. Notably, the MERM framework
encompasses many possible sources of identification at once, including
instrumental variables, additional measurements, or nonlinearities of the
functional form. The identifying information is incorporated in the moment
functions. Essentially, our approach first characterizes in what directions
the measurement errors can bias the moment conditions $\mathbb{E}\left[
g\left( X_{i},S_{i},\theta \right) \right] $, and then uses the moments
orthogonal to those directions for identification of $\theta _{0}$. To be
more specific, we will now consider identification when in addition to the
error-laden $X_{i}$ we have either (i) a general instrument $Z_{i}$ or (ii)
a second measurement $Q_{i}$.

\paragraph{Identification Using A General Instrument $Z_{i}$}

Many applications can be formulated as the following conditional moment
restriction:
\begin{equation}
\mathbb{E}\left[ u\left( X_{i}^{\ast },S_{i},\theta \right) |X_{i}^{\ast }\right] =0\text{ iff }\theta =\theta
_{0},  \label{eq: eg-ID mom u}
\end{equation}
for some moment function $u$. For example, consider the nonlinear regression
model $\mathbb{E}\left[ Y_{i}|X_{i}^{\ast }=x\right] =\rho \left( x,\theta
_{0}\right) $, then $u\left( x,y,\theta \right) =\rho \left( x,\theta
\right) -y$.\footnote{For simplicity of the exposition, in the expectation in equation~(\ref{eq:
eg-ID mom u}) we only condition on $X_{i}^{\ast }$. The discussion applies
in a straightforward way to the settings with additional correctly measured
variables $W_{i}$ in the conditioning set, i.e., the model $\mathbb{E}\left[
u\left( X_{i}^{\ast },S_{i},\theta _{0}\right) |X_{i}^{\ast },W_{i}\right]
=0 $. For example, in the nonlinear regression example with additional
covariates $W_{i}$ we have $\mathbb{E}\left[ Y_{i}|X_{i}^{\ast }=x,W_{i}=w\right] =\rho \left( x,w,\theta
_{0}\right) $, so $u\left( x,y,w,\theta
\right) =\rho \left( x,w,\theta \right) -y$.}

In applications, identification of the models with EIV would typically rely
on an instrumental variable $Z_{i}$. Suppose the instrument satisfies the
exclusion restriction $\mathbb{E}\left[ u\left( X_{i}^{\ast },S_{i},\theta
\right) |X_{i}^{\ast },Z_{i}\right] =\mathbb{E}\left[ u\left( X_{i}^{\ast
},S_{i},\theta \right) |X_{i}^{\ast }\right] $, i.e., conditional on the
true $X_{i}^{\ast }$ the instrument has no further effect on the moment
conditions $u$. Consider the moment functions $h\left( x,s,\theta \right)
\equiv u\left( x,s,\theta \right) \otimes \varphi _{X}\left( x\right) $,
where $\varphi _{X}\left( x\right) $ is a vector of functions of $x$, e.g., $\varphi _{X}\left( x\right) \equiv
\left( 1,x,\ldots ,x^{J}\right) ^{\prime
} $. By the Law of Iterated Expectations, $\mathbb{E}\left[ h\left(
X_{i}^{\ast },S_{i},\theta _{0}\right) |Z_{i}=z\right] =0$ for all $z$.
However, the same expectation with $X_{i}^{\ast }$ replaced by the observed $X_{i}$, $\mathbb{E}\left[ h\left(
X_{i},S_{i},\theta _{0}\right) |Z_{i}=z\right] $, will depend on $z$. One way to see this is to consider a Taylor
expansion similar to equation~(\ref{eq:moment expectation for K=2}):{\small
\begin{equation*}
\mathbb{E}\left[ h\left( X_{i},S_{i},\theta _{0}\right) |Z_{i}=z\right] =\underbrace{\mathbb{E}\left[ h\left(
X_{i}^{\ast },S_{i},\theta _{0}\right)
|Z_{i}=z\right] }_{=0\text{ by the LIE}}+\gamma _{0}\mathbb{E}\left[
h_{x}^{\left( 2\right) }\left( X_{i}^{\ast },S_{i},\theta _{0}\right)
|Z_{i}=z\right] +O\left( \tau _{n}^{3}\right) ,
\end{equation*}
}which shows that $\mathbb{E}\left[ h\left( X_{i},S_{i},\theta _{0}\right)
|Z_{i}=z\right] $ is zero for all $z$ (up to a negligible remainder) unless $\gamma _{0}\neq 0$, where $\gamma
_{0}=\sigma_\varepsilon^{2}/2$. Thus, $\mathbb{E}\left[ h\left( X_{i},S_{i},\theta _{0}\right) |Z_{i}=z\right] $
varies with $z$ only because of the presence of the measurement error. Intuitively, the
magnitude of this variation then identifies the nuisance parameters $\gamma
_{0}$. Thus, one can rely on the original moment functions of the typical
form $g\left( x,s,\theta \right) =h\left( x,s,\theta \right) \otimes \varphi
_{Z}\left( z\right) $, where $\varphi _{Z}\left( z\right) $ is a vector of
functions of $z$.

The above discussion provides the intuition for identification of the
nonlinear moment condition models with EIV. It is important to note that
identification of general nonlinear moment condition models is a complicated
problem. Even in the settings without measurement errors, it is generally
not possible to give low-level conditions guaranteeing that a specific set
of nonlinear moment conditions identifies the parameter vector. The presence
of EIV\ makes the question of identification even harder.

We attempt to address this concern and make the above intuitions more
precise in two ways. First, below we consider a
specific (but frequently employed) kind of an instrument: a second
measurement (possibly non-classical). The specific form of the excluded
variable allows us to provide more transparent identification conditions.
Second, in \cite{EvdokimovZeleneev2022WP-NPID} we study a nonparametric
regression model with EIV using the $\tau _{n}\rightarrow 0$ approximation.
We show that the model is identified using an instrument (even a discrete
one), and motivate the MERM approach from a nonparametric perspective.
Finally, since the MERM estimator is a standard GMM\ estimator, one can test the
strength of identification of the model parameters, or conduct
identification-robust inference using the standard methods (e.g.,
\citealp{StockWright2000Ecta,Kleibergen2005Ecta,GuggenbergerSmith2005ET,GuggenbergerRamalhoSmith2012JoE,AndrewsMikusheva2016Ecta-ConditionalFunctionalNuisance,AndrewsI-2016-Ecta-CLC,AndrewsGuggenberger2019QE}).

\paragraph{Identification Using A Second Measurement}

Suppose we observe a second measurement
\begin{equation*}
Q_{i}=\alpha _{1}X_{i}^{\ast }+\varepsilon _{Q,i},
\end{equation*}
where $\alpha _{1}$ may not be known. Assume that $\alpha _{1}\neq 0$ and $\mathbb{E}\left[ \varepsilon
_{Q,i}|X_{i}^{\ast },S_{i},\varepsilon _{i}\right] =0$. The variance of $\varepsilon _{Q,i}$ does not need to be
small.
Note that the measurement error in $Q_{i}$ can be non-classical: $Q_{i}-X_{i}^{\ast }$ and $X_{i}^{\ast }$ are
correlated unless $\alpha
_{1}=1 $.

Consider the conditional moment restrictions~(\ref{eq: eg-ID mom u}). If $X_{i}^{\ast }$ were observed, we could
have constructed the unconditional
moments
\begin{equation*}
\mathbb{E}\left[ h\left( X_{i}^{\ast },S_{i},\theta _{0}\right) \right] =0,\text{\qquad }h\left( x,s,\theta
\right) \equiv u\left( x,s,\theta \right)
\times \left( 1,x,\ldots ,x^{J}\right) ^{\prime },
\end{equation*}
for some $J\geq \dim \left( \theta \right) -1$. Suppose that the model is identified if $X_i^*$ is observed, which
means that the Jacobian of these moment conditions has full column rank:
\begin{equation*}
\operatorname{Rk}\left( H^{\ast }\right) =\dim \left( \theta \right) \text{,
where }H^{\ast }\equiv E\left[ \nabla _{\theta }h\left( X_{i}^{\ast
},S_{i},\theta _{0}\right) \right] .
\end{equation*}

To deal with the error-laden $X_{i}$, consider the MERM estimator with $K=2$
based on the following moment function
\begin{equation}
g\left( x,s,q,\theta \right) \equiv \left(
\begin{array}{l}
h\left( x,s,\theta \right) \\
u\left( x,s,\theta \right) q\times \left( 1,x,\ldots ,x^{J-1}\right)
^{\prime }
\end{array}
\right) .  \label{eq:rank-cond:eg g}
\end{equation}
Here the total number of moments is $m=2J+1$. The additional $J$ moments
added in equation~(\ref{eq:rank-cond:eg g}) use $Q_{i}$, which will allow
identifying $\gamma _{0}=\mathbb{E}[\varepsilon _{i}^{2}]/2$.

It turns out that in these settings there is a simple sufficient condition
for Assumption~\ref{ass: ID}(\ref{item: local ID}) to hold. Appendix~\ref{sec:full rank Psi}
demonstrates that $\Psi ^{\ast }$ will have full rank if
\begin{equation}
\mathbb{E}\left[ u_{x}^{\left( 1\right) }\left( X_{i}^{\ast },S_{i},\theta
_{0}\right) \times \left( 1,X_{i}^{\ast },\ldots ,\left( X_{i}^{\ast
}\right) ^{J-1}\right) ^{\prime }\right] \neq 0.
\label{eq:eg rk Psi:suff cond}
\end{equation}
Condition~(\ref{eq:eg rk Psi:suff cond}) has a very simple interpretation:
it essentially means that $\mathbb{E}\left[ \left. u_{x}^{\left( 1\right)
}\left( X_{i}^{\ast },S_{i},\theta _{0}\right) \right\vert X_{i}^{\ast }\right] $ should not be identically zero.
For example, in the nonlinear
regression model $u\left( x,y,\theta \right) =\rho \left( x,\theta \right)
-y $, and condition~(\ref{eq:eg rk Psi:suff cond}) is satisfied as long as $\mathbb{E}\left[ \rho _{x}^{(1)}\left(
X_{i}^{\ast },\theta _{0}\right)
\left( X_{i}^{\ast }\right) ^{j}\right] \neq 0$ for some $j\in \left\{
0,\ldots ,J-1\right\} $.

\bigskip

\section{\label{sec:MC}Numerical Evidence}

\subsection{Comparison with a Semi-Nonparametric Estimation Approach}
\label{ssec: MC S07 comparison}

We compare the MERM estimator with the state-of-the-art semiparametric estimator of \citet[henceforth
S07]{Schennach2007Ecta} for nonlinear regression models. The Monte Carlo designs are taken from S07, and include
polynomial, rational fraction, and Probit nonlinear regression models. Identification of the model is ensured by
the availability of an instrument.
\begin{equation}
    Y_i = \rho(X_i^*,\theta_0) + U_i, \quad X_i^* = \pi_1 Z_i + V_i, \quad X_i = X_i^* + \varepsilon_i, \label{eq:
    MC DGP}
\end{equation}
$(Z_i,V_i,\varepsilon_i)' \sim N \left((0, 0, 0)', \text{Diag}(1, 1/4, 1/4)\right)$, $\pi_1=1$, and $n = 1000$.
The conditional expectation function $\rho$, the true value of the parameter of interest $\theta_0$, and the
conditional distribution of the regression error $U_i$ are design-specific and reported in Tables \ref{tab: MC
poly}-\ref{tab: MC probit} below. In all designs, $\tau = \sigma_{\varepsilon}/\sigma_{X^*} \approx 0.45$, so the
measurement error is ``fairly large'' \citep{Schennach2007Ecta}.

We report simulation results for the MERM estimator considering correction schemes with $K=2$ and $K=4$. The
original moment function is
\begin{equation*}
    g(x,y,z,\theta) = (y - \rho(x,\theta)) \varphi(x,z),
\end{equation*}
where we use $\varphi(x,z) = \left(1, x, z, x^2, z^2, x^3, z^3\right)'$ for $K=2$ and $\varphi(x,z) = \left(1, x,
z, x^2, xz, z^2, x^3, x^2 z, x z^2, z^3\right)'$ for $K=4$.

The finite sample properties of the MERM estimators (evaluated based on 5,000 replications) are reported in Tables
\ref{tab: MC poly}-\ref{tab: MC probit} below. For comparison, we also provide the same statistics for naive
estimators (OLS/NLLS) and for the benchmark estimator of S07 (as reported in the original paper). For the
polynomial model (Table \ref{tab: MC poly}), both $K=2$ and $K=4$ MERM estimators effectively remove the EIV bias.
Component-wise, the MERM estimators perform similarly (for $\theta_2$ and $\theta_4$) or better (for $\theta_1$
and $\theta_3$) compared to the benchmark estimator of S07. For the rational fraction model (Table \ref{tab: MC
frac}), both the MERM estimators are vastly superior to the benchmark estimator both in terms of the bias and the
standard deviation. For the probit model (Table \ref{tab: MC probit}), the MERM estimator with $K=2$ removes a
large fraction of the EIV bias compared to the NLLS estimator. However, the EIV bias remains non-negligible when
this simplest correction scheme is used. Employing a higher order correction scheme with $K=4$ completely
eliminates the remaining EIV bias, while at the same time having smaller standard deviations (than the benchmark
estimator of S07). Overall, in the considered designs, the MERM estimator with $K=4$ consistently outperforms the
benchmark estimator. It also proves to be more effective in removing the EIV bias compared to the $K=2$ estimator,
especially in the highly nonlinear settings of the considered probit design.

\begin{table}[h!]
\begin{center}
\begin{footnotesize}
\begin{threeparttable}
\caption{Simulation results for the polynomial model of S07}
\label{tab: MC poly}
      \begin{tabular*}{\textwidth}{@{\extracolsep{\fill}} l  c c c c  c c c c c c c c c }
      \toprule
      \toprule
      & \multicolumn{4}{c}{Bias} & \multicolumn{4}{c}{Std. Dev.} & \multicolumn{5}{c}{RMSE} \\
      \cmidrule(lr){2-5}  \cmidrule(lr){6-9}  \cmidrule(lr){10-14}
      &\textbf{$\theta_1$}&\textbf{$\theta_2$}&\textbf{$\theta_3$}&\textbf{$\theta_4$}&\textbf{$\theta_1$}&\textbf{$\theta_2$}&\textbf{$\theta_3$}&\textbf{$\theta_4$}&\textbf{$\theta_1$}&\textbf{$\theta_2$}&\textbf{$\theta_3$}&\textbf{$\theta_4$}& All \\ \midrule
      OLS &-0.00&-0.43&0.00&0.21&0.07&0.13&0.06&0.04&0.07&0.45&0.06&0.22&0.51\\
      S07 &-0.05&-0.07&-0.02&0.05&0.17&0.19&0.24&0.05&0.17&0.20&0.24&0.07&0.36\\
      $K=2$ &-0.00& 0.10& 0.00& 0.00& 0.10& 0.23& 0.10& 0.08& 0.10& 0.25& 0.10& 0.08& 0.29 \\
      $K=4$ &-0.00& 0.00& 0.00& 0.02& 0.09& 0.21& 0.10& 0.08& 0.09& 0.21& 0.10& 0.08& 0.27 \\
      \bottomrule
    \end{tabular*}
\begin{tablenotes}
\item \footnotesize The DGP is as in \eqref{eq: MC DGP} with $\rho(x,\theta) = \theta_{1} + \theta_{2} x +
\theta_{3} x^2 + \theta_{4} x^3$, $\theta_0 = (1,1,0,-0.5)'$, and $U_i \sim N(0,1/4)$.
\end{tablenotes}
\end{threeparttable}
\end{footnotesize}
\end{center}
\end{table}

\begin{table}[h!]
\begin{center}
\begin{footnotesize}
\begin{threeparttable}
\caption{Simulation results for the rational fraction model of S07}
\label{tab: MC frac}
      \begin{tabular*}{\textwidth}{@{\extracolsep{\fill}} l c c c  c c c  c c c  c }
        \toprule
        \toprule
        & \multicolumn{3}{c}{{Bias}} & \multicolumn{3}{c}{Std. Dev.} & \multicolumn{4}{c}{RMSE} \\
        \cmidrule(lr){2-4}  \cmidrule(lr){5-7}  \cmidrule(lr){8-11}     &\textbf{$\theta_1$}&\textbf{$\theta_2$}&\textbf{$\theta_3$}&\textbf{$\theta_1$}&\textbf{$\theta_2$}&\textbf{$\theta_3$}&\textbf{$\theta_1$}&\textbf{$\theta_2$}&\textbf{$\theta_3$}&{All}\\ \midrule
        OLS &0.339&-0.167&-0.644&0.040&0.020&0.076&0.341&0.168&0.648&0.752\\
        S07 &0.107&0.117&-0.150&0.146&0.139&0.328&0.181&0.182&0.361&0.443\\
        $K=2$ &-0.004&-0.018& 0.014& 0.062& 0.026& 0.139& 0.062& 0.032& 0.139& 0.156\\
        $K=4$ &0.014&-0.002&-0.024& 0.062& 0.031& 0.154& 0.063& 0.031& 0.156& 0.171 \\
        \bottomrule
      \end{tabular*}
\begin{tablenotes}
\item \footnotesize The DGP is as in \eqref{eq: MC DGP} with $\rho(x,\theta) = \theta_{1} + \theta_{2} x +
\frac{\theta_3}{(1 + x^2)^{2}}$, $\theta_0 = (1,1,2)'$, and $U_i \sim N(0,1/4)$.
\end{tablenotes}
\end{threeparttable}
\end{footnotesize}
\end{center}
\end{table}

\begin{table}[h!]
\begin{center}
\begin{footnotesize}
\begin{threeparttable}
\caption{Simulation results for the Probit model of S07}
\label{tab: MC probit}
    \begin{tabular*}{\textwidth}{@{\extracolsep{\fill}} l  c c  c c  c c  c}
      \toprule
      \toprule
      & \multicolumn{2}{c}{Bias} & \multicolumn{2}{c}{Std. Dev.} & \multicolumn{3}{c}{RMSE} \\
      \cmidrule(lr){2-3}  \cmidrule(lr){4-5}  \cmidrule(lr){6-8}
      &\textbf{$\theta_1$}&\textbf{$\theta_2$}&\textbf{$\theta_1$}&\textbf{$\theta_2$}&\textbf{$\theta_1$}&\textbf{$\theta_2$}&{All}\\
      \midrule
      NLLS& 0.38&-0.97&0.06&0.08&0.39&0.98&1.05\\
      S07 & 0.05&-0.06&0.39&0.53&0.39&0.53&0.69\\
      $K=2$ & 0.11&-0.31& 0.18& 0.34& 0.21& 0.46& 0.51 \\
      $K=4$ &-0.01&-0.01& 0.23& 0.42& 0.23& 0.42& 0.48 \\
      \bottomrule
    \end{tabular*}
\begin{tablenotes}
\item \footnotesize The DGP is as in \eqref{eq: MC DGP} with $\rho(x,\theta) = \frac{1}{2} (1 +
\text{erf}(\theta_{1} + \theta_{2} x))$, $\theta_0 = (-1,2)'$, and $U_i = 1 - \rho(X_i^*,\theta_0)$ with
probability $\rho(X_i^*,\theta_0)$ and $- \rho(X_i^*,\theta_0)$ otherwise.
\end{tablenotes}
\end{threeparttable}
\end{footnotesize}
\end{center}
\end{table}

\subsection{\label{ssec: MC Mult Logit} Estimation and Inference in a Multinomial Choice Model}

Consider the standard multinomial logit model, in which an agent chooses between 3 available options. For an agent
$i$ with characteristics $(X_i^*,W_i)$, the utility of option $j$ is given by
\begin{equation*}
  U_{ij} = \theta_{0j1} X_i^* + \theta_{0j2} W_{ij} + \theta_{0j3} + \epsilon_{ij} \quad \text{for } j \in \{1,2\},
\end{equation*}
and $U_{i0} = \epsilon_{i0}$ for the outside option $j = 0$, where $\epsilon_{ij}$ are i.i.d. (across $i$ and $j$)
draws from a standard type-1 extreme value distribution. The researcher observes
$\left\{(X_i,W_i,Y_{i1},Y_{i2},Y_{i0},Z_i)\right\}_{i=1}^n$, where $Y_{ij}$ is a binary variable indicating
whether agent $i$ chooses option $j$, i.e. $Y_{ij} = 1$ if and only if $j = \argmax_{j' \in \{0,1,2\}} U_{i j'}$.
In addition,
\begin{equation*}
  X_i^* = V_{i1} Z_i + V_{i0}, \quad X_i = X_i^* + \varepsilon_i, \quad W_{ij} = {\rho} X_i^*/\sigma_{X^*} +
  \sqrt{1  - \rho^2} \nu_{ij},
\end{equation*}
and $\left(V_{i1}, V_{i0}, Z_i,\varepsilon_i,\nu_{i1},\nu_{i2}\right)' \sim N\left((1,0,0,0,0,0)',
\text{Diag}(\sigma_{V1}^2,\sigma_{V0}^2, \sigma_Z^2, \sigma_\varepsilon^2, \sigma_\nu^2, \sigma_\nu^2)\right)$. In
all of the designs, we fix $(\theta_{011},\theta_{012},\theta_{013},\theta_{021},\theta_{022},\theta_{023},\rho,\sigma_{V1}^2,\sigma_{V0}^2,\sigma_Z^2,\sigma_\nu^2) = (1,0,0,0,0,0,0.7,1/2,1/2,1,1)$ and $n=2000$. We consider $\tau = \sigma_{\varepsilon}/\sigma_{X^*} \in \{1/4, 1/2, 3/4\}$. Setting $\sigma_{V1}=0$ would correspond to the additive control variable model. We omit such simulation results for brevity.

Similarly to Section \ref{ssec: MC S07 comparison}, we report results for the MERM estimators with $K=2$ and $K=4$
based on the following original moment function
\begin{align*}
     g (x,w,y,z,\theta) &= \left(\left(y_1 - p_1(x,w,\theta) \right) \varphi_1(x,z,w)', \left(y_2 - p_2(x,w,\theta)\right) \varphi_2(x,z,w)' \right)',\\
       p_j(x,w,\theta) &= \frac{\exp(\theta_{j1} x + \theta_{j2} w_{j} + \theta_{j3})}{1 + \exp(\theta_{11} x + \theta_{12} w_{1} + \theta_{13}) + \exp(\theta_{21} x + \theta_{22} w_{2} + \theta_{23})},
\end{align*}
where $\varphi_j(x,z,w) = \left(1, x, z, x^2, z^2, x^3, z^3, w_j\right)'$ for $K=2$ and $\varphi_j(x,z,w) =
\left(1, x, z, x^2, xz, z^2, x^3, x^2 z, x z^2, z^3, w_j\right)'$ for $K=4$.

We report the results for estimation and inference on the partial derivatives of the conditional choice
probabilities $p_j(x,w_1,w_2)$ with respect to $x$, $w_1$, and $w_2$, evaluated at the population means.

Table \ref{tab: MC mult logit 2} reports the finite sample biases, standard deviations, and RMSE of the MERM
estimators, as well as the sizes of the corresponding t-tests with nominal size of 5\%. To illustrate the
importance of dealing with EIV, we also report the same statistics for the standard (naive) MLE estimator that
ignores the presence of the measurement errors.

In all designs, the MLE estimator is biased, and the corresponding t-tests over-reject.
Note that failing to account for the EIV in the mismeasured variable $X_i^*$ generally biases estimators of all of
the parameters, including those corresponding to the correctly measured variables $W_{i1}$ and $W_{i2}$. In
particular, the t-tests may falsely reject true null hypotheses $\partial p_j/\partial w_\ell=0$ up to nearly
$100\%$ of the time.

The MERM estimator with $K=2$ removes a large fraction of the EIV bias in all of the designs. While this proves to
be enough to achieve accurate size control when the magnitude of the measurement error  is moderate ($\tau =
1/4$), the remaining EIV bias may still result in size distortions of the t-tests with larger measurement errors,
especially $\tau = 3/4$. Using the higher order correction scheme with $K=4$ effectively removes the EIV bias in
all of the simulation designs for all of the parameters.
Remarkably, the corresponding finite sample null rejection probabilities remain close to the nominal $5\%$ rate
even when the standard deviation of the measurement error is as large as $75\%$ of the standard deviation of the
mismeasured $X^*$.

To further check the limits of applicability of our method, we also consider larger values of $\tau \in
\{1,3/2,2\}$. The numerical results analogous to the ones reported in Table~\ref{tab: MC mult logit 2} are
provided in Table~\ref{tab: MC mult logit 2, large tau} in Appendix~\ref{ssec: large tau MC}. Specifically, we
find that inference results based on the correction scheme with $K=4$ remain accurate even for $\tau = 1$.
Unsurprisingly, inference becomes less reliable for bigger $\tau = 3/2$ and $\tau = 2$. At the same time, while
the $K=4$ correction scheme fails to entirely eliminate the EIV bias in these designs, it still removes a big
fraction of the bias and greatly improves on MLE in terms of the RMSE. Thus, while inference based on our
estimator might be less reliable in extreme settings when the measurement error overwhelms the signal, the MERM
estimator using $K=4$ appears to be sufficiently accurate over a wide range of $\tau$.

\afterpage{\clearpage

\begin{landscape}

\begin{table}[h!]
\begin{footnotesize}
\begin{threeparttable}
\caption{Simulation results for the multinomial logit model}\label{tab: MC mult logit 2}
    \begin{tabular}{ l c c c c| c c c c | c c c c}
    \toprule
    \toprule
      & \multicolumn{4}{c}{MLE} & \multicolumn{4}{c}{$K=2$} & \multicolumn{4}{c}{$K=4$}\\
      \cmidrule(lr){2-5}  \cmidrule(lr){6-9} \cmidrule(lr){10-13}
      &{\text{bias}, $10^{-2}$}&{\text{std}, $10^{-2}$}&{\text{rmse}, $10^{-2}$}&{\text{size}}&{\text{bias}, $10^{-2}$}&{\text{std}, $10^{-2}$}&{\text{rmse}, $10^{-2}$}&{\text{size}}&{\text{bias}, $10^{-2}$}&{\text{std}, $10^{-2}$}&{\text{rmse}, $10^{-2}$}&{\text{size}}\\
      \midrule
      \multicolumn{13}{c}{$\tau = 1/4$}\\
      \midrule
       $\partial p_1/\partial x$&-3.24&1.36&3.51&66.98&0.74&2.63&2.74&4.30&1.13&2.73&2.95&7.86\\
     $\partial p_1/\partial w_1$&2.32&1.64&2.84&30.74&-0.11&2.30&2.30&4.82&-0.31&2.29&2.31&6.54\\
     $\partial p_1/\partial w_2$&0.48&0.75&0.90&9.40&-0.04&0.87&0.87&4.82&-0.08&0.87&0.87&5.36\\
       $\partial p_2/\partial x$&1.96&1.17&2.28&39.44&-0.40&1.88&1.92&4.72&-0.63&1.93&2.03&6.74\\
     $\partial p_2/\partial w_1$&-1.16&0.82&1.42&30.66&0.06&1.15&1.15&4.84&0.15&1.15&1.16&6.48\\
     $\partial p_2/\partial w_2$&-0.96&1.50&1.78&9.48&0.09&1.74&1.74&4.98&0.17&1.74&1.75&5.44\\
       $\partial p_0/\partial x$&1.28&1.01&1.63&25.28&-0.34&1.43&1.47&5.08&-0.50&1.48&1.56&7.36\\
     $\partial p_0/\partial w_1$&-1.16&0.82&1.42&30.60&0.06&1.15&1.15&4.82&0.15&1.15&1.16&6.46\\
     $\partial p_0/\partial w_2$&0.48&0.74&0.88&9.46&-0.05&0.87&0.87&4.90&-0.09&0.87&0.88&5.32\\
      \midrule
      \multicolumn{13}{c}{$\tau = 1/2$}\\
      \midrule
       $\partial p_1/\partial x$&-8.97&1.09&9.04&100.00&-1.69&2.60&3.10&9.84&0.97&2.89&3.05&6.04\\
     $\partial p_1/\partial w_1$&6.44&1.53&6.62&98.96&1.39&2.44&2.81&12.14&-0.21&2.41&2.42&5.54\\
     $\partial p_1/\partial w_2$&1.28&0.72&1.47&42.54&0.29&0.90&0.94&7.00&-0.06&0.92&0.92&5.00\\
       $\partial p_2/\partial x$&5.22&0.97&5.31&99.98&1.05&1.89&2.16&10.16&-0.53&2.08&2.15&6.14\\
     $\partial p_2/\partial w_1$&-3.21&0.77&3.30&98.96&-0.69&1.22&1.40&12.10&0.10&1.21&1.21&5.52\\
     $\partial p_2/\partial w_2$&-2.52&1.41&2.88&42.82&-0.56&1.79&1.87&7.20&0.13&1.84&1.84&4.98\\
       $\partial p_0/\partial x$&3.75&0.86&3.85&98.78&0.64&1.45&1.59&8.18&-0.44&1.59&1.65&6.50\\
     $\partial p_0/\partial w_1$&-3.23&0.78&3.32&98.96&-0.70&1.22&1.41&12.06&0.10&1.21&1.21&5.48\\
     $\partial p_0/\partial w_2$&1.23&0.69&1.41&42.90&0.28&0.89&0.93&7.20&-0.07&0.92&0.92&4.94\\
     \midrule
      \multicolumn{13}{c}{$\tau = 3/4$}\\
      \midrule
       $\partial p_1/\partial x$&-13.35&0.86&13.38&100.00&-6.83&2.64&7.32&80.32&0.71&3.22&3.29&4.74\\
     $\partial p_1/\partial w_1$&9.69&1.45&9.80&100.00&4.95&2.65&5.61&65.52&0.01&2.62&2.62&5.34\\
     $\partial p_1/\partial w_2$&1.81&0.69&1.94&75.30&1.01&0.89&1.35&26.08&-0.01&0.98&0.98&5.24\\
       $\partial p_2/\partial x$&7.48&0.79&7.52&100.00&4.06&1.83&4.45&68.82&-0.37&2.32&2.35&5.74\\
     $\partial p_2/\partial w_1$&-4.83&0.73&4.88&100.00&-2.47&1.32&2.81&65.46&-0.01&1.31&1.31&5.28\\
     $\partial p_2/\partial w_2$&-3.51&1.33&3.76&75.60&-1.99&1.76&2.66&26.38&0.03&1.97&1.97&5.32\\
       $\partial p_0/\partial x$&5.87&0.73&5.92&100.00&2.77&1.47&3.14&56.28&-0.34&1.77&1.80&5.82\\
     $\partial p_0/\partial w_1$&-4.87&0.75&4.93&100.00&-2.48&1.33&2.81&65.40&-0.01&1.31&1.31&5.32\\
     $\partial p_0/\partial w_2$&1.70&0.64&1.82&75.76&0.98&0.87&1.31&26.50&-0.02&0.99&0.99&5.30\\
        \bottomrule
    \end{tabular}
\begin{tablenotes}
\item \footnotesize This table reports the simulated finite sample bias, standard deviation, RMSE, and size of the
MLE and the MERM estimators and the corresponding t-tests for the partial derivatives $\partial
p_j(x,w,\theta_0)/\partial x$, $\partial p_j(x,w,\theta_0)/\partial w_1$, $\partial p_j(x,w,\theta_0)/\partial
w_2$ for $j \in \{1,2,0\}$ evaluated at the population mean. The true values of the marginal effects are
$(\partial p_1/\partial x, \partial p_2/\partial x, \partial p_0/\partial x) = ( 0.222,   -0.111,   -0.111)$ and
zeros for the rest. The results are based on 5,000 replications.
\end{tablenotes}
\end{threeparttable}
\end{footnotesize}
\end{table}

\end{landscape}

}

\subsection{\label{sec:empirical}Empirically Calibrated Simulation: Choice of Transportation Mode}In this section,
we conduct an empirically calibrated Monte Carlo experiment to illustrate the finite sample properties of the MERM
estimator in the context of a classical multinomial choice application: choice of transportation mode (e.g.,
\citealp{McFadden1974JPubE}).

To calibrate the numerical experiment, we use the ModeCanada dataset, a survey of business travelers for the
Montreal-Toronto corridor. We focus on the subset of travelers choosing between train, air, and car ($n = 2769$),
and estimate the conditional logit model with traveler $i$'s utilities given in the table below.
\begin{center}
  \begin{tabular}{l|c}
      Mode & Utility \\
      \midrule
      Air   & $U_{i1} = \theta_{01} \thinspace Income_{i}^* + \theta_{02} \thinspace Urban_{i} + \theta_{03} + \theta_{07} \thinspace Price_{i1} + \theta_{08} \thinspace InTime_{i1} + \epsilon_{i1}$\\
      Car   & $U_{i2} = \theta_{04} \thinspace Income_{i}^* + \theta_{05} \thinspace Urban_{i} + \theta_{06} + \theta_{07} \thinspace Price_{i2} + \theta_{08} \thinspace InTime_{i2} + \epsilon_{i2}$\\
      Train & $U_{i0} = \theta_{07} \thinspace Price_{i0} + \theta_{08} \thinspace InTime_{i0} + \epsilon_{i0}$
  \end{tabular}
\end{center}

To generate the simulated samples, we randomly draw covariates from their joint empirical distribution. To
generate the simulated outcomes, we draw $\epsilon_{ij}$ from the standard type-I extreme value distribution.
The true value of $\theta_0$ is set to be the MLE estimate based on the original dataset. More details about this
numerical experiment are given in Appendix \ref{sec:Empirical Details}.

To evaluate the performance of the MERM estimator in these settings,
we generate mismeasured $Income_i = Income_i^* + \varepsilon_i$,
where $\varepsilon_i$ are independent $N(0,\sigma_\varepsilon^2)$ draws.
We focus on the individual income because it is often mismeasured. We report the results for $\tau =
\sigma_\varepsilon/\sigma_{Income^*} \in \{1/4, 1/2, 3/4\}$.

Table \ref{tab: MC emprical} reports the simulation results for the (naive) MLE estimator and for the MERM
estimators with $K=2$ and $K=4$. We focus on estimation of and inference on the income elasticities (evaluated at
the population mean of the covariates).
The MLE estimator is considerably biased for $\tau \in \{1/2, 3/4\}$, which results in substantial size
distortions of the MLE based t-tests. The MERM estimator with $K=4$ effectively eliminates the EIV bias and the
corresponding t-tests provide accurate size control in all of the considered designs.
The estimator with $K=2$ is more precise, while successfully removing the EIV bias for $\tau\le1/2$.

Overall, the MERM estimators perform well in the considered empirical context, providing a basis for estimation
and inference even for quite large values of $\tau$.

\begin{table}[h!]
\begin{footnotesize}
\begin{threeparttable}
\caption{Simulation results for the empirically calibrated conditional logit model}
\label{tab: MC emprical}
    \begin{tabular}{ l c c c c| c c c c | c c c c}
    \toprule
    \toprule
      & \multicolumn{4}{c}{MLE} & \multicolumn{4}{c}{$K=2$} & \multicolumn{4}{c}{$K=4$}\\
      \cmidrule(lr){2-5}  \cmidrule(lr){6-9} \cmidrule(lr){10-13}
      &{\text{bias}}&{\text{std}}&{\text{rmse}}&{\text{size}}&{\text{bias}}&{\text{std}}&{\text{rmse}}&{\text{size}}&{\text{bias}}&{\text{std}}&{\text{rmse}}&{\text{size}}\\
      \midrule
      \multicolumn{13}{c}{$\tau = 1/4$}\\
      \midrule
     $\partial\ln p_1/\partial\ln I$ &-0.07&0.12&0.14&9.00&0.01&0.14&0.14&5.68&0.02&0.19&0.19&7.02\\
     $\partial\ln p_2/\partial\ln I$ &0.03&0.07&0.08&5.84&-0.00&0.08&0.08&5.68&-0.01&0.10&0.10&6.40\\
     $\partial\ln p_0/\partial\ln I$ &0.05&0.13&0.13&6.10&0.00&0.14&0.14&5.42&-0.00&0.17&0.17&7.66\\
      \midrule
      \multicolumn{13}{c}{$\tau = 1/2$}\\
      \midrule
     $\partial\ln p_1/\partial\ln I$ &-0.24&0.11&0.27&61.84&-0.05&0.14&0.15&6.96&0.02&0.21&0.21&6.16\\
     $\partial\ln p_2/\partial\ln I$ &0.09&0.07&0.11&24.76&0.02&0.09&0.09&5.96&-0.01&0.10&0.10&6.16\\
     $\partial\ln p_0/\partial\ln I$ &0.16&0.12&0.20&25.36&0.04&0.15&0.15&6.06&-0.00&0.18&0.18&6.86\\
     \midrule
      \multicolumn{13}{c}{$\tau = 3/4$}\\
      \midrule
     $\partial\ln p_1/\partial\ln I$ &-0.43&0.09&0.44&99.50&-0.19&0.14&0.24&27.46&0.02&0.22&0.22&5.84\\
     $\partial\ln p_2/\partial\ln I$ &0.16&0.06&0.17&71.88&0.07&0.08&0.11&13.78&-0.01&0.11&0.11&6.32\\
     $\partial\ln p_0/\partial\ln I$ &0.29&0.11&0.31&73.20&0.12&0.15&0.19&13.40&0.00&0.19&0.19&6.32\\
        \bottomrule
    \end{tabular}
\begin{tablenotes}
\item \footnotesize This table reports the simulated finite sample bias, standard deviation, RMSE, and size of the
MLE and the MERM estimators and the corresponding t-tests for the income elasticities $\partial \ln
p_j(I,w,\theta_0)/\partial \ln I$, $j \in \{1,2,0\}$, evaluated at the population mean. The true values of the
income elasticities are $(\partial \ln p_1/\partial \ln I, \partial \ln p_2/\partial \ln I, \partial \ln
p_0/\partial \ln I) = (1.11,   -0.39,   -0.82)$. The results are based on 5,000 replications.
\end{tablenotes}
\end{threeparttable}
\end{footnotesize}
\end{table}

\section{\label{sec:Extensions}Extensions}

\subsection{\label{ssec:adaptive K}Data-driven Choice of $K$}
Making an appropriate choice of the expansion order $K$ is important for the estimation procedure.
One has to be cautious not to take $K$ too small, as this may result in an estimator that only partially removes
the EIV bias. On the other hand, picking a larger $K$ than needed might inflate standard errors and result in less
powerful inference.

In this section, we address this issue by providing a data-dependent procedure for selecting~$K$. We demonstrate
that our procedure has desirable theoretical properties. We also find that the procedure has good finite sample
properties in a set of Monte Carlo simulation experiments across different values of $\tau$ and sample sizes.

\bigskip
Consider two alternative values of the expansion order: $L$ and $K$, where $2 \leq L<K$. In practice, even-order
biases tend to dominate, so to reduce the set of choices, it is useful to focus on even values of the expansion
orders $L$ and $K$. For example, one would typically be interested in choosing between $L=2$ and $K=4$.

Let $\hat \beta_L$ and $\hat \beta_K$ denote the corresponding MERM estimators. The lower-order estimator
$\hat\beta_L$ may be preferred as having smaller asymptotic variance, provided that its remaining EIV bias
is negligible relative to its standard error. Otherwise, the more conservative $\hat \beta_K$ should be used instead.

Note that Lemma~\ref{lem: corrected moment} gives the following
expression for the linearized EIV bias of $\hat\beta_L$ due to
the additional terms accounted for when the expansion of higher
order $K$ is used:
  \begin{equation*}
    \text{AsyB}(\hat \beta_L) \equiv B \sum_{k=L+1}^{K}\gamma _{0k}\mathbb{E}[g_{x}^{(k)}(X_{i},S_{i},\theta
    _{0})] ,\quad
    B \equiv - (\Psi ^{\prime }\Xi \Psi )^{-1}\Psi ^{\prime }\Xi,
  \end{equation*}
  where matrix $B$ is based on the moments used for estimation of $\hat \beta_L$.

 Importantly, $\gamma _{0k}=O\left( \sigma_{\varepsilon}^{k}\right) $, and hence for $k>2$ we can estimate a bound
 on $\sqrt{n}\gamma _{0k}$ sufficiently quickly to provide a valid procedure for choosing $K$. To this end, we
 first estimate the model using the larger $K$. Let $\hat \beta_K = (\hat \theta', \hat \gamma')'$ and $\hat
 \sigma_\varepsilon^2 \equiv 2 \hat \gamma_{2}$, where we dropped the additional subscripts $K$ for notational
 simplicity. Then, we can estimate $\sigma_\varepsilon^K$ by $\hat \sigma_\varepsilon^K \equiv (\hat
 \sigma_\varepsilon^2)^{K/2}$. Using $\hat \sigma_\varepsilon^2 = \sigma_\varepsilon^2 + O_p (n^{-1/2})$, in the
 appendix we show that
\begin{equation}
    \sqrt{n}\left( \hat{\sigma}_\varepsilon^{K}-\sigma_\varepsilon ^{K}\right) =O_{p}\left( \sigma_\varepsilon
    ^{K-2}+n^{-\left( K-2\right) /4}\right) =o_{p}\left( 1\right). \label{eq:estn sigK}
\end{equation}

Next, consider a sequence $\varkappa _{n}\rightarrow 0$, which we will specify
precisely later, and let
\begin{equation}
    \delta _{n}=1\left\{ \max_{1\leq \ell \leq \dim (\beta_L )}\left\vert \hat{\Sigma}_{\ell \ell
    }^{-1/2}\sqrt{n}\hat{\sigma}_\varepsilon^{K} \hat B_{\ell \cdot}
    \overline{g}_{x}^{(K)}(\hat{\theta}) \right\vert \leq c_K \varkappa _{n}\right\}
    ,  \label{eq:crit choice of K}
\end{equation}
where $\hat \Sigma$ and $\hat B$ are consistent estimators of the asymptotic variance of $\hat \beta_L$ and of
$B$, with $\hat \Sigma_{\ell \ell}$ denoting its $\ell$-th diagonal element, $\hat B_{\ell \cdot}$ denoting its
$\ell$-th row, ${\overline g_x^{(K)} (\hat \theta) \equiv n^{-1} \sum_{i=1}^n g_x^{(K)}(X_i,S_i,\hat \theta)}$,
and $c_K > 0$ is a constant that we will calibrate below after we state the main theoretical result of this section.

If $\delta_n = 1$, the researchers should select $\hat \beta_L$. Otherwise, $\hat \beta_K$ should be used. The
following lemma demonstrates that the proposed selection procedure has desirable theoretical properties.

\begin{lemma}
\label{lem:choice-of-K}Suppose the hypotheses of Theorem~\ref{the: asy normality} hold for some even $K\geq 4$,
and consider $L$ satisfying $2 \leq L<K$. Suppose that $\Vert B\mathbb{E}[g_x^{(K)}(X_i,S_i,\theta_0)]\Vert$ is
bounded away from zero as $n\to\infty$. Suppose $\varkappa
_{n}n^{\left( K-L-1\right) /\left( 2L+2\right) }\rightarrow 0$. Then
$\sqrt{n}\text{AsyB}(\hat \beta_L) \delta _{n} =o_{p}\left(1\right).$

Moreover, consider $\varkappa _{n} = n^{-\left( K-L-1\right) /\left(
2L+2\right) }\left( \ln n\right) ^{-a}$ for any $a>0$. Then the criterion is
consistent, in the sense that if $\tau_n =o( n^{-\frac{1}{2L+2}-\epsilon }) $ for an $\epsilon >0$, the criterion
will choose $\hat \beta_L$ with probability approaching one.
\end{lemma}

The first part of Lemma~\ref{lem:choice-of-K} shows that, provided that $\varkappa_n$ goes to zero sufficiently
fast, $\hat \beta_L$ is selected (i.e., $\delta_n = 1$) only when its asymptotic bias is negligible. Next, note
that $\hat \beta_L$ is asymptotically unbiased as long as $\tau_n = o( n^{-\frac{1}{2L+2}})$. The second part of
the lemma shows that, for the suggested choices of $\varkappa_n$, the criterion is non-vacuous, i.e., that it does
select the MERM estimator with a smaller expansion order $L$ when this is appropriate. Typically, one would pick
$L=K-2$, so the lemma suggests taking $\varkappa
_{n} = n^{-1/\left( 2K-2\right) }\left( \ln n\right) ^{-a}$ in this case.

In practice, it is important to pick an appropriate constant $c_K$ used in the construction of $\delta_n$ in
equation~\eqref{eq:crit choice of K}.
When constructing $\delta_n$, we used $\abs{\gamma_{0K}} \propto \sigma_\varepsilon^K$, which motivates the
rule-of-thumb choice $c_K=\sigma_\varepsilon^K/\abs{\gamma_{0K}}$
based on a reference distribution to determine $c_K$.
It turns
out that the normal distribution is not only convenient, but also sufficiently
conservative (note that the bigger $\abs{\gamma_{0K}}$ is, the smaller $c_K$ is, resulting in a more conservative
selection procedure). For example, suppose $K=4$, and consider using Student's $t(\nu )$ distribution as the
reference distribution.
Then, for $\nu \geq 5$, the largest ratio $\abs{\gamma_{04}}/\sigma_\varepsilon^4$ and the smallest $c_4$
correspond to $\nu = \infty$ matching the normal distribution. Thus, in this case, using the normal reference
distribution yields the most conservative value of $c_K$ in the family of Student's $t$ distributions.

Since for the normal distribution we have $\abs{\gamma_{0K}} = \sigma_\varepsilon^K / K!!$ for even $K$, we
recommend using $c_K = K!!$ in equation~\eqref{eq:crit choice of K}. Finally, while we recommend using the normal
distribution as the reference distribution, we also stress that the results of Lemma~\ref{lem:choice-of-K} hold
even if the measurement error is not normal or if it is skewed.

We summarize the proposed procedure in the following suggested algorithm.

\begin{algorithm}
\label{alg: choose K}
\leavevmode
\begin{enumerate}
\item Pick a large enough even $K\geq 4$ and compute $\hat \beta_K = (\hat \theta', \hat \gamma')'$ and
$\hat{\sigma}_\varepsilon^{2} = 2 \hat \gamma_{2}$.

\item Compute $\delta _{n}$ in \eqref{eq:crit choice of K} with $L=K-2$, $c_K = K !!$, and $\varkappa
_{n}=\left( n\ln n\right) ^{-1/\left( 2K-2\right) }$.

\item Pick the length of expansion $L = K-2$ if $\delta _{n}=1$; otherwise, keep the initial $K$.
\end{enumerate}

\smallskip

This procedure can be iterated if desired.
\end{algorithm}

\smallskip

Algorithm~\ref{alg: choose K} can be iterated starting from a fixed
even order $\overline K\geq4$ large enough for the remaining EIV bias
to be negligible. At each step, if $\delta_n=1$, reduce the current
order by two, re-estimate the model, and repeat the procedure.
Stop when $\delta_n=0$ or order $2$ is reached. Let $\hat K_n$ denote
the resulting selected order.

The following lemma establishes the asymptotic properties of the iterated algorithm. Let
$\text{AsyB}(\hat\beta_L)$ be defined as above
with $K=\overline K$, and let
 $\text{AsyB}(\hat\beta_{\overline K})=0$.

\begin{lemma}
\label{lem:iterated-choice-of-K}
Suppose the hypotheses of Theorem~\ref{the: asy normality} hold for $\overline K$, and, for each
$K=4,6,\ldots,\overline{K}$, $\Vert B\,\mathbb{E}[g_x^{(K)}(X_i,S_i,\theta_0)]\Vert$ is bounded away from zero as
$n\to\infty$, with $B$ corresponding to $L=K-2$. Then the iterated Algorithm~\ref{alg: choose K} satisfies $\sqrt
n\text{AsyB}(\hat\beta_{\hat K_n})=o_p(1).$

Moreover, for any even $L\in\left\{2,\ldots,\overline K\right\}$, if $\tau_n=o( n^{-\frac{1}{2L+2}-\epsilon })$
for an $\epsilon>0$, then $\Pr(2 \le \hat K_n\leq L)\to1$. If, in addition, $L>2$ and $\liminf_{n\to\infty}\sqrt
n\tau_n^{L-1}>0$, then $\Pr(\hat K_n=L)\to1$.
\end{lemma}

Lemma~\ref{lem:iterated-choice-of-K} shows that the iterated
Algorithm~\ref{alg: choose K} inherits the properties established
in Lemma~\ref{lem:choice-of-K}. In particular, it guarantees
asymptotically negligible remaining EIV bias for the selected
estimator. Moreover, if
$\tau_n=o(n^{-\frac{1}{2L+2}-\epsilon})$ for some $\epsilon>0$,
so that order $L$ already provides sufficient bias correction,
the procedure selects an order no larger than $L$ with probability
approaching one. For $L>2$, the additional lower bound on $\tau_n$
also rules out selecting an order below $L$ asymptotically, since
smaller orders may fail to remove enough of the EIV bias,
leading to $\Pr(\hat K_n=L)\to1$.

\bigskip

To illustrate the performance of Algorithm~\ref{alg: choose K}, we revisit the numerical experiment considered in
Section~\ref{ssec: MC Mult Logit}. We focus on choosing between $K=4$ and $L=2$, and we report results for
$n=2000$ in Table~\ref{tab: adaptive K mult logit 2000} below. To assess the performance of the proposed algorithm
across sample sizes, we also consider $n=1000$ and $n=4000$ and report the corresponding results in
Tables~\ref{tab: adaptive K mult logit 1000}~and~\ref{tab: adaptive K mult logit 4000} in Appendix \ref{ssec:
additional adaptive MCs}.

We report the finite sample bias and RMSE for the naive MLE estimator, as well as for the MERM estimators using
$K=2$ and $K=4$, and for the adaptive MERM estimator using data-driven $K$ following Algorithm \ref{alg: choose
K}. Notice that for the considered sample sizes, $K=2$ is preferred when $\tau = 1/4$, and $K=4$ is preferred when
$\tau = 3/4$. We find that in both of these regimes and for all sample sizes, the adaptive estimator using
data-driven $K$ is essentially equivalent to the preferred estimators, i.e., our procedure selects the appropriate
$K$. Interestingly, in the intermediate regime with $\tau = 1/2$, the adaptive estimator also has the smallest
bias among the considered estimators, coming at the cost of a slightly bigger RMSE compared to the MERM estimator
using $K=4$.

Given these Monte Carlo results, and our earlier findings that $K=4$ effectively removes the EIV bias even for
fairly large values of $\tau$, in practice one may want to start with $\overline K=4$ and use Algorithm~\ref{alg:
choose K} to choose between $K=4$ and $L=2$.

\begin{table}[h!]
\begin{scriptsize}
\begin{threeparttable}
\caption{Choice of $K$ simulation results for the multinomial logit model, $n=2000$}\label{tab: adaptive K mult
logit 2000}
    \begin{tabular}{ l c c |  c c |  c c |  c c}
    \toprule
    \toprule
      & \multicolumn{2}{c}{MLE} & \multicolumn{2}{c}{$K=2$} & \multicolumn{2}{c}{$K=4$} & \multicolumn{2}{c}{data-driven $K$}\\
      \cmidrule(lr){2-3}  \cmidrule(lr){4-5} \cmidrule(lr){6-7} \cmidrule(lr){8-9}
      &{\text{bias}, $10^{-2}$}&{\text{rmse}, $10^{-2}$}&{\text{bias}, $10^{-2}$}&{\text{rmse}, $10^{-2}$}&{\text{bias}, $10^{-2}$}&{\text{rmse}, $10^{-2}$}&{\text{bias}, $10^{-2}$}&{\text{rmse}, $10^{-2}$}\\
      \midrule
      \multicolumn{9}{c}{$\tau = 1/4$}\\
      \midrule
       $\partial p_1/\partial x$&-3.24&3.51&0.74&2.74&1.13&2.95&0.75&2.75\\
     $\partial p_1/\partial w_1$&2.32&2.84&-0.11&2.30&-0.31&2.31&-0.11&2.30\\
     $\partial p_1/\partial w_2$&0.48&0.90&-0.04&0.87&-0.08&0.87&-0.04&0.87\\
       $\partial p_2/\partial x$&1.96&2.28&-0.40&1.92&-0.63&2.03&-0.41&1.93\\
     $\partial p_2/\partial w_1$&-1.16&1.42&0.06&1.15&0.15&1.16&0.06&1.15\\
     $\partial p_2/\partial w_2$&-0.96&1.78&0.09&1.74&0.17&1.75&0.09&1.74\\
       $\partial p_0/\partial x$&1.28&1.63&-0.34&1.47&-0.50&1.56&-0.34&1.48\\
     $\partial p_0/\partial w_1$&-1.16&1.42&0.06&1.15&0.15&1.16&0.06&1.15\\
     $\partial p_0/\partial w_2$&0.48&0.88&-0.05&0.87&-0.09&0.88&-0.05&0.87\\
      \midrule
      \multicolumn{9}{c}{$\tau = 1/2$}\\
      \midrule
       $\partial p_1/\partial x$&-8.97&9.04&-1.69&3.10&0.97&3.05&0.88&3.15\\
     $\partial p_1/\partial w_1$&6.44&6.62&1.39&2.81&-0.21&2.42&-0.15&2.48\\
     $\partial p_1/\partial w_2$&1.28&1.47&0.29&0.94&-0.06&0.92&-0.05&0.93\\
       $\partial p_2/\partial x$&5.22&5.31&1.05&2.16&-0.53&2.15&-0.48&2.20\\
     $\partial p_2/\partial w_1$&-3.21&3.30&-0.69&1.40&0.10&1.21&0.08&1.24\\
     $\partial p_2/\partial w_2$&-2.52&2.88&-0.56&1.87&0.13&1.84&0.11&1.85\\
       $\partial p_0/\partial x$&3.75&3.85&0.64&1.59&-0.44&1.65&-0.40&1.68\\
     $\partial p_0/\partial w_1$&-3.23&3.32&-0.70&1.41&0.10&1.21&0.08&1.24\\
     $\partial p_0/\partial w_2$&1.23&1.41&0.28&0.93&-0.07&0.92&-0.06&0.93\\
     \midrule
      \multicolumn{9}{c}{$\tau = 3/4$}\\
      \midrule
       $\partial p_1/\partial x$&-13.35&13.38&-6.83&7.32&0.71&3.29&0.71&3.29\\
     $\partial p_1/\partial w_1$&9.69&9.80&4.95&5.61&0.01&2.62&0.01&2.62\\
     $\partial p_1/\partial w_2$&1.81&1.94&1.01&1.35&-0.01&0.98&-0.01&0.98\\
       $\partial p_2/\partial x$&7.48&7.52&4.06&4.45&-0.37&2.35&-0.37&2.35\\
     $\partial p_2/\partial w_1$&-4.83&4.88&-2.47&2.81&-0.01&1.31&-0.01&1.31\\
     $\partial p_2/\partial w_2$&-3.51&3.76&-1.99&2.66&0.03&1.97&0.03&1.97\\
       $\partial p_0/\partial x$&5.87&5.92&2.77&3.14&-0.34&1.80&-0.34&1.80\\
     $\partial p_0/\partial w_1$&-4.87&4.93&-2.48&2.81&-0.01&1.31&-0.01&1.31\\
     $\partial p_0/\partial w_2$&1.70&1.82&0.98&1.31&-0.02&0.99&-0.02&0.99\\
        \bottomrule
    \end{tabular}
\begin{tablenotes}
\item \scriptsize This table reports the simulated finite sample bias and RMSE of the MLE and the MERM estimators
for the partial derivatives $\partial p_j(x,w,\theta_0)/\partial x$, $\partial p_j(x,w,\theta_0)/\partial w_1$,
$\partial p_j(x,w,\theta_0)/\partial w_2$ for $j \in \{1,2,0\}$ evaluated at the population mean. The true values
of the marginal effects are $(\partial p_1/\partial x, \partial p_2/\partial x, \partial p_0/\partial x) = (
0.222,   -0.111,   -0.111)$ and zeros for the rest. The results are based on 5,000 replications.
\end{tablenotes}
\end{threeparttable}
\end{scriptsize}
\end{table}

\begin{remark}[Inference after selecting $K$]
\label{rem:post-selection-inference}
As usual, inference after model selection requires careful treatment. Notably, Lemmas~\ref{lem:choice-of-K}
and~\ref{lem:iterated-choice-of-K} ensure that, under the stated conditions, the remaining EIV bias of the
selected estimator is asymptotically negligible, substantially simplifying subsequent inference.

Asymptotically valid inference on $\theta_0$ can be conducted by randomly splitting the sample into two
independent subsamples of asymptotically comparable sizes. One can use the first subsample to select $K$ using
Algorithm~\ref{alg: choose K}. Then, the second subsample is used to estimate the model with the selected order
$\hat K_n$ and conduct inference based on Theorem~\ref{the: asy normality}. In this sample-splitting procedure, by
construction, $\hat K_n$ is independent of the estimation sample and the selected estimator has asymptotically
negligible EIV bias. Thus Theorem~\ref{the: asy normality} implies asymptotic validity of the resulting inference.
\end{remark}

\subsection{\label{sec:Multivariate}Multiple Mismeasured Variables}

It is easy to extend the MERM framework to deal with multiple mismeasured variables. This is useful in many
applications, including not only settings with multiple mismeasured covariates, but also settings with serially
correlated measurement errors, settings where repeated measurements are available, and panel data models. Using
the MERM approach is particularly advantageous in such applications, since it avoids nonparametric estimation of
multivariate unobserved distributions.

Suppose $X_i^*$, $\varepsilon_{i}$, and $X_i$ are $d \times 1$ vectors. Let $\tau_n \equiv \max_{j\le d}
\sigma_{\varepsilon_j}/\sigma_{X^*_j}$, where $\sigma_{\varepsilon_j}$ and $\sigma_{X^*_j}$ denote the standard
deviations of the $j$-th components of $\varepsilon_i$ and $X_i^*$. We maintain the multivariate analogues of
Assumptions~\ref{ass: CME}
and~\ref{ass:MME}, with the moment bound in
Assumption~\ref{ass:MME}(ii) imposed componentwise. In particular,
$\ex{\abs{\varepsilon_{ij}}^K}=O(\tau_n^K)$.

For a $d \times 1$ vector of non-negative integers $\kappa = (\kappa_1,\ldots,\kappa_d)  \in \mathbb Z_+^{d}$, let
\begin{equation*}
    \partial_\kappa \equiv \frac{\partial^{\abs{\kappa}}}{\partial {x_{1}}^{\kappa_1} \ldots \partial
    {x_d}^{\kappa_d}},\quad \text{where } \abs{\kappa} \equiv \sum_{j=1}^d \kappa_j.
\end{equation*}
Also, for a positive integer $k$, let $\mathcal K_k = \{\kappa \in \mathbb Z_+^d: \abs{\kappa} = k\}$. Then, we
consider the following corrected moment function
\begin{equation*}
    \psi(x,s,\theta,\gamma) = g(x,s,\theta) - \sum_{k = 2}^{K} \sum_{\kappa \in \mathcal K_k } \gamma_\kappa
    \partial_{\kappa} g(x,s,\theta),
\end{equation*}
where, with some abuse of notation, $\gamma$ is a collection of all $\gamma_\kappa$ with $\kappa \in \mathcal K_k$
and $k \in \{2, \dots, K\}$.

Under mild smoothness conditions
\begin{equation*}
    \ex{\psi(X_i,S_i,\theta_0,\gamma_{0})} = \ex{g(X_i^*,S_i,\theta_0)} + O(\tau_n^{K+1}) = o(n^{-1/2}),
\end{equation*}
where the second equality holds provided that $O(\tau_n^{K+1}) = o(n^{-1/2})$. Similarly to the scalar case,
components of $\gamma_{0}$ are determined by the moments of $\varepsilon_i$. Specifically, let $\mu_{\kappa}
\equiv \ex{\varepsilon_{i 1}^{\kappa_1} \dots \varepsilon_{i d}^{\kappa_d}}$. Then
\begin{equation}
    \label{eq: kappa 2 and 3}
    \gamma_{0 \kappa} = \frac{\mu_{\kappa}}{\kappa!}, \quad \text{for } \kappa \in\mathcal K_2\cup\mathcal K_3,
\end{equation}
where $\kappa! \equiv \kappa_1! \ldots \kappa_d!$.
For $\abs{\kappa} \geq 4$, the coefficients can be computed by the following formulas. For example, for $\kappa
\in \mathcal K_4$, let $\mathcal K_{2, \kappa} = \{\tilde \kappa \in \mathcal K_2: \kappa - \tilde \kappa \in
\mathcal K_2 \}$. Then,
\begin{equation*}
    \gamma_{0 \kappa} = \frac{\mu_\kappa}{\kappa!}   - \sum_{\tilde \kappa \in \mathcal K_{2, \kappa}}
    \frac{\mu_{\kappa - \tilde \kappa}}{(\kappa - \tilde \kappa)!} \gamma_{0 \tilde \kappa} , \quad \text{for }
    \kappa \in \mathcal K_4.
\end{equation*}
More generally, for $\kappa \in \mathcal K_{k}$ with $k \geq 4$, let $\mathcal K_{\ell, \kappa} = \{\tilde \kappa
\in \mathcal K_\ell, \kappa - \tilde \kappa \in \mathcal K_{\abs{\kappa} - \ell} \}$ for $\ell \leq \abs{\kappa} -
2$. Then,
\begin{equation*}
    \gamma_{0 \kappa} = \frac{\mu_\kappa}{\kappa!}  - \sum_{\ell=2}^{k-2} \sum_{\tilde \kappa \in \mathcal
    K_{\ell, \kappa}} \frac{\mu_{\kappa - \tilde \kappa}}{(\kappa - \tilde \kappa)!} \gamma_{0 \tilde \kappa}.
\end{equation*}

\begin{example*}[Bivariate $X$, $K = 4$]\hfill \\
    Suppose $X$ is bivariate (i.e., $d = 2$) and $K = 4$. For $\kappa \in \mathcal K_2 = \{(2,0),(1,1),(0,2)\}$
    and $\kappa \in \mathcal K_3 = \{(3,0), (2,1), (1,2), (0,3)\}$, $\gamma_{0 \kappa}$ is given by \eqref{eq:
    kappa 2 and 3}. For $\kappa \in \mathcal K_4$, $\gamma_{0 \kappa}$ is given by
    \begin{center}
    \begin{tabular}{c|c}
        $\kappa$ & $\gamma_{0 \kappa}$ \\
        \midrule
        (4,0) & $\left(\e[{\varepsilon_{i1}^4}] - 6 \e[{\varepsilon_{i1}^2}]^2\right)/24$ \\
        (3,1) & $\left(\e[{\varepsilon_{i1}^3 \varepsilon_{i2}}] - 6 \e[{\varepsilon_{i1}^2}] \e[{\varepsilon_{i1} \varepsilon_{i2}}]\right)/6$ \\
        (2,2) & $\left(\e[{\varepsilon_{i1}^2 \varepsilon_{i2}^2}] - 2 \e[{\varepsilon_{i1}^2}] \e[{\varepsilon_{i2}^2}] - 4 \e[{\varepsilon_{i1} \varepsilon_{i2}}]^2\right)/4$  \\
        (1,3) & $\left(\e[{\varepsilon_{i1} \varepsilon_{i2}^3}] - 6 \e[{\varepsilon_{i2}^2}] \e[{\varepsilon_{i1} \varepsilon_{i2}}]\right)/6$ \\
        (0,4) & $\left(\e[{\varepsilon_{i2}^4}] - 6 \e[{\varepsilon_{i2}^2}]^2\right)/24$
    \end{tabular}
    \end{center}
    If in addition measurement errors $\varepsilon_{i1}$ and $\varepsilon_{i2}$ are independent, $\gamma_{0
    \kappa} = 0$ for $\kappa \in \{(1,1),(2,1),(1,2),(3,1),(1,3)\}$, and
    $\gamma_{0,(2,2)}=-\gamma_{0,(2,0)}\gamma_{0,(0,2)}$.
    Imposing these restrictions, the total number of nuisance parameters to be estimated is 6. \QEDA
\end{example*}

\begin{singlespacing}

\end{singlespacing}

\begin{appendices}
\numberwithin{equation}{section} \numberwithin{remark}{section} \numberwithin{theorem}{section}

\section{\label{sec:LargeSample}Regularity Conditions}

\noindent{\bf Notation.}
Let $\mathcal X \subseteq \mathbb R$ be some closed convex set containing the union of the supports of $X_i^*$ and
$X_i$, and $\mathcal S = \supp{S_i}$. Let $G(x,s,\theta)\equiv\nabla_\theta g(x,s,\theta)$, and use the
conventions $g_x^{(0)}=g$ and $G_x^{(0)}=G$.

\begin{assumption}
    \emph{(Moment function)}
    \label{ass: Moment function}
    Suppose that the moment restrictions \eqref{eq:g descr moment} are satisfied and the following conditions hold:
    \begin{enumerate}[(i)]
        \item \label{item: lipschitz moment} For all $s \in \mathcal S$ and $\theta \in \Theta$,
        $g_x^{(K)}(x,s,\theta)$ exists and is continuous on $\mathcal X$. Moreover, there exist functions
        $b_1,b_2: \mathcal X \times \mathcal S \times \Theta \rightarrow \mathbb R_+$ and an integer $M \geqslant
        K+1$ such that for all $x,x' \in \mathcal{X}$, $s \in \mathcal{S}$, and $\theta \in \Theta$,
        \begin{equation}
            \label{eq: lipschitz moment}
            \norm{g^{(K)}_x (x',s,\theta) - g^{(K)}_x (x,s,\theta)}\leqslant b_1 (x,s,\theta) |x' -
            x|+b_2(x,s,\theta) |x' - x|^{M-K};
        \end{equation}
        \item \label{item: M moments} Assumption \ref{ass:MME} holds with $L \geqslant M$;
        \item \label{item: basic moments} $\ex{g_x^{(k)}(X_i^*,S_i,\theta_0)}$, $k \in \{1,\dots,K\}$, and
        $\ex{b_j (X_i^*,S_i,\theta_0)}$, $j \in \{1,2\}$, exist and are bounded.
    \end{enumerate}
\end{assumption}

Assumption \ref{ass: Moment function} allows us to bound the remainder of the Taylor expansion of
$g(X_i,S_i,\theta)$ around $X_i^*$ by a polynomial in $\abs{X_i - X_i^*} = \abs{\varepsilon_{i}}$. Combined with
Assumption \ref{ass:MME} (which bounds the moments of $\varepsilon_{i}$), it ensures that this remainder is
$o(n^{-1/2})$, which is crucial for establishing validity of the corrected moment function $\psi$ (Lemma \ref{lem:
corrected moment}).

Notice that if $\mathcal X$ is compact, condition \eqref{eq: lipschitz moment} is satisfied if
$g_x^{(K+1)}(x,s,\theta)$ is bounded on $\mathcal X$ (for all $s \in \mathcal S$ and $\theta \in \Theta$). If
$\mathcal X$  is unbounded, condition \eqref{eq: lipschitz moment} is satisfied if for some $J$, such that $K < J
\leqslant M$, $\sup_{x \in \mathcal X} \norm{g_x^{(J)}(x,s,\theta)} \leqslant B(s,\theta)$ for some function
$B(s,\theta)$. Also notice that condition \eqref{eq: lipschitz moment} generalizes standard Lipschitz continuity
to accommodate the fact that in applications $\norm{g_x^{(K)}(x,s,\theta)}$ may behave like a polynomial in $x$
for large $x$.

\begin{assumption}
    \emph{(Parameter space)}
    \label{ass: Parameter space}
    \leavevmode
    \begin{enumerate}[(i)]
        \item \label{item: PS} $\Theta \subset \mathbb{R}^{\dim(\theta)}$ and $\Gamma \subset \mathbb{R}^{K-1}$
        are compact, $\theta_0 \in \inter{\Theta}$ and $\gamma_{0n} \in \Gamma$;
        \item \label{item: interior gamma} $0_{K-1} \in \inter{\Gamma}$.
    \end{enumerate}

\end{assumption}

\begin{assumption}
    \emph{(Regularity and smoothness conditions)}
    \label{ass: basic conditions}
    \leavevmode
    \begin{enumerate}[(i)]
        \item \label{item: lipschitz jacobian} For all $s \in \mathcal{S}$ and $k \in \{0,\ldots,K\}$,
        $G_x^{(k)}(x,s,\theta)$ exists and is continuous on
        $\mathcal X \times \Theta$; moreover, there exist functions $b_{G1},b_{G2}: \mathcal X \times \mathcal S
        \times \Theta \rightarrow \mathbb R_+$ and $\delta > 0$
        such that for all $x,x' \in \mathcal{X}$, $s \in \mathcal{S}$, and $\theta \in B_\delta(\theta_0)$
        \begin{equation*}
            \norm{G^{(K)}_x (x',s,\theta) - G^{(K)}_x (x,s,\theta)} \leqslant  b_{G1} (x,s,\theta) |x' -
            x|+b_{G2}(x,s,\theta) |x' - x|^{M-K}
        \end{equation*}

        \item \label{item: Global and Point Moments} $\ex{\norm{g_{x}^{(k)}(X_i^*,S_i,\theta_{0})}^2}$,
        $\ex{\sup_{\theta \in \Theta} \norm{g_{x}^{(k)}(X_i^*,S_i,\theta)}}$, for $k \in \{0,\dots,K\}$, and
        $\ex{{b_{j}(X_i^*,S_i,\theta_0)}^2}$, $\ex{\sup_{\theta \in \Theta} {b_{j}(X_i^*,S_i,\theta)}}$, for $j
        \in \{1,2\}$, are bounded;

        \item \label{item: Jacobian smoothness and Local Moments} for some $\delta > 0$, $\ex{\sup_{\theta \in
        B_{\delta}(\theta_0)} \norm{G_{x}^{(k)}(X_i^*,S_i,\theta)}}$, for $k \in \{0,\dots, K\}$, and
        $\ex{\sup_{\theta \in B_{\delta}(\theta_0)} b_{Gj} (X_i^*,S_i,\theta)}$, for $j \in \{1, 2\}$, are bounded;

        \item \label{item: moment covar} $\Omega_{gg}^* \equiv \ex{g(X_i^*,S_i,\theta_0) g(X_i^*,S_i,\theta_0)'}$
        is positive definite;

        \item \label{item: Wmatrix} $\hat \Xi \prob \Xi$, where $\Xi$ is a symmetric positive definite matrix;

        \item \label{item: 2M moments} Assumption \ref{ass:MME} holds with $L \geqslant 2M$.
    \end{enumerate}
\end{assumption}

Assumptions \ref{ass: Parameter space} and \ref{ass: basic conditions} include basic regularity conditions, which
help to ensure $\sqrt{n}$-consistency and asymptotic normality of the suggested estimator $\hat \theta$.
Specifically, Assumption \ref{ass: basic conditions}\eqref{item: lipschitz jacobian} is a counterpart of
Assumption \ref{ass: Moment function}\eqref{item: lipschitz moment} applied to the Jacobian function. It ensures
that the effect of the measurement error on the Jacobian is localized and allows us to establish $G \rightarrow
G^*$, so $\Psi \rightarrow \Psi^*$. As a result, local identification of $\theta_0$ and the asymptotic properties
of $\hat \theta$ are controlled by $G^*$ (and $\Psi^*$), the Jacobian associated with the correctly measured
variables (see Assumption \ref{ass: ID}\eqref{item: local ID}).

\section{\label{sec:Proof of Lemma 1}Proof of Lemma~\protect\ref{lem:
corrected moment}}

    To stress that in our asymptotic approximation the variance and the higher moments of $\varepsilon_i$ depend
    on $n$, we will use $\sigma_n^2 \equiv \ex{\varepsilon_i^2}$, $\gamma_{0n} \equiv \gamma_0$.

    Making use of Assumption \ref{ass: Moment function}\eqref{item: lipschitz moment}, we expand
    $g(X_i,S_i,\theta_0)$ around $X_i^*$ as
    \begin{align}
        g(X_i,S_i,\theta_0) = &g(X_i^*,S_i,\theta_0) + g_x^{(1)}(X_i^*,S_i,\theta_0) \varepsilon_i + \sum_{k=2}^K \frac{1}{k!} g_x^{(k)}(X_i^*,S_i,\theta_0) \varepsilon_i^k \nonumber \\ &+ \frac{1}{K!} \left(g_x^{(K)}(\tilde X_i, S_i, \theta_0) - g_x^{(K)}(X_i^*,S_i,\theta_0)\right) \varepsilon_i^K, \label{eq: A g exp}
    \end{align}
    where $\tilde X_i$ lies between $X_i^*$ and $X_i$ (and hereafter $\tilde X_i$ is allowed to be component
    specific). Similarly, for $k \in \{2,\dots,K\}$, we have
    \begin{align}
        g_x^{(k)}(X_i,S_i,\theta_0) = &g_x^{(k)}(X_i^*,S_i,\theta_0) + \sum_{\ell = k+1}^K \frac{1}{(\ell - k)!} g_x^{(\ell)}(X_i^*,S_i,\theta_0) \varepsilon_i^{\ell - k} \nonumber \\ &+ \frac{1}{(K-k)!} \left(g_x^{(K)}(\tilde X_{ki}, S_i, \theta_0) - g_x^{(K)}(X_i^*,S_i,\theta_0)\right) \varepsilon_i^{K-k}, \label{eq: A g k expansion}
    \end{align}
    where $\tilde X_{ki}$ lies between $X_i^*$ and $X_i$ for $k<K$, and $\tilde X_{Ki}=X_i$. Hence, combining
    these expressions and rearranging the terms, we obtain
    \begin{align}
        \psi(X_i,S_i,\theta_0,\gamma) = &g(X_i,S_i,\theta_0) - \sum_{k=2}^K \gamma_{k} g_x^{(k)} (X_i,S_i,\theta_0) \nonumber\\
        = &g(X_i^*,S_i,\theta_0) + g_x^{(1)}(X_i^*,S_i,\theta_0) \varepsilon_i \nonumber \\ &+ \sum_{k=2}^{K} g_x^{(k)}(X_i^*,S_i,\theta_0) \left( \dfrac{1}{k!} \varepsilon_i^k - \sum_{\ell=2}^k \dfrac{1}{(k-\ell)!} \varepsilon_i^{k-\ell} \gamma_{\ell} \right) \nonumber\\
        &+\dfrac{1}{K!} \left(g_x^{(K)}(\tilde X_{i},S_i,\theta_0) - g_x^{(K)}(X_{i}^*,S_i,\theta_0)\right) \varepsilon_i^K  \nonumber\\
        &- \sum_{k=2}^K \dfrac{\gamma_k}{(K-k)!} \left(g_x^{(K)} (\tilde X_{ki},S_i,\theta_0) - g_x^{(K)} (X_i^*,S_i,\theta_0)\right) \varepsilon_i^{K-k}. \label{eq: A psi expansions}
    \end{align}
    We want to show that for a properly chosen $\gamma = \gamma_{0n}$, $\ex{\psi(X_i,S_i,\theta_0,\gamma_{0n})} =
    o(n^{-1/2})$. Note that the first two terms in \eqref{eq: A psi expansions} are mean zero, i.e. we have
    \begin{align}
       \ex{g(X_i^*,S_i,\theta_0)} &= 0, \quad \ex{g_x^{(1)}(X_i^*,S_i,\theta_0) \varepsilon_i} = 0, \label{eq: A zero g and g eps}
    \end{align}
    where the latter is guaranteed by Assumption \ref{ass: CME}.

    Second, we argue that for a properly chosen $\gamma = \gamma_{0n}$, we have
    \begin{equation}
        \label{eq: A gamma system}
        \ex{\dfrac{1}{k!} \varepsilon_i^k - \sum_{\ell=2}^k \dfrac{1}{(k-\ell)!} \varepsilon_i^{k-\ell} \gamma_{0
        \ell n}} = 0,
    \end{equation}
    for all $k \in \{2, \dots, K\}$. Let us reparameterize $\gamma_{0n} = (\gamma_{02n}, \dots, \gamma_{0Kn})'$
    using $\gamma_{0kn} = \sigma_n^k a_{kn}$.\footnote{The case $\sigma_n=0$ is immediate, so we henceforth assume
    $\sigma_n>0$.} Then, \eqref{eq: A gamma system} can be rewritten as
    \begin{equation*}
            \e\left[\dfrac{1}{k!} (\varepsilon_i/\sigma_n)^k - \sum_{\ell=2}^k \dfrac{1}{(k-\ell)!}
            (\varepsilon_i/\sigma_n)^{k-\ell} a_{\ell n}\right] = 0,
    \end{equation*}
    which can also be represented as
    \begin{equation}
        \label{eq: A matrix system}
        B_n a_{n} = c_n
    \end{equation}
    where $a_{n} = (a_{2n}, \dots, a_{Kn})'$, and
    \begin{small}
    \begin{equation*}
        B_n = \begin{bmatrix}
        1 & 0 & \hdots & 0 & 0 \\
        \e[\varepsilon_i/\sigma_n] & 1 & \hdots & 0 & 0 \\
        \vdots & \vdots & \ddots & \vdots & \vdots \\
        \frac{\e[(\varepsilon_i/\sigma_n)^{K-3}]}{(K-3)!}  & \frac{\e[(\varepsilon_i/\sigma_n)^{K-4}]}{(K-4)!} & \hdots & 1 & 0 \\
        \frac{\e[(\varepsilon_i/\sigma_n)^{K-2}]}{(K-2)!}  & \frac{\e[(\varepsilon_i/\sigma_n)^{K-3}]}{(K-3)!} & \hdots & \e[(\varepsilon_i/\sigma_n)] & 1 \\
        \end{bmatrix}, \quad
        c_n = \begin{bmatrix}
        \e [(\varepsilon_i/\sigma_n)^2] / 2! \\
        \e [(\varepsilon_i/\sigma_n)^3] / 3! \\
        \vdots \\
        \e [(\varepsilon_i/\sigma_n)^{K-1}] / (K-1)! \\
        \e [(\varepsilon_i/\sigma_n)^K] / K!
        \end{bmatrix}.
    \end{equation*}
    \end{small}
    Since $B_n$ is invertible, \eqref{eq: A matrix system} has a unique solution $a_n = B_n^{-1} c_n$. Moreover,
    $a_n$ is bounded since both $B_n^{-1}$ and $c_n$ are bounded (Assumption \ref{ass:MME}). Hence, we conclude
    that $\eqref{eq: A gamma system}$ has a unique solution $\gamma_{0n} = \left(\sigma_n^2 a_{2n}, \dots,
    \sigma_n^K a_{Kn}\right)'$. Since \eqref{eq: A gamma system} is satisfied, using Assumption \ref{ass: CME}, we
    also conclude that
    \begin{equation}
        \label{eq: A proper gamma}
        \ex{\sum_{k=2}^{K} g_x^{(k)}(X_i^*,S_i,\theta_0) \left( \dfrac{1}{k!} \varepsilon_i^k - \sum_{\ell=2}^k
        \dfrac{1}{(k-\ell)!} \varepsilon_i^{k-\ell} \gamma_{0\ell n} \right)} = 0.
    \end{equation}
    To complete the proof of $\ex{\psi(X_i,S_i,\theta_0,\gamma_{0n})} = o(n^{-1/2})$, it is sufficient to show that
    \begin{align}
         \ex{\left(g_x^{(K)}(\tilde X_{i},S_i,\theta_0) - g_x^{(K)}(X_{i}^*,S_i,\theta_0)\right) \varepsilon_i^K} &= o(n^{-1/2}), \label{eq: A delta g_K 1}\\
        \gamma_{0kn} \ex{\left(g_x^{(K)} (\tilde X_{ki},S_i,\theta_0) - g_x^{(K)} (X_i^*,S_i,\theta_0)\right) \varepsilon_i^{K-k}} &= o(n^{-1/2}) \label{eq: A delta g_K 2}
    \end{align}
    for $k\in\{2, \dots, K\}$. We start with \eqref{eq: A delta g_K 1}. Using Assumption \ref{ass: Moment
    function}\eqref{item: lipschitz moment}, we obtain
    \begin{equation}
        \label{eq: A delta g bound}
        \norm{\left(g_x^{(K)}(\tilde X_{i},S_i,\theta_0) - g_x^{(K)}(X_{i}^*,S_i,\theta_0)\right) \varepsilon_i^K}
        \leqslant b_1 (X_i^*,S_i,\theta_0) \abs{\varepsilon_i}^{K+1} + b_2(X_i^*,S_i,\theta_0)
        \abs{\varepsilon_i}^{M}.
    \end{equation}
    Hence, using Assumption \ref{ass: CME}, and the fact $\abs{\tilde X_i - X_i^*} \leqslant \abs{\varepsilon_i}$,
    we get
    \begin{align*}
        &\left \Vert \ex{\left(g_x^{(K)}(\tilde X_{i},S_i,\theta_0) - g_x^{(K)}(X_{i}^*,S_i,\theta_0)\right) \varepsilon_i^K} \right \Vert \\ &\leqslant \sigma_n^{K+1} {\ex{b_1 (X_i^*,S_i,\theta_0)} \ex{\abs{\varepsilon_i/\sigma_n}^{K+1}}} + \sigma_n^{M} {\ex{b_2(X_i^*,S_i,\theta_0)} \ex{\abs{\varepsilon_i/\sigma_n}^{M}}}.
    \end{align*}
    Since (i) the expectations above are bounded (Assumptions \ref{ass:MME}, \ref{ass: Moment
    function}\eqref{item: M moments}, and \ref{ass: Moment function}\eqref{item: basic moments}) and (ii)
    $\sigma_n^{K+1} = o(n^{-1/2})$ and $\sigma_n^{M} = o(n^{-1/2})$ (Assumption \ref{ass:MME}), this implies that
    \eqref{eq: A delta g_K 1} holds. To inspect \eqref{eq: A delta g_K 2}, recall that $\gamma_{0kn} = \sigma_n^k
    a_{kn}$. As a result, using Assumptions \ref{ass: Moment function}\eqref{item: lipschitz moment} and \ref{ass:
    CME}, and $\abs{\tilde X_{ki} - X_i^*} \leqslant \abs{\varepsilon_i}$ again, we also have
    \begin{align*}
        &\abs{\gamma_{0kn}} \left \Vert \ex{\left(g_x^{(K)} (\tilde X_{ki},S_i,\theta_0) - g_x^{(K)} (X_i^*,S_i,\theta_0)\right) \varepsilon_i^{K-k}} \right \Vert \\ &\leqslant \abs{a_{kn}} \left(\sigma_n^{K+1} {\ex{b_1 (X_i^*,S_i,\theta_0)} \ex{\abs{\varepsilon_i/\sigma_n}^{K+1-k}}} + \sigma_n^{M} {\ex{b_2(X_i^*,S_i,\theta_0)} \ex{\abs{\varepsilon_i/\sigma_n}^{M-k}}}\right).
    \end{align*}
    Since $\abs{a_{kn}}$ is bounded, we conclude that $\eqref{eq: A delta g_K 2}$ holds analogously to \eqref{eq:
    A delta g_K 1}.

    Combining \eqref{eq: A psi expansions} with \eqref{eq: A zero g and g eps}, and \eqref{eq: A proper
    gamma}-\eqref{eq: A delta g_K 2}, we conclude that $\ex{\psi(X_i,S_i,\theta_0,\gamma_{0n})} = o(n^{-1/2})$.

    Finally, we want to verify the recursive expressions for the components of $\gamma_{0n}$ using \eqref{eq: A
    gamma system}. First, $\gamma_{02n} = \ex{\varepsilon_i^2}/2$ and $\gamma_{03n} = \ex{\varepsilon_i^3}/6$
    (since $\ex{\varepsilon_i} = 0$). For $k \geqslant 4$, suppose that $\gamma_{0\ell n}$ are known for $\ell \in
    \{2, \dots, k-1\}$. Then $\gamma_{0kn}$ can be
    directly computed from \eqref{eq: A gamma system}:
    \begin{equation*}
        \sum_{\ell = 2}^k \frac{\ex{\varepsilon_i^{k - \ell}}}{(k - \ell)!} \gamma_{0 \ell n}   =
        \frac{\ex{\varepsilon_i^{k}}}{k!}.
    \end{equation*}
    Plugging in $\ex{\varepsilon_i} = 0$ and rearranging the terms give the expression in \eqref{eq: gammas}.
    \hfill Q.E.D.

\section{\label{sec:Proofs for LargeSample}Proof of Theorem \protect\ref
{the: asy normality}}

\noindent{\bf Notation.} To stress that in our asymptotic approximation the variance and the higher moments of
$\varepsilon_i$ depend on $n$, we will use $\sigma_n^2 \equiv \ex{\varepsilon_i^2}$, $\gamma_{0n} \equiv
\gamma_0$, and $\beta_{0n} \equiv \beta_0 \equiv (\theta_0',\gamma_{0n}')'$.

All vectors are columns. For some generic parameter vector $\alpha$ and a vector (or matrix) valued function
$a(x,s, \alpha )$, let $a_{i}(\alpha )\equiv a(X_{i},S_{i},\alpha )$, $\overline{a}(\alpha )\equiv
n^{-1}\sum_{i=1}^{n}a_{i}(\alpha )$, $a(\alpha )\equiv \mathbb{E}[a_{i}(\alpha)]$. Similarly, we let
$a_{i}^*(\alpha)\equiv a(X_{i}^*,S_{i},\alpha )$, $\overline{a}^*(\alpha )\equiv
n^{-1}\sum_{i=1}^{n}a_{i}^*(\alpha )$, $a^*(\alpha ) \equiv \mathbb{E}[a_{i}^*(\alpha)]$.

For the true value of the parameter $\alpha_{0}$, we often write $a_i \equiv a_i (\alpha_0)$, $\overline{a}\equiv
\overline{a}(\alpha_0)$, $a\equiv a(\alpha_0)$, $a_i^* \equiv a_i^*(\alpha_0)$, $\overline{a}^*\equiv
\overline{a}^*(\alpha_0)$, $a^*\equiv a^*(\alpha_0)$.

\subsection{\label{ssec:aux lemmas}Auxiliary Lemmas}
\begin{lemma}
    \label{lem: moments UC}
    Suppose that $\{(X_i^*,S_i^{\prime},\varepsilon_i)\}_{i=1}^n$ are i.i.d. Then, under Assumptions
    \ref{ass:MME}, \ref{ass: CME}, \ref{ass: Moment function}, \ref{ass: Parameter space}\eqref{item: PS}, and
    \ref{ass: basic conditions}\eqref{item: lipschitz jacobian}-\eqref{item: Jacobian smoothness and Local
    Moments}, we have
    \begin{enumerate}[(i)]
        \item
        \begin{equation*}
            \sup_{\theta \in \Theta} \norm{ \overline g_x^{(k)}(\theta) - g_x^{(k)*} (\theta)} = o_p(1)
        \end{equation*}
        and $g_x^{(k)*}(\theta)$ is continuous on $\Theta$ for $k \in \{0, \dots, K\}$;
        \item for some $\delta > 0$,
        \begin{equation*}
            \sup_{\theta \in B_\delta(\theta_0)} \norm{ \overline G_x^{(k)}(\theta) - G_x^{(k)*} (\theta)} = o_p(1),
        \end{equation*}
        and $G_x^{(k)*}(\theta)$ is continuous on $B_\delta(\theta_0)$ for $k \in \{0, \dots, K\}$.
    \end{enumerate}
\end{lemma}

\begin{proof}[Proof of Lemma \ref{lem: moments UC}]
    First, we show
    \begin{equation*}
        \sup_{\theta \in \Theta} \norm{ \overline g(\theta) - g^* (\theta)} = o_p(1).
    \end{equation*}
    By the triangle inequality,
    \begin{equation*}
        \sup_{\theta \in \Theta} \norm{ \overline g(\theta) - g^* (\theta)} \leqslant \sup_{\theta \in \Theta}
        \norm{ \overline g(\theta) - \overline g^* (\theta)} + \sup_{\theta \in \Theta} \norm{ \overline
        g^*(\theta) - g^* (\theta)}.
    \end{equation*}
    Then, it is sufficient to show that both terms on the right hand side of the inequality above are $o_p(1)$.
    Expanding $g(X_i,S_i,\theta)$ around $X_i^*$ as in \eqref{eq: A g exp} and invoking Assumption \ref{ass:
    Moment function}\eqref{item: lipschitz moment},
    \begin{align*}
      \sup_{\theta \in \Theta} \norm{\overline g (\theta) - \overline g^* (\theta)}
      = & \sup_{\theta \in \Theta} \norm{\sum_{k=1}^{K-1} \frac{1}{k!} \frac{1}{n} \sumin  g_x^{(k)}(X_i^*,S_i,\theta) \varepsilon_i^k + \dfrac{1}{K!} \frac{1}{n} \sumin  g_x^{(K)}(\tilde X_{i},S_i,\theta) \varepsilon_i^K }  \\
       \leqslant &\sum_{k=1}^{K} \frac{1}{k!} \underbrace{\frac{1}{n} \sumin \sup_{\theta \in \Theta} \norm{g_x^{(k)}(X_i^*,S_i,\theta)} |\varepsilon_i|^k}_{o_p(1)}\\ &+ \dfrac{1}{K!} \underbrace{\frac{1}{n} \sumin \sup_{\theta \in \Theta} b_1(X_i^*,S_i,\theta) |\varepsilon_i|^{K+1}}_{o_p(1)} \\ &+ \dfrac{1}{K!} \underbrace{\frac{1}{n} \sumin \sup_{\theta \in \Theta} b_2(X_i^*,S_i,\theta) |\varepsilon_i|^{M}}_{o_p(1)},
    \end{align*}
    where $\tilde X_{i}$ lies between $X_i^*$ and $X_i$ (note that intermediate points $\tilde X_i$ may depend on
    \(\theta\) and differ across components). Now observe that all the terms following the inequality sign are
    $o_p(1)$. Indeed, this is guaranteed by Markov's inequality paired with Assumptions \ref{ass:MME}, \ref{ass:
    CME}, and \ref{ass: basic conditions}\eqref{item: Global and Point Moments}. Hence, $\sup_{\theta \in \Theta}
    \norm{\overline g (\theta) - \overline g^* (\theta)} = o_p(1)$, and we are left to show $\sup_{\theta \in
    \Theta} \norm{ \overline g^*(\theta) - g^* (\theta)} = o_p(1)$. This, in turn, follows from the standard ULLN
    (e.g., Lemma 2.4 in \citealp{NeweyMcFadden1994HB}), which also ensures continuity of $g^*(\theta)$ on
    $\Theta$. Hence, we conclude that the assertion of the lemma holds for $g$.

    Applying nearly identical arguments, one can also establish the desired results for $g_x^{(k)}$ for $k \in
    \{1, \dots, K\}$ and for $G_x^{(k)}$ for $k \in \{0, \dots, K\}$ (for the latter, Assumptions~\ref{ass: basic
    conditions}\eqref{item: lipschitz jacobian} and \eqref{item: Jacobian smoothness and Local Moments} take the
    places of Assumptions \ref{ass: Moment function}\eqref{item: lipschitz moment} and \ref{ass: basic
    conditions}\eqref{item: Global and Point Moments}, respectively).
\end{proof}

\begin{lemma}
    \label{lem: consistency of objects}
    Suppose that the hypotheses of Lemma \ref{lem: moments UC} are satisfied. Then, $g_x^{(k)}  \rightarrow
    g_x^{(k)*}$ and $G_x^{(k)} \rightarrow G_x^{(k)*}$ for $k \in \{0,\dots,K\}$. Suppose also $\hat \theta \prob
    \theta_0$. Then, $\overline g_x^{(k)} (\hat \theta) \prob g_x^{(k)*}$ and $\overline G_x^{(k)} (\hat \theta)
    \prob G_x^{(k)*}$ for $k \in \{0,\dots,K\}$.
\end{lemma}

\begin{proof}[Proof of Lemma \ref{lem: consistency of objects}]
    First, we prove the assertions of the lemma for $g_x^{(k)}$. Note that, by the standard expansion of
    $g_x^{(k)}(X_i,S_i,\theta_0)$ around $X_i^*$ (see Eq. \eqref{eq: A g k expansion} above), we have
    \begin{align*}
        \norm{g_x^{(k)} - g_x^{(k)^*}} \leqslant &\ex{\norm{g_x^{(k)}(X_i,S_i,\theta_0) - g_x^{(k)}(X_i^*,S_i,\theta_0)}} \\
         \leqslant &\sum_{\ell = k+1}^K \frac{1}{(\ell - k)!} \ex{\norm{g_x^{(\ell)}(X_i^*,S_i,\theta_0)} \abs{\varepsilon_i}^{\ell - k}} \\ &+ \frac{1}{(K-k)!} \ex{\norm{\left(g_x^{(K)}(\tilde X_{ki}, S_i, \theta_0) - g_x^{(K)}(X_i^*,S_i,\theta_0)\right)} \abs{\varepsilon_i}^{K-k}}.
    \end{align*}
    By Assumptions \ref{ass:MME}, \ref{ass: CME}, and \ref{ass: basic conditions}\eqref{item: Global and Point
    Moments}, $\ex{\norm{g_x^{(\ell)}(X_i^*,S_i,\theta_0)} \abs{\varepsilon_i}^{\ell-k}} \rightarrow 0$ for all
    $\ell \in \{k+1, \dots, K\}$. Next, using Assumptions \ref{ass: Moment function}\eqref{item: lipschitz moment}
    and \ref{ass: CME},
    \begin{align*}
         &\ex{\norm{\left(g_x^{(K)}(\tilde X_{ki}, S_i, \theta_0) - g_x^{(K)}(X_i^*,S_i,\theta_0)\right)} \abs{\varepsilon_i}^{K-k}} \\ &\leqslant {\ex{b_1 (X_i^*,S_i,\theta_0)} \ex{\abs{\varepsilon_i}^{K+1-k}}} +  {\ex{b_2(X_i^*,S_i,\theta_0)} \ex{\abs{\varepsilon_i}^{M-k}}} \rightarrow 0,
    \end{align*}
    where the convergence follows from Assumptions \ref{ass:MME}, \ref{ass: Moment function}\eqref{item: M
    moments} and \ref{ass: basic conditions}\eqref{item: Global and Point Moments}. Hence, we conclude $g_x^{(k)}
    \rightarrow g_x^{(k)*}$.

    Next, we show $\overline g_x^{(k)} (\hat \theta) \prob g_x^{(k)*}$. By the triangle inequality,
    \begin{equation*}
        \norm{\overline g_x^{(k)} (\hat \theta) - g_x^{(k)*}} \leqslant \sup_{\theta \in B_\delta (\theta_0)}
        \norm{\overline g_x^{(k)} (\theta) - g_x^{(k)*}(\theta)} + \norm{g_x^{(k)*} (\hat \theta) - g_x^{(k)*}
        (\theta_0)},
    \end{equation*}
    where the inequality holds with probability approaching one since $\hat \theta \in B_\delta (\theta_0)$ with
    probability approaching one. Note that, by Lemma \ref{lem: moments UC}, $\sup_{\theta \in B_\delta (\theta_0)}
    \norm{\overline g_x^{(k)} (\theta) - g_x^{(k)*}(\theta)} = o_p(1)$ and $\norm{g_x^{(k)*} (\hat \theta) -
    g_x^{(k)*} (\theta_0)} = o_p(1)$, where the second result follows from consistency of $\hat \theta$ and
    continuity of $g_x^{(k)*}(\theta)$. Hence, $\overline g_x^{(k)} (\hat \theta) \prob g_x^{(k)*}$, which
    completes the proof of the results for $g_x^{(k)}$ for all $k \in \{0, \dots, K\}$.

    A nearly identical argument can be invoked to establish the same results for $G_x^{(k)}$ for $k \in \{0,
    \dots, K\}$, with Assumptions \ref{ass: basic conditions}\eqref{item: lipschitz jacobian} and \eqref{item:
    Jacobian smoothness and Local Moments} taking the places of Assumptions \ref{ass: Moment function}\eqref{item:
    lipschitz moment} and \ref{ass: basic conditions}\eqref{item: Global and Point Moments}, respectively.
\end{proof}

\begin{lemma}
    \label{lem: consistency}
    Suppose that the hypotheses of Lemma \ref{lem: moments UC} are satisfied. Then, under additional Assumptions
    \ref{ass: basic conditions}\eqref{item: Wmatrix} and \ref{ass: ID}\eqref{item: global ID}, we have $\hat
    \theta \prob \theta_0$, $\hat \gamma \prob 0$ and $\hat \gamma \prob \gamma_{0n}$.
\end{lemma}

\begin{proof}[Proof of Lemma \ref{lem: consistency}]
    First, we argue that $\sup_{\beta \in \mathcal B} \norm{\overline \psi(\beta) - \psi^*(\beta)} = o_p(1)$.
    Notice that, by the triangle inequality,
    \begin{equation}
        \sup_{\beta \in \mathcal B} \norm{\overline \psi(\beta) - \psi^*(\beta)} \leqslant \sup_{\theta \in
        \Theta} \norm{\overline g(\theta) - g^*(\theta)} + \sum_{k=2}^K \sup_{\gamma \in \Gamma} \abs{\gamma_k}
        \sup_{\theta \in \Theta} \norm{\overline g_x^{(k)}(\theta) - g_x^{(k)*}(\theta)} =o_p(1), \label{eq: A psi UC}
    \end{equation}
    where the equality follows from Lemma \ref{lem: moments UC}(i) and boundedness of $\gamma$ (Assumption
    \ref{ass: Parameter space}\eqref{item: PS}). Moreover, Lemma \ref{lem: moments UC}(i) also ensures that
    $\psi^*(\beta)$ is continuous on compact $\mathcal B$ and, consequently, is bounded.

    Let $\hat Q(\beta) = \overline \psi (\beta)' \hat \Xi \overline \psi (\beta)$ and $Q^*(\beta) = \psi^*
    (\beta)' \Xi \psi^* (\beta)$. Notice that \eqref{eq: A psi UC}, boundedness of $\psi^*(\beta)$, and Assumption
    \ref{ass: basic conditions}\eqref{item: Wmatrix} together guarantee that $\sup_{\beta \in \mathcal B}
    \abs{\hat Q(\beta) - Q^*(\beta)} = o_p(1)$. Next, recall that $\gamma_{0n} \rightarrow 0_{K-1}$ (Lemma
    \ref{lem: corrected moment}). Since $\Gamma$ is compact and $\gamma_{0n} \in \Gamma$ (Assumption \ref{ass:
    Parameter space}\eqref{item: PS}), $0_{K-1} \in \Gamma$. Consequently, Assumptions~\ref{ass: ID}\eqref{item:
    global ID} and \ref{ass: basic conditions}\eqref{item: Wmatrix} together guarantee that $Q^* (\beta)$ is
    uniquely minimized at $\theta = \theta_{0}$ and $\gamma = 0_{K-1}$. Consequently, applying the standard
    consistency argument (e.g., Theorem 2.1 of \citealp{NeweyMcFadden1994HB}), we conclude that $\hat \theta
    \rightarrow \theta_0$ and $\hat \gamma \rightarrow 0_{K-1}$. Finally, since $\gamma_{0n} \rightarrow 0$ (Lemma
    \ref{lem: corrected moment}), we also have $\hat \gamma \prob \gamma_{0n}$.
\end{proof}

\begin{lemma}
    \label{lem: moment normality}
    Suppose that $\{(X_i^*,S_i^{\prime},\varepsilon_i)\}_{i=1}^n$ are i.i.d. Then, under Assumptions
    \ref{ass:MME}, \ref{ass: CME}, \ref{ass: Moment function}, and \ref{ass: basic conditions}\eqref{item: Global
    and Point Moments} and \eqref{item: 2M moments}, we have
    \begin{equation*}
        n^{1/2} \overline \psi (\beta_{0n}) \dist N(0,\Omega_{gg}^*),
    \end{equation*}
    where $\Omega_{gg}^* \equiv \ex{g(X_i^*,S_i,\theta_0) g(X_i^*,S_i,\theta_0)'}$.
\end{lemma}

\begin{proof}[Proof of Lemma \ref{lem: moment normality}]
    Using expansion \eqref{eq: A psi expansions}, we obtain
    \begin{align}
        n^{1/2} \overline \psi(\beta_{0n}) = & n^{-1/2} \sumin g(X_i^*,S_i,\theta_0) + n^{-1/2} \sumin g_x^{(1)}(X_i^*,S_i,\theta_0) \varepsilon_i \nonumber \\ &+  \sum_{k=2}^{K} n^{-1/2} \sumin g_x^{(k)}(X_i^*,S_i,\theta_0) \left( \dfrac{1}{k!} \varepsilon_i^k - \sum_{\ell=2}^k \dfrac{1}{(k-\ell)!} \varepsilon_i^{k-\ell} \gamma_{0 \ell n} \right) \nonumber\\
        &+\dfrac{1}{K!} n^{-1/2} \sumin \left(g_x^{(K)}(\tilde X_{i},S_i,\theta_0) - g_x^{(K)}(X_{i}^*,S_i,\theta_0)\right) \varepsilon_i^K  \nonumber\\
        &- \sum_{k=2}^K \dfrac{\gamma_{0kn}}{(K-k)!} n^{-1/2} \sumin  \left(g_x^{(K)} (\tilde X_{ki},S_i,\theta_0) - g_x^{(K)} (X_i^*,S_i,\theta_0)\right) \varepsilon_i^{K-k}. \label{eq: A sample psi exp}
    \end{align}
    First, note that, by the standard CLT, $n^{-1/2} \sumin g(X_i^*,S_i,\theta_0) \dist N(0,\Omega_{gg}^*)$. The
    rest of the proof is to show that the remaining terms are $o_p(1)$. By Assumptions \ref{ass:MME}, \ref{ass:
    CME}, \ref{ass: basic conditions}\eqref{item: Global and Point Moments}, Chebyshev's inequality guarantees
    \begin{equation*}
        n^{-1/2} \sumin g_x^{(1)} (X_i^*,S_i,\theta_{0}) \varepsilon_i = o_p(1)
    \end{equation*}
    Next, \eqref{eq: A gamma system} ensures that we can similarly apply Chebyshev's inequality (combined with
    Assumptions \ref{ass:MME}, \ref{ass: CME}, \ref{ass: basic conditions}\eqref{item: Global and Point Moments}
    and \eqref{item: 2M moments}) to ensure that for $k \in \{2, \ldots, K\}$
    \begin{equation*}
        n^{-1/2} \sumin g_x^{(k)}(X_i^*,S_i,\theta_0) \left( \dfrac{1}{k!} \varepsilon_i^k - \sum_{\ell=2}^k
        \dfrac{1}{(k-\ell)!} \varepsilon_i^{k-\ell} \gamma_{0 \ell n} \right) = o_p(1).
    \end{equation*}
    Next, using \eqref{eq: A delta g bound},
    \begin{align*}
        &\norm{n^{-1/2} \sumin \left(g_x^{(K)}(\tilde X_{i},S_i,\theta_0) - g_x^{(K)}(X_{i}^*,S_i,\theta_0)\right) \varepsilon_i^K} \\ \leqslant &n^{-1/2} \sumin  b_1 (X_i^*,S_i,\theta_0) \abs{\varepsilon_i}^{K+1} + n^{-1/2} \sumin b_2(X_i^*,S_i,\theta_0) \abs{\varepsilon_i}^{M} \\
        \leqslant &\underbrace{n^{1/2} \sigma_n^{K+1}}_{\rightarrow 0}  \underbrace{\left(n^{-1} \sumin  b_1 (X_i^*,S_i,\theta_0) \abs{\varepsilon_i/\sigma_n}^{K+1}\right)}_{O_p(1)} \\ &+ \underbrace{n^{1/2} \sigma_n^M}_{\rightarrow 0} \underbrace{\left(n^{-1} \sumin b_2(X_i^*,S_i,\theta_0) \abs{\varepsilon_i/\sigma_n}^{M}\right)}_{O_p(1)} = o_p(1),
    \end{align*}
    where both $n^{1/2} \sigma_n^{K+1}$ and $n^{1/2} \sigma_n^M$ converge to zero by Assumption \ref{ass:MME}, and
    the terms
    in the brackets are $O_p(1)$ by Markov's inequality (ensured by Assumptions \ref{ass:MME}, \ref{ass: CME},
    \ref{ass: Moment function}\eqref{item: M moments} and \ref{ass: basic conditions}\eqref{item: Global and Point
    Moments}). Recall that in the proof of Lemma \ref{lem: corrected moment}, we have demonstrated that
    $\gamma_{0kn} = \sigma_n^k a_{kn}$, where $a_{kn}$ are bounded, for $k \in \{2, \dots, K\}$. Hence, similarly,
    we have
    \begin{align*}
        &\norm{\gamma_{0kn} n^{-1/2} \sumin \left(g_x^{(K)}(\tilde X_{ki},S_i,\theta_0) - g_x^{(K)}(X_{i}^*,S_i,\theta_0)\right) \varepsilon_i^{K-k}} \\ \leqslant &\vert a_{kn} \vert \sigma_n^k \left[n^{-1/2} \sumin  b_1 (X_i^*,S_i,\theta_0) \abs{\varepsilon_i}^{K-k+1} + n^{-1/2} \sumin b_2(X_i^*,S_i,\theta_0) \abs{\varepsilon_i}^{M-k}\right] \\
        \leqslant & \vert a_{kn} \vert \underbrace{ n^{1/2} \sigma_n^{K+1}}_{\rightarrow 0}  \underbrace{\left(n^{-1} \sumin  b_1 (X_i^*,S_i,\theta_0) \abs{\varepsilon_i/\sigma_n}^{K-k+1}\right)}_{O_p(1)} \\ &+ \vert a_{kn} \vert \underbrace{n^{1/2} \sigma_n^M}_{\rightarrow 0} \underbrace{\left(n^{-1} \sumin b_2(X_i^*,S_i,\theta_0) \abs{\varepsilon_i/\sigma_n}^{M-k}\right)}_{O_p(1)} = o_p(1).
    \end{align*}
    Hence, we have demonstrated that all the remaining terms in \eqref{eq: A sample psi exp} are $o_p(1)$, i.e. we
    have
    \begin{align*}
        n^{1/2} \overline \psi (\beta_{0n}) = &n^{-1/2} \sumin g(X_i^*,S_i,\theta_0) + o_p(1) \dist N(0,\Omega_{gg}^*),
    \end{align*}
    which completes the proof.
\end{proof}

\subsection{Proof of Theorem \ref{the: asy normality}}
Equipped with Lemmas \ref{lem: moments UC}-\ref{lem: moment normality}, we are ready to prove Theorem \ref{the:
asy normality}.

    Since (i) $\hat \theta$ and $\hat \gamma$ are consistent for $\theta_0$ and $\gamma_{0n}$, respectively (Lemma
    \ref{lem: consistency}) and (ii) both $\theta_0$ and $\gamma_{0n}$
	(with $\gamma_{0n}\to0$ by Assumption~\ref{ass:MME})
	are bounded away from the boundaries of $\Theta$ and $\Gamma$,
	respectively, the standard GMM FOC is satisfied with probability approaching one, i.e., we have (with probability
	approaching one)
    \begin{equation*}
        \overline \Psi (\hat \beta)' \hat \Xi \overline \psi (\hat \beta) = 0.
    \end{equation*}
    Expanding $\overline \psi (\hat \beta)$ around $\overline \psi (\beta_{0n})$ gives
    \begin{equation}
        \label{eq: A FOC}
        \overline \Psi (\hat \beta)' \hat \Xi \left(\overline \psi (\beta_{0n}) + \overline \Psi(\tilde \beta)
        (\hat \beta - \beta_0) \right) = 0,
    \end{equation}
where $\overline\Psi(\tilde\beta)$ denotes the Jacobian evaluated
row by row at possibly different points $\tilde\beta$ on the line
segment between $\beta_{0n}$ and $\hat\beta$ (so
$\tilde\theta\prob\theta_0$ and $\tilde\gamma\prob0$ for each row). Observe
    \begin{equation*}
        \overline \Psi (\hat \beta) = \left[\overline G(\hat \theta) - \sum_{k=2}^K \hat \gamma_k \overline
        G_x^{(k)}(\hat \theta), - \overline g_x^{(2)} (\hat \theta), \dots,  - \overline g_x^{(K)} (\hat \theta)
        \right].
    \end{equation*}
    Since $\hat \theta \prob \theta_0$ (Lemma \ref{lem: consistency}), we can invoke the result of Lemma \ref{lem:
    consistency of objects} to argue that $\overline g_x^{(k)} (\hat \theta) \prob g_x^{(k)*}$ and $\overline
    G_x^{(k)} (\hat \theta) \prob G_x^{(k)*}$ for all $k \in \{0, \dots, K\}$. This, combined with $\hat \gamma
    \rightarrow 0$ (Lemma \ref{lem: consistency}), ensures that $\overline \Psi (\hat \beta) = \Psi^* + o_p(1)$
    and, by the same argument applied row by row, $\overline \Psi (\tilde \beta) = \Psi^* + o_p(1)$. Coupling
    these results with Assumption \ref{ass: basic conditions}\eqref{item: Wmatrix}, we conclude that $\overline
    \Psi (\hat \beta)' \hat \Xi \overline \Psi (\tilde \beta) \prob \Psi^{* \prime} \Xi \Psi^*$, which is
    invertible by Assumption \ref{ass: ID}\eqref{item: local ID}. Hence, \eqref{eq: A FOC} can be rearranged as
    (with probability approaching one)
    \begin{align*}
        n^{1/2} (\hat \beta - \beta_{0n}) &= - \left(\overline \Psi (\hat \beta)' \hat \Xi \overline \Psi (\tilde \beta) \right)^{-1} \overline \Psi (\hat \beta)' \hat \Xi n^{1/2} \overline \psi (\beta_{0n}) \\
        &= - \left( \Psi^{* \prime} \Xi \Psi^* \right)^{-1} \Psi^{* \prime} \Xi  n^{1/2}  \overline \psi (\beta_{0n}) + o_p(1),
    \end{align*}
    where, by Lemma \ref{lem: moment normality}, $n^{1/2} \overline \psi (\beta_{0n}) \dist N(0,\Omega_{gg}^*)$.
    Hence, we conclude
    \begin{equation*}
        n^{1/2} (\hat \beta - \beta_{0n}) \dist N(0,\Sigma^*),
    \end{equation*}
    where
    \begin{equation*}
        \Sigma^* = \left( \Psi^{* \prime} \Xi \Psi^* \right)^{-1} \Psi^{* \prime} \Xi \Omega_{gg}^*  \Xi \Psi^*
        \left( \Psi^{* \prime} \Xi \Psi^* \right)^{-1}.
    \end{equation*}
    Assumption~\ref{ass: basic conditions}\eqref{item: moment covar}, positive definiteness of $\Xi$, and full
    column rank of $\Psi^*$ imply that $\Sigma^*$ is positive definite.
    To complete the proof, we need to show that $\Sigma \rightarrow \Sigma^*$. First, note that, by
    Lemma~\ref{lem: consistency of objects} and $\gamma_{0n} \rightarrow 0$ (Assumption \ref{ass:MME})
    \begin{equation*}
        \Psi  = \left[G - \sum_{k=2}^K \gamma_{0kn} G_x^{(k)}, -g_x^{(2)}, \dots, -g_x^{(K)}\right] \rightarrow
        \left[G^*, -g_x^{(2)*}, \dots, -g_x^{(K)*} \right] = \Psi^*.
    \end{equation*}
    Next, we want to argue that $\Omega_{\psi \psi} \rightarrow \Omega_{gg}^*$. Observe that
    \begin{equation*}
        \Omega_{\psi \psi} = \ex{\left(g_i - \sum_{k=2}^K \gamma_{0kn} g_{xi}^{(k)}\right) \left(g_i -
        \sum_{k=2}^K \gamma_{0kn} g_{xi}^{(k)}\right)'} = \ex{g_i g_i'} + o(1),
    \end{equation*}
    where the equality follows since (i) $\gamma_{0kn} \rightarrow 0$ for all $k \in \{2, \dots, K\}$
    (Assumption~\ref{ass:MME}) and (ii) $\ex{g_{xi}^{(k)} \left(g_{xi}^{(k')}\right)'}$ is bounded for all $k,k'
    \in \{0, \dots, K\}$. In particular, (ii) can be verified by expanding $g_x^{{(k)}}(X_i,S_i,\theta_0)$ and
    $g_x^{{(k')}}(X_i,S_i,\theta_0)$ around $X_i^*$ as in \eqref{eq: A g k expansion} and bounding the
    expectations as in the proof of Lemma \ref{lem: consistency of objects} (using Assumptions~\ref{ass:MME},
    \ref{ass: CME}, \ref{ass: Moment function}\eqref{item: lipschitz moment}, \ref{ass: basic
    conditions}\eqref{item: Global and Point Moments}, and \ref{ass: basic conditions}\eqref{item: 2M moments}).
    Similarly, by expanding $g(X_i,S_i,\theta_0)$ around $X_i^*$ and bounding the residual terms as in the proof
    of Lemma~\ref{lem: consistency of objects} (again, using Assumptions \ref{ass:MME}, \ref{ass: CME}, \ref{ass:
    Moment function}\eqref{item: lipschitz moment}, \ref{ass: basic conditions}\eqref{item: Global and Point
    Moments}, and \ref{ass: basic conditions}\eqref{item: 2M moments}), we verify that $\ex{g_i g_i'} \rightarrow
    \ex{g_i^* g_i^{* \prime}} = \Omega_{gg}^*$. Hence, $\Omega_{\psi \psi} \rightarrow \Omega_{gg}^*$ and,
    consequently, we have verified that $\Sigma \rightarrow \Sigma^*$. Finally, we conclude
    \begin{equation*}
        n^{1/2} \Sigma^{-1/2} (\hat \beta - \beta_{0n}) \dist N(0,I_{\dim(\theta)+K-1}),
    \end{equation*}
    which completes the proof. \hfill \qed

\section{\label{sec:Proofs for Adaptive}Proofs of Lemmas \protect\ref{lem:choice-of-K} and
\protect\ref{lem:iterated-choice-of-K}}

As in the previous proofs, we use notation $\sigma_n^2 \equiv \mathbb E[\varepsilon_i^2] \equiv
\sigma_\varepsilon^2$, $\hat \sigma^2_n \equiv \hat \sigma_\varepsilon^2$, and $\hat \sigma_n^K \equiv \hat
\sigma_\varepsilon^K$.

\subsection*{Proof of Lemma~\ref{lem:choice-of-K}}

First, we show that \eqref{eq:estn sigK} holds. By Theorem~\ref{the: asy normality}, we have $\hat \sigma_n^2 =
\sigma_n^2 + O_p(n^{-1/2})$. Since $K$ is even, \eqref{eq:estn sigK} can be obtained by expanding
\begin{align*}
    \hat \sigma_n^K &= (\sigma_n^2 + O_p(n^{-1/2}))^{K/2} = \sum_{\ell=0}^{K/2} \binom{K/2}{\ell} (\sigma_n^2)^{K/2 - \ell} (O_p(n^{-1/2}))^{\ell} \\ &=  \sigma_n^K + \sum_{\ell=1}^{K/2} \binom{K/2}{\ell} (\sigma_n^2)^{K/2 - \ell} O_p(n^{-\ell/2}) \\ & = \sigma_n^K + O_p (\max\{\sigma_n^{K - 2} n^{-1/2}, n^{-K/4} \}).
\end{align*}

Next, we prove the first statement of the lemma. Using Lemma~\ref{lem: consistency of objects}, we have $\overline
g_x^{(K)}(\hat\theta)=g_x^{(K)*}+o_p(1)$ and $\mathbb E[g_x^{(K)}(X_i,S_i,\theta_0)]=g_x^{(K)*}+o(1)$. Since the
limiting Jacobian corresponding to $\hat\beta_L$ is obtained from that corresponding to $\hat\beta_K$ by deleting
the columns for $\gamma_{L+1},\ldots,\gamma_K$, it has full column rank by Assumption~\ref{ass: ID}\eqref{item:
local ID}. Thus, by the same continuity argument as in the proof of Theorem~\ref{the: asy normality}, $B$
converges to a finite matrix of full row rank. Assumption~\ref{ass: basic conditions}\eqref{item: moment covar}
then ensures that the limiting marginal variances of $\hat\beta_L$ are positive. The additional condition of the
lemma implies that $\Vert B g_x^{(K)*}\Vert$ is bounded away from zero, and the moment bounds ensure that it is
bounded above. Hence, using consistency of $\hat B$ and $\hat\Sigma$, we conclude that, with probability
approaching one, for some $\underline C,\overline C>0$, we have
\begin{align}
    \label{eq: bounds for delta_n}
    \underline C\sqrt n\vert\hat\sigma_n^K\vert
    \leq \max_{1\leq\ell\leq\dim(\beta_L)}\left\vert\hat\Sigma_{\ell\ell}^{-1/2}\sqrt n\hat\sigma_n^K\hat
    B_{\ell\cdot}\overline g_x^{(K)}(\hat\theta)\right\vert
    \leq\overline C\sqrt n\vert\hat\sigma_n^K\vert.
\end{align}

To prove the first part of the lemma, it is sufficient to show that $\delta_n\sigma_n^{L+1}=o_p(n^{-1/2})$, since
$\Vert\text{AsyB}(\hat\beta_L)\Vert=O(\sigma_n^{L+1})$. We prove this by contradiction. Suppose this convergence
does not hold. Then there exist constants $\epsilon,\eta>0$ and a deterministic subsequence $n_m$ such that
$\sqrt{n_m}\sigma_{n_m}^{L+1}\geq\epsilon$ and $\Pr(\delta_{n_m}=1)\geq\eta$ for every $m$. Since we have
$\sigma_{n_m}\geq C n_m^{-1/(2(L+1))}$ and $\sqrt{n_m}\sigma_{n_m}^K\geq C n_m^{-(K-L-1)/(2(L+1))}$, using
\eqref{eq:estn sigK}, we conclude that, for some $C>0$, with probability approaching one, we have
\begin{align}
    \label{eq: sigma_n^K bound}
    \sqrt{n_m}\vert\hat\sigma_{n_m}^K\vert
    &=\sqrt{n_m}\sigma_{n_m}^K+O_p\left(\sigma_{n_m}^{K-2}+n_m^{-(K-2)/4}\right) \nonumber\\
    &\geq C n_m^{-(K-L-1)/(2(L+1))},
\end{align}
because the remainder is $o_p(\sqrt{n_m}\sigma_{n_m}^K)$ when $\sigma_{n_m}\geq C n_m^{-1/(2(L+1))}$ and $L\geq2$.
At the same time, whenever \eqref{eq: bounds for delta_n} holds, $\delta_{n_m}=1$ implies that
$\sqrt{n_m}\vert\hat\sigma_{n_m}^K\vert\leq C\varkappa_{n_m}$ for some $C>0$. Since $\varkappa_n
n^{(K-L-1)/(2L+2)}\to0$, this is incompatible with \eqref{eq: sigma_n^K bound} for all sufficiently large $m$.
Consequently, $\Pr(\delta_{n_m}=1)\to0$, contradicting the choice of the subsequence. The contradiction completes
the proof of the first part.

We now turn to the second part of the lemma. By Jensen's inequality, we obtain
\begin{align*}
    \sqrt n\vert\hat\sigma_n^K\vert
    =\sqrt n\vert\sigma_n^2+O_p(n^{-1/2})\vert^{K/2}
    \leq C\sqrt n\sigma_n^K+O_p(n^{-(K-2)/4}).
\end{align*}
In particular, for $\varkappa_n=n^{-(K-L-1)/(2L+2)}(\ln n)^{-a}$ and $\sigma_n=o(n^{-1/(2L+2)-\epsilon})$, the
above implies that $\sqrt n\vert\hat\sigma_n^K\vert/\varkappa_n=o_p(1)$, since
$(K-2)/4-(K-L-1)/(2L+2)=K(L-1)/(4(L+1))>0$. This, together with \eqref{eq: bounds for delta_n}, implies that
$\delta_n=1$ with probability approaching one, which completes the proof. \hfill\qed

\subsection*{Proof of Lemma~\ref{lem:iterated-choice-of-K}}

First, note that Theorem~\ref{the: asy normality} applies at any candidate order $K$ whenever $\sqrt
n\sigma_n^{K+1}\to0$.

We start by proving the first statement of the lemma. The algorithm reaches order $K$ precisely when $\hat K_n\leq
K$. We first show that
\begin{equation}
    1\{\hat K_n\leq K\}\sqrt n\sigma_n^{K+1}=o_p(1),
    \label{eq: iterated K reached rate}
\end{equation}
for $K=\overline K,\overline K-2,\ldots,2$. At $K=\overline K$, this follows from Assumption~\ref{ass:MME}.
Suppose \eqref{eq: iterated K reached rate} holds at $K\geq4$ but fails at $K-2$. Then, as in the proof of
Lemma~\ref{lem:choice-of-K}, there exist $\epsilon,\eta>0$ and a deterministic subsequence $n_m$ such that
\[
    \sqrt{n_m}\sigma_{n_m}^{K-1}\geq\epsilon,
    \qquad \Pr(\hat K_{n_m}\leq K-2)\geq\eta.
\]
Since we also have $\Pr(\hat K_{n_m}\leq K)\geq\eta$, the induction hypothesis (i.e., that \eqref{eq: iterated K
reached rate} holds at $K$) gives $\sqrt{n_m}\sigma_{n_m}^{K+1}\to0$, implying that Theorem~\ref{the: asy
normality} applies at $K$ along this subsequence. With $\varkappa_n=(n\ln n)^{-1/(2K-2)}$ as in
Algorithm~\ref{alg: choose K}, the proof of Lemma~\ref{lem:choice-of-K} applied with $L = K - 2$ gives
$\delta_{n_m}\sqrt{n_m}\sigma_{n_m}^{K-1}=o_p(1)$. Since $\sqrt{n_m}\sigma_{n_m}^{K-1}\geq\epsilon$, it follows
that $\Pr(\delta_{n_m}=1)\to0$. Since the algorithm can reach $K-2$ only if $\delta_{n_m}=1$, this contradicts
$\Pr(\hat K_{n_m}\leq K-2)\geq\eta$ and proves \eqref{eq: iterated K reached rate} by induction.

As in the proof of Lemma~\ref{lem:choice-of-K}, $\Vert\text{AsyB}(\hat\beta_K)\Vert=O(\sigma_n^{K+1})$ for
$K<\overline K$, and $\text{AsyB}(\hat\beta_{\overline K})=0$ by definition. Hence \eqref{eq: iterated K reached
rate} gives $1\{\hat K_n=K\}\sqrt n\Vert\text{AsyB}(\hat\beta_K)\Vert=o_p(1)$ for every candidate $K$, proving the
first statement of the lemma.

We now turn to the second statement of the lemma. Suppose $\tau_n=o(n^{-1/(2L+2)-\epsilon})$ for some
$\epsilon>0$. For every even $K>L$ up to $\overline K$, this implies $\tau_n=o(n^{-1/(2K-2)-\epsilon})$. The
second part of Lemma~\ref{lem:choice-of-K} therefore gives $\Pr(\delta_n=1)\to1$ for each comparison above $L$.
Since the number of such comparisons is finite, the procedure reaches $L$ with probability approaching one, i.e.,
$\Pr(2\leq\hat K_n\leq L)\to1$.

Finally, if $L>2$ and $\liminf_{n\to\infty}\sqrt n\tau_n^{L-1}>0$ (and thus $\liminf_{n\to\infty}\sqrt
n\sigma_n^{L-1}>0$), \eqref{eq: iterated K reached rate} evaluated at $K = L - 2$ implies $\Pr(\hat K_n\leq
L-2)\to0$ and hence $\Pr(\hat K_n=L)\to1$. \hfill\qed

\section{\label{sec:full rank Psi}Details on the Sufficient Condition~(\protect\ref{eq:eg rk Psi:suff cond})}

In this section, we demonstrate that the Jacobian $\Psi ^{\ast }$ associated
with the moment conditions (\ref{eq:rank-cond:eg g}) in Section~\ref
{ssec:MERM-ID} is guaranteed to have full rank (for $K=2$) when the
sufficient condition (\ref{eq:eg rk Psi:suff cond}) holds.

It is convenient to define $\varrho _{i}\left( x,\theta \right) \equiv
u\left( x,S_{i},\theta \right) \times \left( 1,x,\ldots ,x^{J-1}\right)
^{\prime }$ and $u_{i}\left( x,\theta \right) \equiv u \left( x,S_{i},\theta
\right) $. Then
\begin{eqnarray*}
g\left( x,S_{i},q,\theta \right) &\equiv &\left(
\begin{array}{c}
u_{i}\left( x,\theta \right) \\
\varrho _{i}\left( x,\theta \right) x \\
\varrho _{i}\left( x,\theta \right) q
\end{array}
\right) ,\ \text{so} \\
g_{x}^{(2)}\left( x,S_{i},q,\theta \right) &=&\left(
\begin{array}{l}
u_{x,i}^{(2)}\left( x,\theta \right) \\
\varrho _{x,i}^{(2)}\left( x,\theta \right) x+2\varrho _{x,i}^{(1)}\left(
x,\theta \right) \\
\varrho _{x,i}^{(2)}\left( x,\theta \right) q
\end{array}
\right) .
\end{eqnarray*}

First, notice that
\begin{equation*}
\Psi ^{\ast }=\mathbb{E}
\begin{pmatrix}
u_{\theta ,i}\left( X_{i}^{\ast },\theta_0 \right) & -u_{x,i}^{(2)}\left(
X_{i}^{\ast },\theta_0 \right) \\
\varrho _{\theta ,i}\left( X_{i}^{\ast },\theta _{0}\right) X_{i}^{\ast } &
-\varrho _{x,i}^{(2)}\left( X_{i}^{\ast },\theta _{0}\right) X_{i}^{\ast
}-2\varrho _{x,i}^{(1)}\left( X_{i}^{\ast },\theta _{0}\right) \\
\varrho _{\theta ,i}\left( X_{i}^{\ast },\theta _{0}\right) \left( \alpha
_{1}X_{i}^{\ast }\right) & -\varrho _{x,i}^{(2)}\left( X_{i}^{\ast },\theta
_{0}\right) \left( \alpha _{1}X_{i}^{\ast }\right)
\end{pmatrix}
,
\end{equation*}
where we used $\mathbb{E}\left[ \varepsilon _{Q,i}|X_{i}^{\ast
},S_{i},\varepsilon _{i}\right] =0$. Dividing moments $J+2,\ldots ,2J+1$ by $\alpha _{1}$ and subtracting them
from the moments $2,\ldots ,J+1$, we
obtain
\begin{eqnarray*}
\operatorname{Rk}\left( \Psi ^{\ast }\right) &=&\operatorname{Rk}\mathbb{E}
\begin{pmatrix}
u_{\theta ,i}\left( X_{i}^{\ast },\theta_0 \right) & -u_{x,i}^{(2)}\left(
X_{i}^{\ast },\theta_0 \right) \\
0 & -2\varrho _{x,i}^{(1)}\left( X_{i}^{\ast },\theta _{0}\right) \\
\varrho _{\theta ,i}\left( X_{i}^{\ast },\theta _{0}\right) X_{i}^{\ast } &
-\varrho _{x,i}^{(2)}\left( X_{i}^{\ast },\theta _{0}\right) X_{i}^{\ast }
\end{pmatrix}
\\
&=&\operatorname{Rk}\mathbb{E}
\begin{pmatrix}
u_{\theta ,i}\left( X_{i}^{\ast },\theta_0 \right) & -u_{x,i}^{(2)}\left(
X_{i}^{\ast },\theta_0 \right) \\
\varrho _{\theta ,i}\left( X_{i}^{\ast },\theta _{0}\right) X_{i}^{\ast } &
-\varrho _{x,i}^{(2)}\left( X_{i}^{\ast },\theta _{0}\right) X_{i}^{\ast }
\\
0 & -2\varrho _{x,i}^{(1)}\left( X_{i}^{\ast },\theta _{0}\right)
\end{pmatrix}
\\
&=&\operatorname{Rk}
\begin{pmatrix}
H^{\ast } & \mathbb{E}
\begin{pmatrix}
u_{x,i}^{(2)}\left( X_{i}^{\ast },\theta_0 \right) \\
\varrho _{x,i}^{(2)}\left( X_{i}^{\ast },\theta _{0}\right) X_{i}^{\ast }
\end{pmatrix}
\\
0 & 2\mathbb{E}\left[ \varrho _{x,i}^{(1)}\left( X_{i}^{\ast },\theta
_{0}\right) \right]
\end{pmatrix}
.
\end{eqnarray*}
Here $\operatorname{Rk}\left( H^{\ast }\right) =\dim \left( \theta \right) $,
because this is the rank identification condition for $\theta _{0}$ in the
model without EIV. Thus, for $\Psi ^{\ast }$ to have full rank $\dim \left(
\theta \right) +1$, it is sufficient to have $\mathbb{E}\left[ \varrho
_{x,i}^{(1)}\left( X_{i}^{\ast },\theta _{0}\right) \right] \neq 0$. Note
that since $\mathbb{E}\left[ u\left( X_{i}^{\ast },S_{i},\theta _{0}\right)
|X_{i}^{\ast }\right] =0$, we have
\begin{equation*}
\mathbb{E}\left[ \varrho _{x,i}^{(1)}\left( X_{i}^{\ast },\theta _{0}\right) \right] =\mathbb{E}\left[
u_{x}^{(1)}\left( X_{i}^{\ast },S_{i},\theta
_{0}\right) \left( 1,X_{i}^{\ast },\ldots ,\left( X_{i}^{\ast }\right)
^{J-1}\right) ^{\prime }\right] \neq 0,
\end{equation*}
where the last expression is nonzero by condition~\eqref{eq:eg rk Psi:suff cond}.

\section{\label{sec:MME expansion with large ME}MERM Derivation when $\protect\sigma _{\protect\varepsilon }$ is
Not Small}

Note that $\tau $ can be small without $\sigma _{\varepsilon }$ being small
in absolute magnitude. For example, suppose $\sigma _{\varepsilon }=10$ and $\sigma _{X^{\ast }}=100$. Then $\tau
=0.1$, so the measurement error is
quite small relative to $\sigma _{X^{\ast }}$, and relying on the
approximation $\tau \rightarrow 0$ is reasonable. At the same time,
the approximation $\sigma _{\varepsilon }\rightarrow 0$ may not be suitable for
this example.

In this Appendix we show that the corrected moment conditions and the MERM
estimator are valid without assuming that $\sigma _{\varepsilon }$ is small
in absolute magnitude. In Section~\ref{sec:framework} we used Taylor
expansions in $\varepsilon _{i}$ around $\varepsilon _{i}=0$, with the
remainder of order $\mathbb{E}\left[ \left\vert \varepsilon _{i}\right\vert
^{K+1}\right] $. When $\sigma _{\varepsilon }>1$, term $O\left( \mathbb{E}\left[ \left\vert \varepsilon
_{i}\right\vert ^{K+1}\right] \right) $ in
equation~(\ref{eq: moment exp approx - step 1}) cannot be viewed as a
negligible remainder, because $\mathbb{E}\left[ \left\vert \varepsilon
_{i}\right\vert ^{K+1}\right] >1$ and, moreover, terms $\mathbb{E}\left[
\left\vert \varepsilon _{i}\right\vert ^{k}\right] $ increase rather than
decrease with $k$.

In Section~\ref{sec:framework}, to simplify the exposition, we have assumed
that $X^{\ast }$ is scaled so that $\sigma _{X^{\ast }}$ is of order one.
This in particular ensures that the moments $\mathbb{E}\left[ \left\vert \varepsilon_{i}\right\vert ^{k}\right] $
decrease with $k$. We will now show that this assumption about the scale of $X^{\ast }$ is not necessary, and that the
procedure remains valid without any such scaling.

We will show that, after rescaling, the Taylor expansions in Section~\ref
{sec:framework} can be written in terms of powers of $\tau ^{k}$, which
necessarily decrease with $k$ when $\tau <1$.

Remember the model of Section~\ref{sec:framework}:
\begin{equation}
\mathbb{E}[g(X_{i}^{\ast },S_{i},\theta _{0})]=0,\quad X_{i}=X_{i}^{\ast
}+\varepsilon _{i},\quad \mathbb{E}[\varepsilon _{i}]=0.
\label{eq:appx:LargeME:original model}
\end{equation}
Let $\xi_{i}$ denote a random variable such that $\varepsilon _{i}\equiv \sigma_{\varepsilon }\xi _{i}$ with
$\mathbb{E}\left[ \xi _{i}\right] =0$, $\mathbb{E}\left[ \xi _{i}^{2}\right] =1$, and $\mathbb{E}\left[
\left\vert \xi_{i}\right\vert ^{K+1}\right] $ bounded. Also, let us denote
\begin{equation*}
\tau \equiv \sigma _{\varepsilon }/\sigma _{X^{\ast }},\quad \widetilde{X}_{i}\equiv X_{i}/\sigma _{X^{\ast
}},\quad \widetilde{X}_{i}^{\ast }\equiv
X_{i}^{\ast }/\sigma _{X^{\ast }},\quad \widetilde{g}(\widetilde{x},s,\theta
)\equiv g(\sigma _{X^{\ast }}\widetilde{x},s,\theta ).
\end{equation*}
Then, {} we can rewrite equation~(\ref
{eq:appx:LargeME:original model}) as
\begin{equation*}
\mathbb{E}[\widetilde{g}(\widetilde{X}_{i}^{\ast },S_{i},\theta
_{0})]=0,\quad \widetilde{X}_{i}=\widetilde{X}_{i}^{\ast }+\tau \xi
_{i},\quad \mathbb{E}[\xi _{i}]=0.
\end{equation*}
Expand $\widetilde{g}(\widetilde{X}_{i},S_{i},\theta )=\widetilde{g}(\widetilde{X}_{i}^{\ast }+\tau \xi
_{i},S_{i},\theta )$ around $\tau =0$ to
obtain
\begin{equation*}
\mathbb{E}[\widetilde{g}(\widetilde{X}_{i},S_{i},\theta )]=\mathbb{E}[\widetilde{g}(\widetilde{X}_{i}^{\ast
},S_{i},\theta )]+\sum_{k=2}^{K}\frac{\tau ^{k}\mathbb{E}\left[ \xi _{i}^{k}\right] }{k!}\mathbb{E}\left[
\widetilde{g}_{x}^{(k)}(\widetilde{X}_{i}^{\ast },S_{i},\theta )\right]
+O(\tau ^{K+1}),
\end{equation*}
which is similar to equation~(\ref{eq: moment exp approx - step 1}), except $\mathbb{E}\left[ \varepsilon
_{i}^{k}\right] $ is replaced by $\tau ^{k}\mathbb{E}\left[ \xi _{i}^{k}\right] $, and $X_i$, $X_i^{\ast}$, and
$g$ are replaced by
$\widetilde{X}_i$, $\widetilde{X}_i^{\ast}$, and $\widetilde{g}$,
respectively. Then, the corrected moment condition has the form
\begin{equation}
\widetilde{\psi }(\widetilde{X}_{i},S_{i},\theta ,\widetilde{\gamma
})=\widetilde{g}(\widetilde{X}_{i},S_{i},\theta )-\sum_{k=2}^{K}\widetilde{\gamma
}_{k}\widetilde{g}_{x}^{(k)}(\widetilde{X}_{i},S_{i},\theta ),
\label{eq:psi tilde def}
\end{equation}
where true parameter values $\widetilde{\gamma }_{0}$ are $\widetilde{\gamma
}_{02}=\tau ^{2}\mathbb{E}\left[ \xi _{i}^{2}\right] /2=\tau ^{2}/2$, $\widetilde{\gamma }_{03}=\tau
^{3}\mathbb{E}\left[ \xi _{i}^{3}\right] /6$,
and $\widetilde{\gamma }_{0k}=\frac{\tau ^{k}\mathbb{E}\left[ \xi _{i}^{k}\right] }{k!}-\sum_{\ell
=2}^{k-2}\frac{\tau ^{k-\ell }\mathbb{E}\left[ \xi
_{i}^{k-\ell }\right] }{(k-\ell )!}\widetilde{\gamma }_{0\ell }$ for $k\geq
4 $.

We will now show that
\begin{equation*}
\gamma _{0k}=\sigma _{X^{\ast }}^{k}\widetilde{\gamma }_{0k}\text{ for }2\leq k\leq K.
\end{equation*}
First, $\gamma _{02}=\mathbb{E}\left[ \varepsilon _{i}^{2}\right] /2=\mathbb{E}\left[ \left( \sigma _{\varepsilon
}\xi _{i}\right) ^{2}\right] /2=\sigma
_{X^{\ast }}^{2}\widetilde{\gamma }_{02}$, $\gamma _{03}=\mathbb{E}\left[
\varepsilon _{i}^{3}\right] /6=\sigma _{X^{\ast }}^{3}\widetilde{\gamma }_{03}$ by definition. Then, for $k\geq
4$, by induction we have
\begin{eqnarray*}
\gamma _{0k} &=&\frac{\mathbb{E}\left[ \varepsilon _{i}^{k}\right] }{k!}-\sum_{\ell =2}^{k-2}\frac{\mathbb{E}\left[ \varepsilon _{i}^{k-\ell }\right]
}{(k-\ell )!}\gamma _{0\ell } \\
&=&\sigma _{X^{\ast }}^{k}\left( \frac{\left( \sigma _{\varepsilon }/\sigma
_{X^{\ast }}\right) ^{k}\mathbb{E}\left[ \xi _{i}^{k}\right] }{k!}-\sum_{\ell =2}^{k-2}\frac{\left( \sigma
_{\varepsilon }/\sigma _{X^{\ast
}}\right) ^{k-\ell }\mathbb{E}\left[ \xi _{i}^{k-\ell }\right] }{(k-\ell )!}\frac{\gamma _{0\ell }}{\sigma
_{X^{\ast }}^{\ell }}\right) \\
&=&\sigma _{X^{\ast }}^{k}\left( \frac{\tau ^{k}\mathbb{E}\left[ \xi _{i}^{k}\right] }{k!}-\sum_{\ell =2}^{k-2}\frac{\tau ^{k-\ell }\mathbb{E}\left[ \xi
_{i}^{k-\ell }\right] }{(k-\ell )!}\widetilde{\gamma }_{0\ell }\right)
=\sigma _{X^{\ast }}^{k}\widetilde{\gamma }_{0k}\text{.}
\end{eqnarray*}

Finally, let us now show that moment condition $\widetilde{\psi }$ in
equation~(\ref{eq:psi tilde def}) is numerically identical to $\psi $ in
equation~(\ref{eq: psi moms definition}) with $\gamma _{k}=\sigma _{X^{\ast
}}^{k}\widetilde{\gamma }_{k}$. Note that for $\widetilde{x}=x/\sigma
_{X^{\ast }}$ we have $\widetilde{g}_{\widetilde{x}}^{(k)}(\widetilde{x},s,\theta )\equiv \nabla
_{\widetilde{x}}^{k}g(\sigma _{X^{\ast }}\widetilde{x},s,\theta )=\sigma _{X^{\ast }}^{k}g_{a}^{(k)}(a,s,\theta
)|_{a=\sigma
_{X^{\ast }}\widetilde{x}}=\sigma _{X^{\ast }}^{k}g_{x}^{(k)}(x,s,\theta )$,
and hence
\begin{eqnarray*}
\widetilde{\psi }(\widetilde{X}_{i},S_{i},\theta ,\widetilde{\gamma })
&=&g(\sigma _{X^{\ast }}\widetilde{X}_{i},S_{i},\theta
)-\sum_{k=2}^{K}\left( \widetilde{\gamma }_{k}\sigma _{X^{\ast }}^{k}\right)
g_{x}^{(k)}(\sigma _{X^{\ast }}\widetilde{X}_{i},S_{i},\theta ) \\
&=&g(X_{i},S_{i},\theta )-\sum_{k=2}^{K}\left( \widetilde{\gamma }_{k}\sigma
_{X^{\ast }}^{k}\right) g_{x}^{(k)}(X_{i},S_{i},\theta ) = \psi (X_{i},S_{i},\theta ,\gamma )\text{.}
\end{eqnarray*}

\section{\label{sec:Implementation}Some Implementation Details}

\noindent\textbf{Numerical Optimization} \newline
Since $\overline{\psi }\left( \theta ,\gamma \right) $ is a linear function
of $\gamma $ it can be profiled out of the quadratic form $\hat{Q}(\theta
,\gamma )$. Thus, the criterion function only needs to be minimized
numerically over $\theta $.

\bigskip

\noindent\textbf{Choice of the weighting matrix $\hat \Xi$} \newline
As for the standard GMM estimator, the optimal weighting matrix can be
estimated by
\begin{equation*}
\hat{\Xi}_{\text{eff}}\equiv \hat{\Omega}_{\psi \psi }^{-1}(\tilde{\theta},\tilde{\gamma}),
\end{equation*}
where $\tilde{\theta}$ and $\tilde{\gamma}$ are some preliminary estimators
of $\theta _{0}$ and $\gamma _{0}$, and $\hat{\Omega}_{\psi \psi }(\theta
,\gamma )\equiv n^{-1}\sum_{i=1}^n\psi _{i}(\theta ,\gamma )\psi _{i}(\theta
,\gamma )^{\prime }.$ One example of such a preliminary estimator would be
the 1-step (GMM-)MERM estimator using $\hat{\Xi}_{\text{GMM}1}\equiv \hat{\Omega}_{\psi \psi
}^{-1}(\hat{\theta}_{\text{Naive}},0)$ as the first-step
GMM weighting matrix, where $\hat{\theta}_{\text{Naive}}$ is a naive
estimator of $\theta _{0}$ that ignores EIV. Note that $\hat{\Omega}_{\psi
\psi }(\hat{\theta}_{\text{Naive}},0)=$ $\hat{\Omega}_{gg}\left( \hat{\theta}_{\text{Naive}}\right) $, where
$\hat{\Omega}_{gg}(\theta )\equiv
n^{-1}\sum_{i=1}^n g_{i}(\theta )g_{i}(\theta )^{\prime }$.

One may also consider the regularized version of the efficient weighting
matrix estimator $\hat{\Xi}_{\text{eff,R}}\equiv \hat{\Omega}_{\psi \psi
}^{-1}(\tilde{\theta},0)$. Since $\gamma _{0}\rightarrow 0$, using the
regularized version $\hat{\Xi}_{\text{eff,R}}$ does not lead to a loss of
efficiency. {} Moreover, our
simulation studies suggest that using the regularized weighting matrix $\hat{\Xi}_{\text{eff,R}}$ results in
better finite sample performance of the MERM
estimator and, hence, is recommended in practice.

Although not indicated by the notation in equation~(\ref{eq:MME definition}), the weighting matrix
$\hat{\Xi}\equiv \hat{\Xi}(\theta ,\gamma )$ is
allowed to be a function of $\theta $ and $\gamma $. For example,
the Continuously Updating GMM estimator (CUE)\ corresponds to taking $\hat{\Xi}_{\text{CUE}}\left( \theta ,\gamma
\right) \equiv \hat{\Omega}_{\psi \psi
}^{-1}(\theta ,\gamma )$. Similarly to $\hat{\Xi}_{\text{eff,R}}$, one may
also consider $\hat{\Xi}_{\text{CUE,R}}\left( \theta ,\gamma \right) \equiv
\hat{\Omega}_{\psi \psi }^{-1}(\theta ,0)$ without introducing any loss of
efficiency. In contrast to the criterion function of the CUE estimator,
the criterion function $\hat{Q}_{\text{CUE,R}}(\theta ,\gamma )$ is quadratic
in $\gamma $. {}
This implies that $\gamma $ can be profiled out analytically. This
simplifies the numerical optimization problem reducing it to minimizing $\hat{Q}_{\text{CUE,R}}\left( \theta
,\hat{\gamma}\left( \theta \right)
\right) $ over $\theta \in \Theta $. Then, the dimension of the optimization
parameter $\theta $ for the corrected moment condition problem remains the
same as for the original (naive) estimation problem without the EIV
correction.

\bigskip

\noindent\textbf{Estimation of the asymptotic variance $\Sigma$} \newline
Theorem \ref{the: asy normality} shows that the MERM estimator $\hat{\beta}=(\hat{\theta}^{\prime
},\hat{\gamma}^{\prime })^{\prime }$ behaves like a
standard GMM estimator based on the corrected moment function $\psi (\theta
,\gamma )$. The researcher can rely on the standard GMM inference
procedures. The asymptotic variance of $\hat{\beta}$ can be consistently
estimated by
\begin{equation*}
\hat{\Sigma}\equiv (\hat{\Psi}^{\prime }\hat{\Xi}\hat{\Psi})^{-1}\hat{\Psi}^{\prime }\hat{\Xi}\hat{\Omega}_{\psi
\psi }\hat{\Xi}\hat{\Psi}(\hat{\Psi}^{\prime }\hat{\Xi}\hat{\Psi})^{-1},
\end{equation*}
where $\hat{\Xi}$ is the chosen weighting matrix, and $\hat{\Psi}\equiv
\overline{\Psi }(\hat{\theta},\hat{\gamma})=n^{-1}\sum_{i=1}^n\Psi _{i}(\hat{\theta},\hat{\gamma})$ and
$\hat{\Omega}_{\psi \psi }=\hat{\Omega}_{\psi
\psi }(\hat{\theta},\hat{\gamma})$ are estimators of $\Psi $ and $\Omega
_{\psi \psi }$.

\section{\label{sec:additional MCs}Additional Numerical Results}

In this section, we provide additional numerical results for the experiments considered in Sections~\ref{ssec: MC
Mult Logit} and \ref{ssec:adaptive K}.

\subsection{\label{ssec: large tau MC}Additional Numerical Results for Section \protect{\ref{ssec: MC Mult Logit}}}

In this section, we provide additional simulation results for the same experiment as considered in
Section~\ref{ssec: MC Mult Logit} but for larger values of $\tau \in \{1, 1.5, 2\}$. The results are provided in
Table~\ref{tab: MC mult logit 2, large tau} below, reporting the same statistics as in Table~\ref{tab: MC mult
logit 2} in the main text.

\begin{landscape}
\begin{table}[h!]
\begin{footnotesize}
\begin{threeparttable}
\caption{Simulation results for the multinomial logit model}\label{tab: MC mult logit 2, large tau}
    \begin{tabular}{ l c c c c| c c c c | c c c c}
    \toprule
    \toprule
      & \multicolumn{4}{c}{MLE} & \multicolumn{4}{c}{$K=2$} & \multicolumn{4}{c}{$K=4$}\\
      \cmidrule(lr){2-5}  \cmidrule(lr){6-9} \cmidrule(lr){10-13}
      &{\text{bias}, $10^{-2}$}&{\text{std}, $10^{-2}$}&{\text{rmse}, $10^{-2}$}&{\text{size}}&{\text{bias}, $10^{-2}$}&{\text{std}, $10^{-2}$}&{\text{rmse}, $10^{-2}$}&{\text{size}}&{\text{bias}, $10^{-2}$}&{\text{std}, $10^{-2}$}&{\text{rmse}, $10^{-2}$}&{\text{size}}\\
      \midrule
      \multicolumn{13}{c}{$\tau = 1$}\\
      \midrule
       $\partial p_1/\partial x$&-16.14&0.70&16.15&100.00&-11.45&2.21&11.66&99.16&0.21&3.53&3.54&5.44\\
     $\partial p_1/\partial w_1$&11.82&1.40&11.91&100.00&8.49&2.46&8.84&96.14&0.48&2.88&2.92&6.30\\
     $\partial p_1/\partial w_2$&2.08&0.67&2.19&89.28&1.63&0.84&1.84&56.78&0.10&1.04&1.05&5.94\\
       $\partial p_2/\partial x$&8.78&0.65&8.80&100.00&6.59&1.50&6.76&97.64&-0.04&2.53&2.53&6.42\\
     $\partial p_2/\partial w_1$&-5.87&0.71&5.91&100.00&-4.24&1.23&4.42&96.08&-0.24&1.44&1.46&6.26\\
     $\partial p_2/\partial w_2$&-4.03&1.27&4.22&89.50&-3.20&1.62&3.59&57.34&-0.18&2.09&2.10&6.24\\
       $\partial p_0/\partial x$&7.36&0.61&7.39&100.00&4.86&1.33&5.04&94.72&-0.17&1.95&1.96&5.30\\
     $\partial p_0/\partial w_1$&-5.95&0.73&6.00&100.00&-4.24&1.24&4.42&96.14&-0.24&1.44&1.46&6.28\\
     $\partial p_0/\partial w_2$&1.94&0.61&2.04&89.52&1.57&0.79&1.76&57.96&0.08&1.05&1.05&6.24\\
      \midrule
      \multicolumn{13}{c}{$\tau = 3/2$}\\
      \midrule
       $\partial p_1/\partial x$&-19.01&0.50&19.01&100.00&-16.59&1.32&16.64&100.00&-2.29&3.67&4.33&16.56\\
     $\partial p_1/\partial w_1$&14.07&1.34&14.13&100.00&12.83&1.90&12.97&99.98&2.70&3.40&4.34&24.14\\
     $\partial p_1/\partial w_2$&2.31&0.64&2.40&95.70&2.20&0.75&2.32&87.72&0.56&1.18&1.31&12.18\\
       $\partial p_2/\partial x$&9.97&0.47&9.98&100.00&9.05&0.95&9.10&100.00&1.58&2.70&3.13&16.10\\
     $\partial p_2/\partial w_1$&-6.96&0.69&6.99&100.00&-6.40&0.96&6.48&99.98&-1.35&1.71&2.18&24.08\\
     $\partial p_2/\partial w_2$&-4.44&1.21&4.60&95.88&-4.29&1.43&4.52&87.88&-1.09&2.35&2.59&12.68\\
       $\partial p_0/\partial x$&9.04&0.46&9.05&100.00&7.54&0.94&7.60&100.00&0.71&2.05&2.18&9.46\\
     $\partial p_0/\partial w_1$&-7.11&0.71&7.15&100.00&-6.42&0.98&6.50&99.98&-1.35&1.70&2.17&23.90\\
     $\partial p_0/\partial w_2$&2.13&0.58&2.21&95.90&2.09&0.69&2.20&88.08&0.53&1.17&1.29&13.00\\
     \midrule
      \multicolumn{13}{c}{$\tau = 2$}\\
      \midrule
       $\partial p_1/\partial x$&-20.28&0.38&20.29&100.00&-18.86&0.88&18.88&100.00&-5.98&3.91&7.15&55.08\\
     $\partial p_1/\partial w_1$&15.09&1.32&15.15&100.00&14.91&1.66&15.00&100.00&5.87&3.87&7.03&55.08\\
     $\partial p_1/\partial w_2$&2.39&0.63&2.47&97.42&2.37&0.71&2.47&94.38&1.12&1.30&1.72&26.44\\
       $\partial p_2/\partial x$&10.45&0.36&10.46&100.00&9.97&0.69&9.99&100.00&3.80&2.82&4.73&47.94\\
     $\partial p_2/\partial w_1$&-7.45&0.68&7.48&100.00&-7.43&0.85&7.48&100.00&-2.95&1.95&3.54&54.96\\
     $\partial p_2/\partial w_2$&-4.58&1.18&4.73&97.42&-4.61&1.36&4.81&94.54&-2.20&2.54&3.36&27.46\\
       $\partial p_0/\partial x$&9.83&0.36&9.84&100.00&8.88&0.70&8.91&100.00&2.18&2.21&3.11&28.50\\
     $\partial p_0/\partial w_1$&-7.64&0.71&7.68&100.00&-7.48&0.87&7.53&100.00&-2.92&1.93&3.50&54.80\\
     $\partial p_0/\partial w_2$&2.19&0.57&2.26&97.42&2.24&0.66&2.34&94.62&1.08&1.26&1.65&27.78\\
        \bottomrule
    \end{tabular}
\begin{tablenotes}
\item \footnotesize This table reports the simulated finite sample bias, standard deviation, RMSE, and size of the
MLE and the MERM estimators and the corresponding t-tests for the partial derivatives $\partial
p_j(x,w,\theta_0)/\partial x$, $\partial p_j(x,w,\theta_0)/\partial w_1$, $\partial p_j(x,w,\theta_0)/\partial
w_2$ for $j \in \{1,2,0\}$ evaluated at the population mean. The true values of the marginal effects are
$(\partial p_1/\partial x, \partial p_2/\partial x, \partial p_0/\partial x) = ( 0.222,   -0.111,   -0.111)$ and
zeros for the rest. The results are based on 5,000 replications.
\end{tablenotes}
\end{threeparttable}
\end{footnotesize}
\end{table}
\end{landscape}

\subsection{\label{ssec: additional adaptive MCs}Additional Numerical Results for Section
\protect{\ref{ssec:adaptive K}}}
In this section, we provide additional numerical results for the experiment considered in
Section~\protect{\ref{ssec:adaptive K}}. In particular, Tables~\ref{tab: adaptive K mult logit 1000} and \ref{tab:
adaptive K mult logit 4000} below report the same statistics as Table~\ref{tab: adaptive K mult logit 2000} in the
main text, but for smaller and larger sample sizes $n=1000$ and $n=4000$.

\begin{table}[h!]
\begin{scriptsize}
\begin{threeparttable}
\caption{Choice of $K$ simulation results for the multinomial logit model, $n=1000$}\label{tab: adaptive K mult
logit 1000}
    \begin{tabular}{ l c c |  c c |  c c |  c c}
    \toprule
    \toprule
      & \multicolumn{2}{c}{MLE} & \multicolumn{2}{c}{$K=2$} & \multicolumn{2}{c}{$K=4$} & \multicolumn{2}{c}{data-driven $K$}\\
      \cmidrule(lr){2-3}  \cmidrule(lr){4-5} \cmidrule(lr){6-7} \cmidrule(lr){8-9}
      &{\text{bias}, $10^{-2}$}&{\text{rmse}, $10^{-2}$}&{\text{bias}, $10^{-2}$}&{\text{rmse}, $10^{-2}$}&{\text{bias}, $10^{-2}$}&{\text{rmse}, $10^{-2}$}&{\text{bias}, $10^{-2}$}&{\text{rmse}, $10^{-2}$}\\
      \midrule
      \multicolumn{9}{c}{$\tau = 1/4$}\\
      \midrule
       $\partial p_1/\partial x$&-3.22&3.74&1.34&3.96&2.10&4.56&1.37&4.05\\
     $\partial p_1/\partial w_1$&2.29&3.27&-0.13&3.32&-0.57&3.39&-0.15&3.34\\
     $\partial p_1/\partial w_2$&0.52&1.19&-0.01&1.22&-0.10&1.26&-0.02&1.23\\
       $\partial p_2/\partial x$&1.93&2.55&-0.73&2.78&-1.18&3.11&-0.75&2.82\\
     $\partial p_2/\partial w_1$&-1.15&1.64&0.07&1.66&0.29&1.70&0.07&1.67\\
     $\partial p_2/\partial w_2$&-1.02&2.36&0.03&2.46&0.21&2.53&0.04&2.46\\
       $\partial p_0/\partial x$&1.28&1.93&-0.61&2.16&-0.92&2.41&-0.62&2.20\\
     $\partial p_0/\partial w_1$&-1.14&1.63&0.07&1.66&0.29&1.70&0.07&1.67\\
     $\partial p_0/\partial w_2$&0.51&1.17&-0.02&1.24&-0.11&1.27&-0.02&1.24\\
      \midrule
      \multicolumn{9}{c}{$\tau = 1/2$}\\
      \midrule
       $\partial p_1/\partial x$&-8.97&9.10&-1.82&4.10&1.76&4.57&1.35&4.82\\
     $\partial p_1/\partial w_1$&6.42&6.79&1.88&4.04&-0.35&3.52&-0.10&3.73\\
     $\partial p_1/\partial w_2$&1.32&1.67&0.42&1.33&-0.05&1.32&0.00&1.34\\
       $\partial p_2/\partial x$&5.21&5.39&1.16&2.91&-0.97&3.20&-0.73&3.30\\
     $\partial p_2/\partial w_1$&-3.21&3.39&-0.94&2.02&0.18&1.76&0.05&1.87\\
     $\partial p_2/\partial w_2$&-2.59&3.26&-0.83&2.65&0.11&2.65&0.01&2.68\\
       $\partial p_0/\partial x$&3.76&3.95&0.66&2.18&-0.79&2.50&-0.62&2.56\\
     $\partial p_0/\partial w_1$&-3.21&3.40&-0.94&2.02&0.18&1.76&0.05&1.86\\
     $\partial p_0/\partial w_2$&1.27&1.60&0.41&1.33&-0.06&1.33&-0.01&1.35\\
     \midrule
      \multicolumn{9}{c}{$\tau = 3/4$}\\
      \midrule
       $\partial p_1/\partial x$&-13.36&13.41&-7.17&7.97&1.07&4.74&1.07&4.75\\
     $\partial p_1/\partial w_1$&9.68&9.90&5.59&6.67&0.17&3.79&0.17&3.79\\
     $\partial p_1/\partial w_2$&1.85&2.09&1.15&1.69&0.07&1.41&0.07&1.41\\
       $\partial p_2/\partial x$&7.48&7.57&4.26&4.91&-0.54&3.36&-0.54&3.36\\
     $\partial p_2/\partial w_1$&-4.82&4.94&-2.80&3.34&-0.09&1.90&-0.09&1.90\\
     $\partial p_2/\partial w_2$&-3.59&4.05&-2.27&3.33&-0.13&2.82&-0.13&2.82\\
       $\partial p_0/\partial x$&5.87&5.96&2.91&3.53&-0.53&2.67&-0.53&2.68\\
     $\partial p_0/\partial w_1$&-4.86&4.97&-2.79&3.34&-0.08&1.90&-0.08&1.90\\
     $\partial p_0/\partial w_2$&1.74&1.96&1.12&1.64&0.06&1.41&0.06&1.42\\
        \bottomrule
    \end{tabular}
\begin{tablenotes}
\item \scriptsize This table reports the simulated finite sample bias and RMSE of the MLE and the MERM estimators
for the partial derivatives $\partial p_j(x,w,\theta_0)/\partial x$, $\partial p_j(x,w,\theta_0)/\partial w_1$,
$\partial p_j(x,w,\theta_0)/\partial w_2$ for $j \in \{1,2,0\}$ evaluated at the population mean. The true values
of the marginal effects are $(\partial p_1/\partial x, \partial p_2/\partial x, \partial p_0/\partial x) = (
0.222,   -0.111,   -0.111)$ and zeros for the rest. The results are based on 5,000 replications.
\end{tablenotes}
\end{threeparttable}
\end{scriptsize}
\end{table}

\begin{table}[h!]
\begin{scriptsize}
\begin{threeparttable}
\caption{Choice of $K$ simulation results for the multinomial logit model, $n=4000$}\label{tab: adaptive K mult
logit 4000}
    \begin{tabular}{ l c c |  c c |  c c |  c c}
    \toprule
    \toprule
      & \multicolumn{2}{c}{MLE} & \multicolumn{2}{c}{$K=2$} & \multicolumn{2}{c}{$K=4$} & \multicolumn{2}{c}{data-driven $K$}\\
      \cmidrule(lr){2-3}  \cmidrule(lr){4-5} \cmidrule(lr){6-7} \cmidrule(lr){8-9}
      &{\text{bias}, $10^{-2}$}&{\text{rmse}, $10^{-2}$}&{\text{bias}, $10^{-2}$}&{\text{rmse}, $10^{-2}$}&{\text{bias}, $10^{-2}$}&{\text{rmse}, $10^{-2}$}&{\text{bias}, $10^{-2}$}&{\text{rmse}, $10^{-2}$}\\
      \midrule
      \multicolumn{9}{c}{$\tau = 1/4$}\\
      \midrule
       $\partial p_1/\partial x$&-3.28&3.41&0.39&1.90&0.56&1.93&0.39&1.90\\
     $\partial p_1/\partial w_1$&2.32&2.57&-0.08&1.59&-0.15&1.57&-0.08&1.59\\
     $\partial p_1/\partial w_2$&0.48&0.72&-0.04&0.62&-0.05&0.61&-0.04&0.62\\
       $\partial p_2/\partial x$&1.97&2.13&-0.23&1.35&-0.33&1.37&-0.23&1.35\\
     $\partial p_2/\partial w_1$&-1.16&1.28&0.04&0.80&0.08&0.78&0.04&0.79\\
     $\partial p_2/\partial w_2$&-0.96&1.43&0.07&1.23&0.11&1.23&0.07&1.23\\
       $\partial p_0/\partial x$&1.31&1.49&-0.16&1.01&-0.23&1.03&-0.16&1.01\\
     $\partial p_0/\partial w_1$&-1.16&1.29&0.04&0.80&0.08&0.78&0.04&0.80\\
     $\partial p_0/\partial w_2$&0.47&0.71&-0.04&0.62&-0.05&0.62&-0.04&0.62\\
      \midrule
      \multicolumn{9}{c}{$\tau = 1/2$}\\
      \midrule
       $\partial p_1/\partial x$&-9.00&9.03&-1.52&2.41&0.49&2.03&0.48&2.04\\
     $\partial p_1/\partial w_1$&6.43&6.52&1.06&1.99&-0.11&1.65&-0.10&1.65\\
     $\partial p_1/\partial w_2$&1.28&1.38&0.21&0.68&-0.04&0.65&-0.04&0.65\\
       $\partial p_2/\partial x$&5.22&5.27&0.90&1.64&-0.30&1.46&-0.29&1.47\\
     $\partial p_2/\partial w_1$&-3.21&3.25&-0.53&1.00&0.05&0.82&0.05&0.83\\
     $\partial p_2/\partial w_2$&-2.51&2.70&-0.42&1.35&0.09&1.30&0.09&1.30\\
       $\partial p_0/\partial x$&3.78&3.83&0.62&1.19&-0.20&1.10&-0.19&1.10\\
     $\partial p_0/\partial w_1$&-3.22&3.26&-0.53&1.00&0.05&0.82&0.05&0.83\\
     $\partial p_0/\partial w_2$&1.23&1.32&0.21&0.67&-0.05&0.65&-0.05&0.65\\
     \midrule
      \multicolumn{9}{c}{$\tau = 3/4$}\\
      \midrule
       $\partial p_1/\partial x$&-13.37&13.38&-6.49&6.79&0.41&2.25&0.41&2.25\\
     $\partial p_1/\partial w_1$&9.68&9.73&4.44&4.84&-0.02&1.80&-0.02&1.80\\
     $\partial p_1/\partial w_2$&1.80&1.87&0.91&1.12&-0.03&0.69&-0.03&0.69\\
       $\partial p_2/\partial x$&7.48&7.50&3.83&4.08&-0.24&1.63&-0.24&1.63\\
     $\partial p_2/\partial w_1$&-4.82&4.85&-2.22&2.42&0.01&0.90&0.01&0.90\\
     $\partial p_2/\partial w_2$&-3.50&3.63&-1.80&2.22&0.06&1.39&0.06&1.39\\
       $\partial p_0/\partial x$&5.89&5.91&2.66&2.87&-0.16&1.22&-0.16&1.22\\
     $\partial p_0/\partial w_1$&-4.86&4.89&-2.22&2.42&0.01&0.90&0.01&0.90\\
     $\partial p_0/\partial w_2$&1.70&1.76&0.89&1.09&-0.03&0.69&-0.03&0.69\\
        \bottomrule
    \end{tabular}
\begin{tablenotes}
\item \scriptsize This table reports the simulated finite sample bias and RMSE of the MLE and the MERM estimators
for the partial derivatives $\partial p_j(x,w,\theta_0)/\partial x$, $\partial p_j(x,w,\theta_0)/\partial w_1$,
$\partial p_j(x,w,\theta_0)/\partial w_2$ for $j \in \{1,2,0\}$ evaluated at the population mean. The true values
of the marginal effects are $(\partial p_1/\partial x, \partial p_2/\partial x, \partial p_0/\partial x) = (
0.222,   -0.111,   -0.111)$ and zeros for the rest. The results are based on 5,000 replications.
\end{tablenotes}
\end{threeparttable}
\end{scriptsize}
\end{table}

\section{\label{sec:Empirical Details}Implementation Details of the Empirically Calibrated Simulation}

In this section, we provide additional details on the implementation of the numerical experiment in
Section~\ref{sec:empirical}.

\bigskip

\noindent{\bf Data} \newline The original dataset is the ModeCanada dataset supplied with the R package
\verb|mlogit|. This dataset has been extensively used in transportation research. For a detailed description of
the dataset, see, for example, \citet{koppelman2000paired}, \citet{wen2001generalized}, and
\citet{Hansen2021-Textbook}. As in \citet{koppelman2000paired}, we use only the subset of travelers who chose
train, air, or car (and had all of those alternatives available to them), which leaves $n=2769$ observations.

\bigskip

\noindent{\bf Monte-Carlo design} \newline We set $\theta_0$ equal to the MLE estimate using the considered
dataset, which is reported in the table below.

\begin{center}
\begin{tabular}{l| c c c c c c c c}
	\toprule
	\toprule
	&$\theta_1$ & $\theta_2$ & $\theta_3$ & $\theta_4$ & $\theta_5$ & $\theta_6$ & $\theta_7$ & $\theta_8$ \\
	\midrule
    Estimates& 0.0355&   0.2976&   -2.0891&    0.0079&   -0.9900&    1.8794&   -0.0223&   -0.0149 \\
    Std. Err.& 0.0036&   0.0844&    0.4674&    0.0036&    0.0876&    0.2037&    0.0038&    0.0008 \\
    \bottomrule
\end{tabular}
\end{center}

To generate the simulated samples, we randomly draw the covariates (with replacement) from their joint empirical
distribution.
To ensure identification of the model, we also generate an instrumental variable $Z_i$ as
\begin{align*}
	Z_i = \kappa \thinspace Income_i^* / \sigma_{Income^*} + \sqrt{1 - \kappa^2} \zeta_i,
\end{align*}
where $\sigma_{Income^*} \approx 17.5$ is the standard deviation of $Income^*$, $\kappa = 0.5$, and $\zeta_i$ are
i.i.d. draws from $N(0,1)$ (which are also independent from all the other variables). Note that the instrument
$Z_i$ is ``caused by $X_{i}^*$''. For example, $Z_i$ can be some (noisy) measure of individual consumption.

\bigskip

\noindent{\bf Moments} \newline To simplify the notation, let $X_i^* \equiv Income_i^*$, $X_i \equiv Income_i$,
$R_i \equiv Urban_i$, $R_{ij} \equiv (Price_{ij}, InTime_{ij})'$ for $j \in \{0,1,2\}$, and $W_i \equiv (R_i,
R_{i1}', R_{i2}', R_{i0}')'$. Also let $Y_{ij} \equiv \ind \{j = \argmax_{j' \in \{0,1,2\}} U_{i j'}\}$ for $j \in
\{0,1,2\}$, $Y_i \equiv (Y_{i1},Y_{i2},Y_{i0})'$, and $p_j(x,w,\theta) \equiv \pr (Y_{ij} = 1|X_i^* = x, W_i = w;
\theta)$ with $w \equiv (r, r_1', r_2', r_0')'$, so
\begin{align*}
	  p_1(x,w,\theta) &= \frac{e^{\theta_{1} x + \theta_{2} r + \theta_{3} + (\theta_7,\theta_8)r_1}}{e^{\theta_{1} x + \theta_{2} r + \theta_{3} + (\theta_7,\theta_8)r_1} + e^{\theta_{4} x + \theta_{5} r + \theta_{6} + (\theta_7,\theta_8)r_2} + e^{(\theta_7,\theta_8)r_0}}, \\
	p_2(x,w,\theta) &= \frac{e^{\theta_{4} x + \theta_{5} r + \theta_{6} + (\theta_7,\theta_8)r_2}}{e^{\theta_{1} x + \theta_{2} r + \theta_{3} + (\theta_7,\theta_8)r_1} + e^{\theta_{4} x + \theta_{5} r + \theta_{6} + (\theta_7,\theta_8)r_2} + e^{(\theta_7,\theta_8)r_0}},
\end{align*}
and $p_0(x,w,\theta) = 1 - p_1(x,w,\theta) - p_2(x,w,\theta)$. Then, the original moment function takes the form of
\begin{align*}
     g (x,w,y,z,\theta) = \left(\left(y_1 - p_1(x,w,\theta) \right) \varphi_1(x,z,w)', \left(y_2 -
     p_2(x,w,\theta)\right) \varphi_2(x,z,w)' \right)',
\end{align*}
where $\varphi_j(x,z,w) = \left(1, x, z, x^2, z^2, x^3, z^3, r, (r_j - r_0)' \right)'$ for $K=2$ and
$\varphi_j(x,z,w) = \left(1, x, z, x^2, xz, z^2, x^3, x^2 z, x z^2, z^3, r, (r_j - r_0)'\right)'$ for $K=4$.

\bigskip

\noindent{\bf Income Elasticities} \newline In Section~\ref{sec:empirical}, we focus on estimation of and
inference on the income elasticities
\begin{align*}
	\frac{\partial \ln p_j}{\partial \ln x}(x,w,\theta) = \frac{x}{p_j(x,w,\theta)} \frac{\partial p_j
	(x,w,\theta)}{\partial x}.
\end{align*}
We report results for the income elasticities evaluated at the sample mean of $X^*$ and $W$ in the original sample.

\bigskip

\noindent\textbf{Estimation of and Inference on $\theta_0$} \newline In Table \ref{tab: MC emp theta} below, we
also report the estimation and inference results for the vector of parameters $\theta_0$ underlying the reported
results about elasticities.

\begin{table}[h!]
\begin{scriptsize}
\begin{threeparttable}
\caption{Simulation results for the empirically calibrated conditional logit model}
\label{tab: MC emp theta}
    \begin{tabular}{ l c c c c| c c c c | c c c c}
    \toprule
    \toprule
      & \multicolumn{4}{c}{MLE} & \multicolumn{4}{c}{$K=2$} & \multicolumn{4}{c}{$K=4$}\\
      \cmidrule(lr){2-5}  \cmidrule(lr){6-9} \cmidrule(lr){10-13}
      &{\text{bias}}&{\text{std}}&{\text{rmse}}&{\text{size}}&{\text{bias}}&{\text{std}}&{\text{rmse}}&{\text{size}}&{\text{bias}}&{\text{std}}&{\text{rmse}}&{\text{size}}\\
      \midrule
      \multicolumn{13}{c}{$\tau = 1/4$}\\
      \midrule
                            $\theta_1$ &-0.0021&0.0035&0.0041&8.70&0.0001&0.0042&0.0042&5.48&0.0005&0.0057&0.0058&7.38\\
                            $\theta_2$ &0.0047&0.0932&0.0933&5.10&0.0028&0.0957&0.0957&5.36&0.0022&0.0960&0.0960&5.40\\
                            $\theta_3$ &0.1152&0.4452&0.4599&6.00&-0.0048&0.4821&0.4821&5.94&-0.0251&0.5336&0.5342&6.68\\
                            $\theta_4$ &-0.0004&0.0031&0.0031&4.52&-0.0001&0.0034&0.0034&5.32&-0.0001&0.0036&0.0036&6.86\\
                            $\theta_5$ &-0.0023&0.0894&0.0895&5.18&-0.0088&0.0918&0.0922&5.54&-0.0113&0.0922&0.0929&5.96\\
                            $\theta_6$ &0.0232&0.1821&0.1836&4.64&0.0250&0.1982&0.1998&5.72&0.0329&0.2089&0.2115&6.74\\
                            $\theta_7$ &-0.0001&0.0035&0.0035&5.58&-0.0002&0.0036&0.0036&6.24&-0.0003&0.0036&0.0037&6.04\\
                            $\theta_8$ &-0.0001&0.0007&0.0007&4.82&-0.0001&0.0007&0.0007&5.48&-0.0002&0.0007&0.0007&5.58\\
      \midrule
      \multicolumn{13}{c}{$\tau = 1/2$}\\
      \midrule
                            $\theta_1$ &-0.0073&0.0032&0.0080&60.08&-0.0016&0.0043&0.0046&6.86&0.0005&0.0061&0.0061&6.60\\
                            $\theta_2$ &0.0109&0.0930&0.0936&5.18&0.0050&0.0959&0.0960&5.54&0.0026&0.0964&0.0965&5.36\\
                            $\theta_3$ &0.4080&0.4452&0.6039&17.12&0.0936&0.4874&0.4963&6.58&-0.0263&0.5475&0.5481&6.38\\
                            $\theta_4$ &-0.0012&0.0029&0.0031&6.46&-0.0003&0.0035&0.0035&5.22&-0.0002&0.0038&0.0038&6.52\\
                            $\theta_5$ &-0.0006&0.0894&0.0894&5.16&-0.0083&0.0919&0.0923&5.52&-0.0110&0.0924&0.0930&5.92\\
                            $\theta_6$ &0.0655&0.1752&0.1870&6.22&0.0348&0.2035&0.2064&5.86&0.0326&0.2158&0.2183&6.42\\
                            $\theta_7$ &-0.0003&0.0035&0.0036&5.64&-0.0003&0.0036&0.0037&6.34&-0.0003&0.0037&0.0037&6.06\\
                            $\theta_8$ &-0.0001&0.0007&0.0007&5.06&-0.0001&0.0007&0.0007&5.54&-0.0002&0.0007&0.0007&5.48\\
     \midrule
      \multicolumn{13}{c}{$\tau = 3/4$}\\
      \midrule
                            $\theta_1$ &-0.0132&0.0029&0.0135&99.34&-0.0056&0.0043&0.0071&25.12&0.0003&0.0065&0.0065&6.06\\
                            $\theta_2$ &0.0180&0.0923&0.0940&5.36&0.0102&0.0961&0.0966&5.76&0.0033&0.0973&0.0973&5.44\\
                            $\theta_3$ &0.7336&0.4496&0.8604&41.66&0.3203&0.4859&0.5820&12.00&-0.0130&0.5666&0.5667&6.20\\
                            $\theta_4$ &-0.0024&0.0026&0.0035&14.00&-0.0009&0.0035&0.0036&5.94&-0.0002&0.0041&0.0041&5.76\\
                            $\theta_5$ &0.0021&0.0890&0.0891&5.08&-0.0071&0.0921&0.0924&5.68&-0.0109&0.0926&0.0932&5.82\\
                            $\theta_6$ &0.1204&0.1654&0.2046&9.76&0.0648&0.2048&0.2148&6.56&0.0334&0.2294&0.2318&5.98\\
                            $\theta_7$ &-0.0004&0.0036&0.0036&6.00&-0.0004&0.0036&0.0037&6.34&-0.0003&0.0037&0.0037&6.06\\
                            $\theta_8$ &-0.0001&0.0007&0.0007&5.54&-0.0002&0.0007&0.0008&6.00&-0.0002&0.0007&0.0008&5.42\\
        \bottomrule
    \end{tabular}
\begin{tablenotes}
\item \scriptsize This table reports the simulated finite sample bias, standard deviation, RMSE, and size of the
MLE and the MERM estimators and the corresponding t-tests for the components of $\theta_0$. The true parameter
vector is $\theta_0 = (0.0355,   0.2976,   -2.0891,    0.0079,   -0.9900,    1.8794,   -0.0223,   -0.0149)'$. The
results are based on 5,000 replications.
\end{tablenotes}
\end{threeparttable}
\end{scriptsize}
\end{table}

\end{appendices}

\end{document}